%% file: neurips_2024.tex
\def\IsFull{}
\documentclass[10pt]{article}

\usepackage[margin=1in]{geometry}   

\include{macros/packages}

\include{macros/macro}
\include{macros/macros}

\include{local_macros}

\usepackage[compress,numbers]{natbib}
\usepackage[utf8]{inputenc} 
\usepackage[T1]{fontenc}    
\usepackage{hyperref}       
\usepackage{url}            
\usepackage{booktabs}       
\usepackage{amsfonts}       
\usepackage{nicefrac}       
\usepackage{microtype}      
\usepackage{xcolor}         
\usepackage{varwidth}
\usepackage{orcidlink}

\title{$\mathsf{OPA}$: One-shot Private Aggregation with Single Client Interaction and its Applications to Federated Learning\thanks{This is an extended version of the work accepted at CRYPTO 2025~\cite{C:KP25}.}}

\author{%
 \begin{tabular}{c c}
Harish Karthikeyan~\orcidlink{0000-0002-1787-4906} & Antigoni Polychroniadou~\orcidlink{0009-0003-0125-2971} \\
\texttt{harish.karthikeyan@jpmorgan.com} & \texttt{antigoni.polychroniadou@jpmorgan.com}
\end{tabular}
\\[5ex]
JPMorgan AI Research, JPMorgan AlgoCRYPT CoE
}
\date{}

\begin{document}

\maketitle
\input{abstract}
\input{neurips-main}

\input{neurips_2024.bbl}

\appendix
\input{neurips-app}


\end{document}

%% file: macros/packages.tex
\ifdefined\IsLLNCS
  \usepackage{epsfig, amsmath}
  \usepackage[hidelinks]{hyperref}
\hypersetup{
    colorlinks,
    linkcolor={red!50!black},
    citecolor={blue!50!black},
    urlcolor={blue!80!black}
}
\usepackage[table]{xcolor}
\else
  \usepackage{epsfig, amsmath, amsthm}
  \usepackage{thmtools}
  \usepackage{resizegather}
\fi
\ifdefined\IsFull{}
  \usepackage[hidelinks]{hyperref}
\hypersetup{
    colorlinks,
    linkcolor={red!50!black},
    citecolor={blue!50!black},
    urlcolor={blue!80!black}
}
\usepackage{booktabs}
\usepackage[table,x11names,dvipsnames,table]{xcolor}

\fi
\usepackage{enumitem}
\usepackage{thm-restate}
\usepackage{resizegather}
\usepackage{todonotes, float}
\usepackage{mathtools}
\usepackage{xparse,thm-restate}

\usepackage{tabularray}
\UseTblrLibrary{booktabs}
\usepackage{mathrsfs}
\usepackage{pgfplots}
\usepackage{subcaption}
\usepackage{tikz}
\usetikzlibrary{matrix,arrows}
\usetikzlibrary{decorations}
\usetikzlibrary{decorations.pathreplacing}
\usetikzlibrary{shapes.geometric, chains, calc, fit, matrix}
\tikzstyle{system}=[rectangle,draw,fill=lightgray,minimum height=0.8cm,minimum
            width=0.8cm,thick]
\tikzstyle{BC}=[system]
\tikzstyle{rzesource}=[system]
\tikzstyle{RO}=[resource, minimum width=1cm]
\tikzstyle{protocol}=[circle, inner sep=0.7mm, draw]
\tikzstyle{simulator}=[circle, inner sep=0.7mm, draw]
\tikzstyle{memory}=[resource]
\tikzstyle{distinguisher}=[resource,fill=white,minimum width=3.5cm,
minimum height=1.2cm]
\tikzstyle{link}=[]

\usepackage[final]{pdfpages}
\usepackage{adjustbox}
\usepackage{tabularx}

\usepackage{rotating}	
\usepackage{tikz}

\usetikzlibrary{positioning}
\usetikzlibrary{shapes}

\usepackage[utf8]{inputenc}
\usepackage{etoolbox}
\usepackage{pifont}
\usepackage{xcolor}

\usepackage{graphicx}

\usepackage{lscape}
\usepackage{algorithm}
\usepackage{algpseudocode}
\usepackage{float}
\usepackage{supertabular}
\usepackage{multirow}
\usepackage{longtable}
\usepackage{enumitem}
\usepackage{varwidth} 
\usepackage[most]{tcolorbox} 
\usetikzlibrary{calc,arrows,positioning,decorations.pathmorphing,shapes,backgrounds}
\usepackage{placeins}

\usepackage{booktabs}

\PassOptionsToPackage{hyphens}{url}
\ifdefined\IsPPAI
\else
\fi

\usepackage{makecell}

%% file: macros/macro.tex
\RequirePackage{algpseudocode}
\newcommand{\algoHead}[1]{\vspace{0.2em} \underline{\textbf{#1}} \vspace{0.3em}}

\makeatletter
\algnewcommand{\ExtendedState}[1]{\State
\parbox[t]{\dimexpr\linewidth-\ALG@thistlm}{\hangindent=\algorithmicindent\strut\hangafter=3#1\strut}}
\makeatother
\algnewcommand\algorithmicinput{\textbf{Input:}}
\algnewcommand\Input{\item[\algorithmicinput]}
\renewcommand{\algorithmicensure}{\textbf{Output:}}
\algrenewcommand\algorithmicrequire{\textbf{Input:}}
\algrenewcommand\algorithmicensure{\textbf{Output:}}
\algrenewcommand{\algorithmiccomment}[1]{{\color{blue}// #1}}
\algnewcommand{\IIf}[1]{\State\algorithmicif\ #1\ \algorithmicthen}
\algnewcommand{\EndIIf}{\unskip\ \algorithmicend\ \algorithmicif}

 \RequirePackage{varwidth}
 \RequirePackage{color}
 \RequirePackage[most]{tcolorbox}

 \newtcolorbox{titlebox}[5]{enhanced,halign=center,colframe=black,colback=white,boxrule={#3},arc={#2},auto outer arc,%
 pad at break*=5pt,vfill before first,before={\par\smallskip\noindent},after={\par\smallskip},top=12pt,left=4pt,%
 enlarge top by=7pt,
 fontupper=\small,
 title={\rule[-.3\baselineskip]{0pt}{\baselineskip}\sffamily\bfseries #1}, varwidth boxed
 title*=-30pt,
 attach boxed title to top left={yshift=-10pt,xshift=10pt}, coltitle=black,
 boxed title style={colback=white,boxrule={#5},arc={#4},auto outer arc}
 }

 \newtcolorbox{notitlebox}{enhanced,halign=center,colframe=black,colback=white,boxrule={0.5pt},arc={0.5pt},auto outer arc,
  pad at break*=5pt,vfill before first,after={\par\smallskip},left=4pt,
  fontupper=\small,
  attach boxed title to top left={yshift=-10pt,xshift=10pt}, coltitle=black
}

 \newenvironment{protocolbox}[1]
 {\begin{titlebox}{Protocol \normalfont #1}{0.5pt}{0.5pt}{1pt}{0.75pt}}
 {\end{titlebox}}

\usepackage[noend]{macros/algpseudocodex}

\makeatletter
\global\def\@linemarker{}
\algrenewcommand{\alglinenumber}[1]{%
	\ifdefempty{\@linemarker}{
		\rule{7pt}{0pt}
	}{
		\rlap{
			\footnotesize\textcolor{black}{\@linemarker}:
			\global\def\@linemarker{}
		}{\rule{7pt}{0pt}}
	}
}
\algnewcommand{\linemarker}[1]{\global\def\@linemarker{#1}}

\newcommand{\refmarker}[1]{line~\textcolor{black}{#1}}

\newcounter{markercounter}

\usepackage{algpseudocode}

%% file: macros/macros.tex
\newenvironment{code}{\begin{tabbing}
    12345\=12345\=12345\=12345\=12345\=12345\=12345\=12345\= \kill }
  {\end{tabbing}}

\newcommand{\ignore}[1]{}

\newcommand{\floor}[1]{\lfloor#1\rfloor}

\newcommand{\etal}{\textit{et al.}~}

\newcommand{\getsr}{{\:{\leftarrow{\hspace*{-3pt}\raisebox{.75pt}{$\scriptscriptstyle\$$}}}\:}}

{
  \begin{minipage}{#1\linewidth}
    \begin{algorithm}[H]  
}
{
    \end{algorithm}
  \end{minipage}
}

\usetikzlibrary{positioning}
\usetikzlibrary{calc}

\def \resw {2}
\def \resh {0.75}
\def \convrad {0.5}

\tikzstyle{border}=[rounded corners, dashed]

\tikzstyle{res}=[thick,violet,rounded corners,draw,minimum width=\resw cm,minimum height=\resh cm]
\tikzstyle{conv}=[thick,green!60!black,circle,draw,minimum width=2*\convrad cm]
\tikzstyle{conn}=[thick]

\DeclareMathOperator{\negl}{negl}

\newcommand{\dotcup}{\mmbox{\kern0.5em{\cdot\kern-0.5em\cup}\kern0.3em}}

\newcommand{\set}[1]{{\{\allowbreak #1\allowbreak \}}}

\ifdefined\IsLLNCS
  \newcommand{\sqed}{\qed}
\else
  \newcommand{\sqed}{}
\fi

\ifdefined\IsLLNCS
\renewenvironment{proof}[1][\proofname]{\par
  \normalfont\topsep6\p@\@plus6\p@\relax
  \trivlist
  \item[\hskip\labelsep\itshape
    #1\@addpunct{.}]\ignorespaces
}{%
  \nobreak\hfill \sqed\endtrivlist\@endpefalse
}

\else
\ifdefined\IsTPMPC
\else
\newtheorem{theorem}{Theorem}
\newtheorem{lemma}[theorem]{Lemma}

\newenvironment{claim*}[1]{\par\noindent\textbf{Claim:}\space#1}{}

\newtheorem{definition}{Definition}

\theoremstyle{definition}

\newtheorem{remark}{Remark}
\newtheorem{construction}{Construction}
\fi
\fi

\newcommand{\prf}{\term{PRF}}
\newcommand{\tprf}{\term{DPRF}}
\newcommand{\tpgen}{\term{Gen}}
\newcommand{\tpeval}{\term{Eval}}
\newcommand{\tpshare}{\term{Share}}
\newcommand{\tppeval}{\term{P\mhyphen Eval}}
\newcommand{\tpcombine}{\term{Combine}}
\newcommand{\tppp}{\pp_{\prf}}
\newcommand{\tpk}{\term{k}}

\newcommand{\pres}{\term{Setup}}

\newcommand{\presk}{r}
\newcommand{\prect}{\term{ct}}

\newcommand{\psae}{\term{Enc}}
\newcommand{\psakg}{\term{KeyGen}}
\newcommand{\psas}{\term{Setup}}
\newcommand{\psag}{\term{PublicKeyGen}}
\newcommand{\psaa}{\term{Aggregate}}

\newcommand{\psask}{\tpk}
\newcommand{\psapk}{\upkepk}
\newcommand{\psact}{\term{ct}}
\newcommand{\psaaux}{\term{aux}}
\newcommand{\psain}{x}
\newcommand{\csize}{\term{m}}
\newcommand{\cthr}{\term{t}}
\newcommand{\crec}{\term{r}}
\newcommand{\lab}{\term{\ell}}
\newcommand{\modulus}{\term{M}}
\newcommand{\AUX}{\term{AUX}}
\renewcommand{\set}[1]{\allowbreak\left\{\allowbreak #1\allowbreak \right\}}

\newcommand{\oEnc}{{\cO}_{\term{Enc}}}

\newcommand{\oCorr}{{\cO}_{\term{Corr}}}
\newcommand{\oEval}{\cO_\term{Eval}}

\mathchardef\mhyphen="2D

\newcommand{\upkepk}{\mathsf{pk}}

\newcommand{\pke}{\mathcal{E}}

\newcommand{\pkee}{E}

\renewcommand{\pkee}{\mathsf{Enc}}

\newcommand{\pkepk}{\mathsf{pk}}

\newcommand{\liss}{\term{LISS}}
\newcommand{\lshare}{\term{Share}}
\newcommand{\lcomb}{\term{Reconstruct}}
\newcommand{\lcoeff}{\term{GetCoeff}}

\newcommand{\ssize}{\term{\ell}_s}
\newcommand{\rsize}{\term{\ell}_r}
\newcommand{\secret}{\term{s}}
\newcommand{\offset}{\Delta}
\newcommand{\coeff}{r}

\newcommand{\hash}{\texttt{H}}

\newcommand{\share}[2]{{[#1]_{#2}}}

\newcommand{\kcorr}{\term{K}}
\newcommand{\khon}{\term{H}}

\newcommand{\CAPS}{One-shot Private Aggregation}
\newcommand{\caps}{\term{OPA}}

\newcommand{\cFor}{\mbox{\bf for }}

\newcommand{\cDo}{\mbox{\bf do }}

\usepackage[
	n,
	operators,
	advantage,
	sets,
	adversary,
	landau,
	probability,
	notions,	
	logic,
	ff,
	mm,
	primitives,
	events,
	complexity,
	oracles,
	asymptotics,
	keys]{cryptocode}

\newcommand{\gglbasic}{BIK+17}
\newcommand{\gglgroup}{BBG+20}
\newcommand{\flmgo}{Flamingo}

\renewcommand{\ss}{\term{SS}}
\renewcommand{\share}{\term{Share}}
\newcommand{\recons}{\term{Reconstruct}}

\makeatletter

\ifdefined\IsFull
\fi

%% file: local_macros.tex
\def\ddefloop#1{\ifx\ddefloop#1\else\ddef{#1}\expandafter\ddefloop\fi}

\def\ddef#1{\expandafter\def\csname bf#1\endcsname{{\mbox{\bf #1}}}}
\ddefloop ABCDEFGHIJKLMNOPQRSTUVWXYZabcdefghijklmnopqrstuvwxyz\ddefloop

\def\ddef#1{\expandafter\def\csname bf#1\endcsname{\ensuremath{\pmb{\csname #1\endcsname}}}}
\ddefloop {alpha}{beta}{gamma}{delta}{epsilon}{varepsilon}{zeta}{eta}{theta}{vartheta}{iota}{kappa}{lambda}{mu}{nu}{xi}{pi}{varpi}{rho}{varrho}{sigma}{varsigma}{tau}{upsilon}{phi}{varphi}{chi}{psi}{omega}{Gamma}{Delta}{Theta}{Lambda}{Xi}{Pi}{Sigma}{varSigma}{Upsilon}{Phi}{Psi}{Omega}{ell}\ddefloop

\def\ddef#1{\expandafter\def\csname bb#1\endcsname{\ensuremath{\mathbb{#1}}}}
\ddefloop ABCDEFGHIJKLMNOPQRSTUVWXYZ\ddefloop

\def\ddef#1{\expandafter\def\csname c#1\endcsname{\ensuremath{\mathcal{#1}}}}
\ddefloop ABCDEFGHIJKLMNOPQRSTUVWXYZ\ddefloop

\newcommand{\bpsact}{\boldsymbol{\psact}}
\newcommand{\bpsain}{\boldsymbol{\term{x}}}

\newcommand{\bpsaout}{\boldsymbol{\term{X}}}

\renewcommand{\prg}{\term{PRG}}
\newcommand{\prgg}{\term{Gen}}
\newcommand{\prge}{\term{Expand}}
\newcommand{\seed}{\term{sd}}

\ifdefined\IsLLNCS
\spnewtheorem{construction}{Construction}{\bfseries}{\itshape}
\fi

\newcommand{\eqq}{\stackrel{?}{=}}

\newcommand{\tSim}{\term{Sim}}
\newcommand{\bmask}{\boldsymbol{\term{mask}}}

  \newcommand{\ran}[2]{{r}}

\renewcommand{\floor}[1]{\left\lfloor #1\right\rfloor}
\newcommand{\Floor}[1]{\left\lfloor #1\right\rfloor}

\newcommand{\term}[1]{\mathsf{#1}}

\renewcommand{\pp}{\term{pp}}

\newcommand{\DL}{\term{DL}}
\newcommand{\Gen}{\term{Gen}}
\newcommand{\Solve}{\term{Solve}}

\newcommand{\Hybrid}{\mbox{\sf Hybrid}}

\renewcommand{\state}{st}
\newcommand{\Ts}{\tilde{\secret}}
\newcommand{\Tbs}{\tilde{\bfs}}
\newcommand{\LWR}{\texttt{LWR}}
\newcommand{\LWE}{\texttt{LWE}}
\newcommand{\DCR}{\texttt{DCR}}

\newcommand{\CL}{Castagnos and Laguillaumie}
\newcommand{\cl}{\term{CL}}
\newcommand{\Ghat}{\widehat{\bbG}}
\newcommand{\shat}{\widehat{s}}
\newcommand{\HSM}{\ensuremath{\texttt{HSM}_\modulus}}

\newcommand{\Prime}{p}
\newcommand{\Gp}{H}
\newcommand{\gp}{h}
\newcommand{\Ghp}{\Ghat^\Prime}
\newcommand{\baro}{\overline{\omega}}
\newcommand{\Dhat}{\widehat{\cD}}
\newcommand{\Dhp}{\Dhat_H}
\newcommand{\Dp}{\cD_H}
\newcommand{\sbar}{\bar{s}}

\renewcommand{\angle}[1]{\langle #1 \rangle}

\newcommand{\Adv}{\term{Adv}}
\newcommand{\sss}{\term{SS}}

\newcommand{\bsp}{\term{\mathbf{sp}}}
\renewcommand{\sp}{\term{sp}}
\newcommand{\lwr}{\term{LWR}}

\newenvironment{customlegend}[1][]{%
    \begingroup
    \csname pgfplots@init@cleared@structures\endcsname
    \pgfplotsset{#1}%
}{%
    \csname pgfplots@createlegend\endcsname
    \endgroup
}%

\def\addlegendimage{\csname pgfplots@addlegendimage\endcsname}

\newcommand{\Active}[1]{{\color{blue}\underline{#1}}}
 \newcommand{\hseed}{\term{dig}}

%% file: abstract.tex
\begin{abstract}


Our work minimizes interaction in secure computation, addressing the high cost of communication rounds, especially with many clients. We introduce One-shot Private Aggregation $\mathsf{OPA}$, enabling clients to communicate only once per aggregation evaluation in a single-server setting. This simplifies dropout management and dynamic participation, contrasting with multi-round protocols like Bonawitz et al. (CCS'17) (and subsequent works) and avoiding complex committee selection akin to YOSO. $\mathsf{OPA}$'s communication behavior \emph{closely} mimics learning-in-the-clear where each client party speaks only once. 

$\mathsf{OPA}$, built on LWR, LWE, class groups, and DCR, ensures single-round communication for all clients while also achieving sub-linear overhead in the number of clients, making it asymptotically efficient and practical. We achieve malicious security with abort and input validation to defend against poisoning attacks, which are particularly relevant in Federated Learning, where adversaries attempt to manipulate the gradients to degrade model performance or introduce biases.

We build two flavors of $\mathsf{OPA}$ (1) from (threshold) key homomorphic PRF and (2) from seed homomorphic PRG and secret sharing. 
The threshold Key homomorphic PRF addresses shortcomings observed in previous works that relied on DDH and LWR in the work of Boneh~\textit{et al.}(CRYPTO, 2013), marking it as an independent contribution to our work. Our other contributions include new constructions of (threshold) key-homomorphic PRFs and seed-homomorphic PRGs that are secure under the LWE, DCR Assumption, and other Class Groups of Unknown Order. 

\end{abstract}



%% file: neurips-main.tex
 \input{intro}

 \input{related}
 \input{tech-overview}
\input{crypto}

\input{caps-vector}
\input{caps-cons}
\input{key-reuse}

\input{malicious}
\subsection{Heterogeneity and Poisoning Attacks}
\label{sub:byzantine}
In Section~\ref{sec:byzantine}, we present two approaches, FedOpt~\cite{fedopt} and Byzantine-Robust Stochastic Aggregation~\cite{brsa}. The former is designed to handle heterogeneity in data distribution, while the latter is designed to
handle poisoning attacks where some clients behave arbitrarily. We further describe how to combine $\caps$ with these two approaches, bringing these techniques to the much-needed privacy-preserving alternatives.
\input{exp}

\input{conclusion}

  \paragraph{Disclaimer.} 
  This paper was prepared for informational purposes by the Artificial Intelligence Research group of JPMorgan Chase \& Co and its affiliates (“J.P. Morgan”) and is not a product of the Research Department of J.P. Morgan.  
  J.P. Morgan makes no representation and warranty whatsoever and disclaims all liability for the completeness, accuracy or reliability of the information contained herein.  
  This document is not intended as investment research or investment advice, or a recommendation, offer or solicitation for the purchase or sale of any security, financial instrument, financial product or service, or to be used in any way for evaluating the merits of participating in any transaction, and shall not constitute a solicitation under any jurisdiction or to any person, if such solicitation would be unlawful.

%% file: intro.tex
\section{Introduction}
\label{sec:intro}

Minimizing interaction in Multiparty Computation (MPC) is a highly sought-after objective in secure computation. This is primarily because each communication round is costly, and ensuring the liveness of participants, particularly in scenarios involving a large number of parties, poses significant challenges. Unlike throughput, latency is now primarily constrained by physical limitations, making it exceedingly difficult to reduce the time required for a communication round substantially. Furthermore, non-interactive primitives offer increased versatility and are better suited as foundational building blocks. However, any non-interactive protocol, which operates with a single communication round, becomes susceptible to a vulnerability referred to as the ``residual attack''~\cite{HaleviLP11} where the server can collude with some clients and evaluate the function on as many inputs as they wish revealing the inputs of the honest parties. 

We explore a natural ``hybrid'' model between the 2-round and 1-round settings. Our model allows for private aggregation, aided by an ephemeral committee of members, where the clients and committee members only speak once. This approach brings us closer to achieving non-interactive protocols while preserving traditional security guarantees. Our specific focus is within the domain of secure aggregation protocols, where a group of $n$ clients $P_i$ for $i\in [n]$ hold a private value $x_i$, wish to learn the sum $\sum_i x_i$ without leaking any information about the individual $x_i$. Furthermore, we support multiple sessions or iterations of the secure aggregation, where a different random set of clients is selected in each session, each with a different input. In this model, per aggregation session clients release encoded versions of their confidential inputs $x_i$ to a designated committee of ephemeral members and they go offline, they only speak once. Later, any subset of the ephemeral members can compute these encodings by transmitting a single public message to an unchanging, stateless evaluator or server. This message conveys solely the outcome of the secure aggregation and nothing else. The ephemeral members are stateless, speak only once, and can change (or not) per aggregation session. With that in mind, the committee members can be regarded as another subset of clients who abstain from contributing input when selected to serve on the committee during a current aggregation session. Each client/committee member communicates once per aggregation, eliminating the complexity of handling dropouts commonly encountered in multi-round secure aggregation protocols. The security guarantee ensures that adversaries corrupting some clients and committee members learn only the sum of honest clients' outputs, with no additional information. We provide a standard simulation-based proof against malicious adversaries with abort.

It is crucial to highlight the distinction between our single-server setting and multi-server protocols~\cite{DBLP:conf/eurocrypt/GilboaI14,prio,DBLP:conf/scn/AddankiGJOP22,10179468,EPRINT:ZhaZhoWan24}. In the multi-server model, clients securely distribute their inputs across a committee of multiple servers, which then engage in interactive protocols to achieve secure aggregation. However, these servers must be stateful, retaining data across aggregation iterations, and require significant computational resources.
In contrast, our approach leverages ephemeral committee members who are stateless and operate with lightweight computation, eliminating the need for persistent data storage or extensive computation A key design goal of our model is to ensure that committee members remain computationally light, making it feasible even for resource-constrained client devices to be committee members.

Our main application is Federated Learning (FL) in which a server trains a model using data from multiple clients. This process unfolds in iterations where a randomly chosen subset of clients (or a set of clients based on the history of their availability) receives the model's current weights. These clients update the model using their local data and return the updated weights. The server then computes the ``average'' of these weights (in the naive setting this is known as FedAvg~\cite{pmlr-v54-mcmahan17a} where the global model is simply the average of the client model weights), repeating this cycle until model convergence is achieved. Intuitively, this was supposed to guarantee the privacy of the client-held data as the server only sees the final weights. Unfortunately, prior works, such as~\cite{shokri2017membership}, have shown that the final weights can be successfully used to recover client-held data. This motivates the need for a secure aggregation tool. Unlike previous multi-round secure aggregation schemes with or without an offline setup~\cite{bonawitz2017practical,bell2020secure,EPRINT:GPSBB22,AC:LLPT23,SP:MWAPR23}, we introduce the first protocol that minimizes client involvement, ensuring that both clients and committee members can be computationally lightweight devices while speaking only once per aggregation iteration. Furthermore, our protocol does not require offline setup and guarantees completion, even if selected clients or committee members drop out and remain silent. 
\input{contributions}

%% file: contributions.tex
\section{Our Contributions} 
\label{sec:contr}
\label{sec:contributions}
We introduce $\caps$ designed to achieve maximal flexibility by granting parties the choice to speak once or remain silent, fostering dynamic participation from one aggregation to the next. We present the communication model in Figure~\ref{fig:comm-model}. This diverges from prior approaches~\cite{bonawitz2017practical,bell2020secure,AC:LLPT23,SP:MWAPR23,GPSBB22}, which necessitate multiple interaction rounds and the management of dropout parties to handle communication delays or lost connections in FL. At its core, every iteration of $\caps$ employs a random committee of size $\csize$ as intermediate helper parties. We build $\caps$ from new leakage-resilient, seed-homomorphic PRGs with simulatable leakage.\footnote{We can also build it from a (length-extended) key-homomorphic pseudorandom function (KHPRF) in the random oracle model with the same asymptotic performance. See \ifdefined\IsPRF{}the full version of this paper~\cite[\S C]{EPRINT:KarPol24} for details of this construction.\else Section~\ref{sec:khprf-cons} for the instantiations.\fi} $\caps$ relies on a mechanism such as PKI or authenticated channels. 

\noindent {\bf Cryptographic Assumptions:} We construct $\caps$ protocols providing a suite of six distinct versions based on a diverse spectrum of assumptions:
    \begin{itemize}
        \itemsep0em
        \item Learning With Rounding (\LWR) Assumption.
        \item Learning with Errors (\LWE) Assumption.
        \item Hidden Subgroup Membership ($\HSM$) assumption where $\modulus$ is a prime integer.
        \item $\HSM$ assumption where $\modulus=\Prime^k$ for some prime $\Prime$ and integer $k$.
        \item $\HSM$ assumption where $\modulus=2^k$.
        \item $\HSM$ assumption where $\modulus=\term{N}$ where $\term{N}$ is an RSA modulus (i.e., the \DCR\ assumption).
    \end{itemize}
$\caps$ does not require any trusted setup for keys, and for $\modulus$ being either a prime or an exponent of prime, or the \LWR~ assumption, we do not require any trusted setup of parameters either. The contributions are summarized in Figure~\ref{fig:contrib}.\\
\textbf{Threat Model:} We assume a static, malicious adversary that can corrupt up to an $\eta<N$ fraction of them, where $N$ is the number of clients in the universe of clients. For the committee of size $\csize$, the adversary can corrupt up to an $\eta_C$ fraction. Additionally, up to a $\delta$ fraction of input-providing clients and a $\delta_C$ fraction of committee members may drop out per iteration. For $\caps$ to function, we require $\delta_C + \eta_C < 1/3$ per iteration. In each iteration, a malicious adversary only learns the sum of the inputs of at least $(1-\delta-\eta)|\cC|$ where $\cC$ is the number of online clients in that iteration. We ensure security against a malicious adversary, achieving overall malicious security with abort, a feature not present in prior works such as ~\cite{bonawitz2017practical,bell2020secure,SP:MWAPR23,GPSBB22,AC:LLPT23}. Specifically, when all parties adhere to the protocol, the server successfully computes the sum of the online clients' contributions for a given iteration, provided a sufficient number of parties remain active during that iteration. Notably, we introduce a novel alternative approach instead of relying on signatures\footnote{See Remark~\ref{rem:signatures} for more details.}, as done in prior work to secure against a malicious server.
    \begin{figure*}[!tb]
    \centering
     \resizebox{0.95\textwidth}{!}{\begin{subfigure}[t]{0.6\textwidth}
\resizebox{6.76cm}{!}{\begin{tikzpicture}[x=0.75pt,y=0.75pt,yscale=-1,xscale=1,scale=0.8]
\draw (513,244) node  {\includegraphics[width=52.5pt,height=52.5pt]{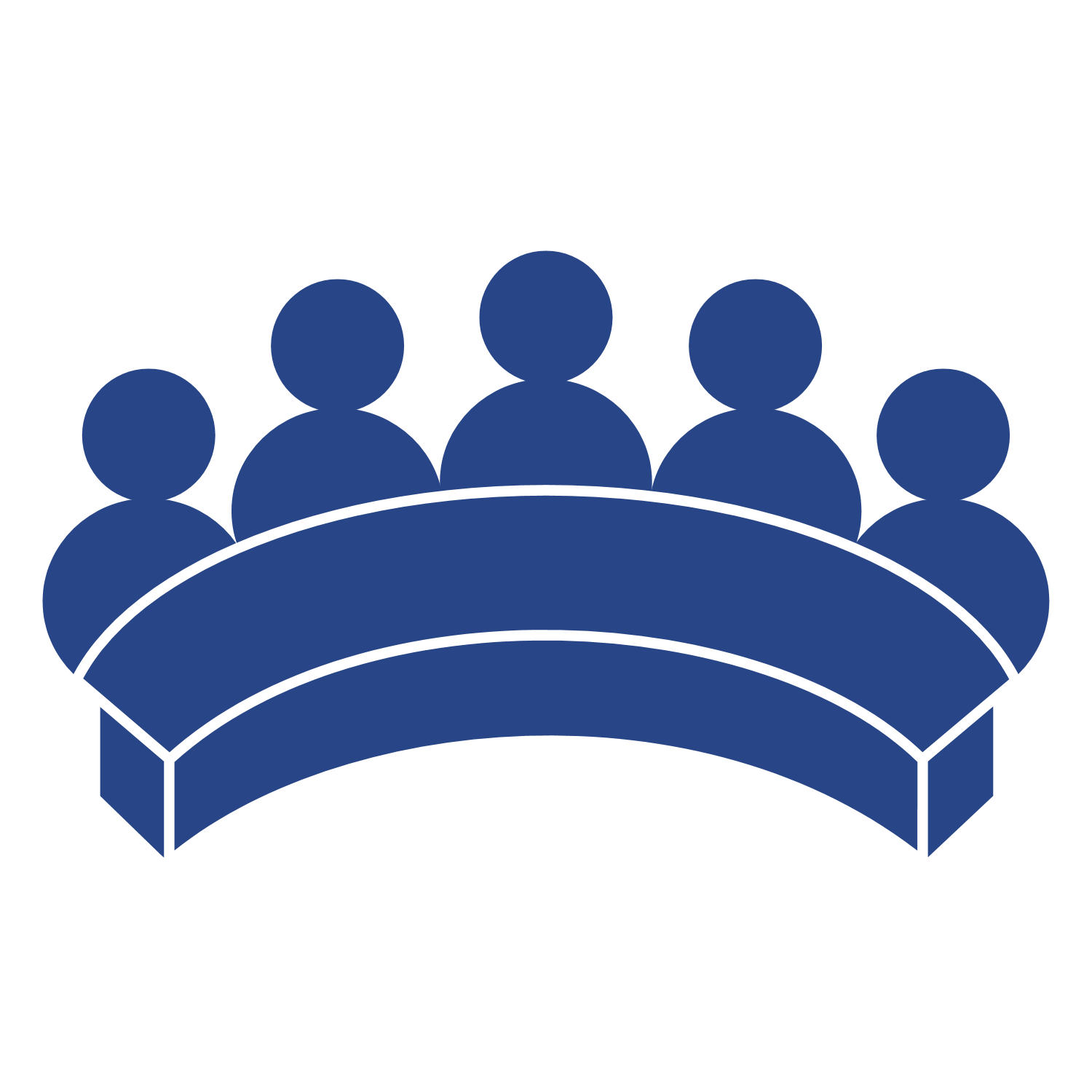}};
\draw (168,380) node  {\includegraphics[width=52.5pt,height=52.5pt]{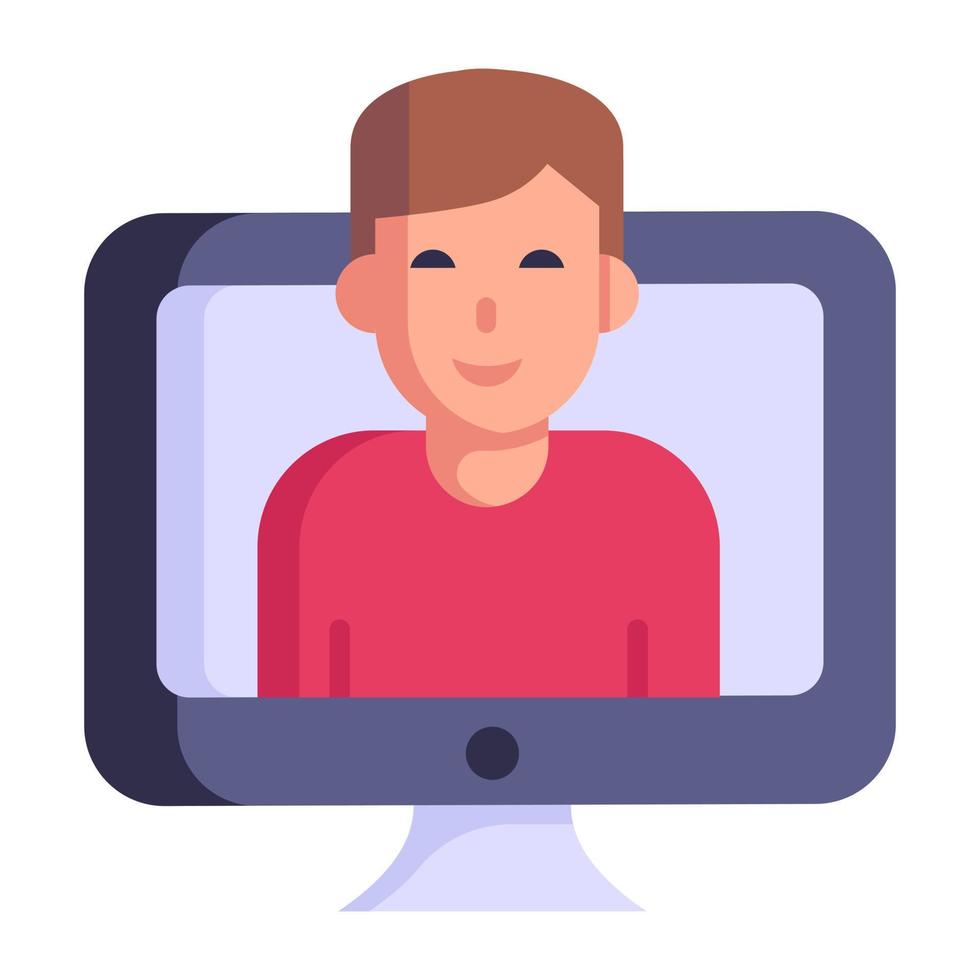}};
\draw (251,380) node  {\includegraphics[width=52.5pt,height=52.5pt]{images/client.jpg}};
\draw (329,380) node  {\includegraphics[width=52.5pt,height=52.5pt]{images/client.jpg}};
\draw (273,104) node  {\includegraphics[width=60pt,height=60pt]{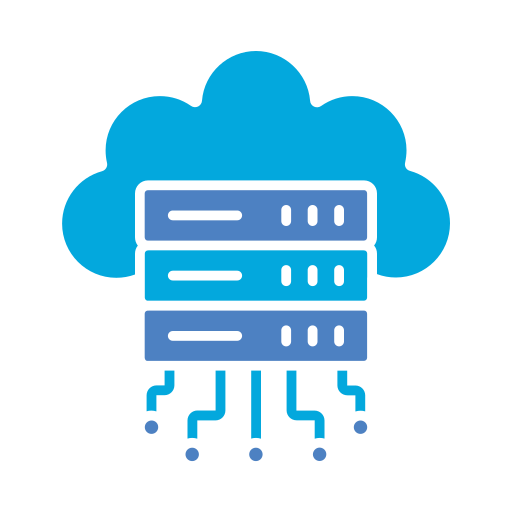}};
\draw    (209,161) -- (209,274) -- (209,347) ;
\draw [shift={(209,349)}, rotate = 269.24] [color={rgb, 255:red, 0; green, 0; blue, 0 }  ][line width=0.75]    (10.93,-3.29) .. controls (6.95,-1.4) and (3.31,-0.3) .. (0,0) .. controls (3.31,0.3) and (6.95,1.4) .. (10.93,3.29)   ;

\draw    (250,347) -- (250,161) ;
\draw [shift={(250,157)}, rotate = 90.29] [color={rgb, 255:red, 0; green, 0; blue, 0 }  ][line width=0.75]    (10.93,-3.29) .. controls (6.95,-1.4) and (3.31,-0.3) .. (0,0) .. controls (3.31,0.3) and (6.95,1.4) .. (10.93,3.29)   ;
\draw    (299,347) -- (299,161) ;
\draw [shift={(299,159)}, rotate = 90.29] [color={rgb, 255:red, 0; green, 0; blue, 0 }  ][line width=0.75]    (10.93,-3.29) .. controls (6.95,-1.4) and (3.31,-0.3) .. (0,0) .. controls (3.31,0.3) and (6.95,1.4) .. (10.93,3.29)   ;
\draw    (325.41,82.41) -- (472,229) ;
\draw [shift={(324,81)}, rotate = 45] [color={rgb, 255:red, 0; green, 0; blue, 0 }  ][line width=0.75]    (10.93,-3.29) .. controls (6.95,-1.4) and (3.31,-0.3) .. (0,0) .. controls (3.31,0.3) and (6.95,1.4) .. (10.93,3.29)   ;


\draw (229.97,163.55) node [anchor=north west][inner sep=0.75pt]  [rotate=-90.3] [align=left] {\small Message 0: Begin Iteration};

\draw (273.97,163.55) node [anchor=north west][inner sep=0.75pt]  [rotate=-90.3] [align=left] {\small Message 1a: Masked Input};
\draw (317.97,163.55) node [anchor=north west][inner sep=0.75pt]  [rotate=-90.3] [align=left] {\small Message 1b: $\pkee$(Aux Info)};
\draw (336.35,62.7) node [anchor=north west][inner sep=0.75pt]  [rotate=-45.54] [align=left] {\small Message 2: Combined Aux. Info};

\draw    (490,207) -- (343.41,60.41) ;
\draw [shift={(489,206)}, rotate = -135] [color={rgb, 255:red, 0; green, 0; blue, 0 }  ][line width=0.75]    (10.93,-3.29) .. controls (6.95,-1.4) and (3.31,-0.3) .. (0,0) .. controls (3.31,0.3) and (6.95,1.4) .. (10.93,3.29)   ;

\draw (354.35,40.7) node [anchor=north west][inner sep=0.75pt]  [rotate=-45.54] [align=left] {\small Forward: Aux Info to Committee};

\draw (513,280) node [anchor=north] [align=center] {Committee};
\draw (240,104) node [anchor=east] [align=right] {Server/Aggregator};
\draw (251,420) node [anchor=north] [align=center] {Clients};
\end{tikzpicture}}
    \caption{
    The $\caps$ system model operates in iterations. Each iteration begins with the server sending a message to initiate the process (Message 0). In response, clients train the model on their local data, obtain updates, and mask the input. (Message 1): masked input is sent to the server, while auxiliary information is transmitted (via encryption) to the committee via the server. Upon receiving the forwarded information, the committee members combine these into a single value. Finally, this consolidated data is sent to the server (Message 2), concluding the iteration. Importantly, the encryption information sent to the committee includes context information such as the current model and iteration. This prevents a malicious server from asking the committee for help with the same ciphertext. This thwarts residual attacks. 
    }
    \label{fig:comm-model}
    \end{subfigure}
    \hspace{1em}
    \begin{subfigure}[t]{0.4\textwidth}
        \definecolor{lightblue}{RGB}{173,216,230}
\definecolor{lightyellow}{RGB}{255,255,224}
\definecolor{orangecolor}{RGB}{255,165,0}
\definecolor{darkgreen}{RGB}{0,100,0}
\tikzset{
    block/.style={rectangle, draw, rounded corners, fill=lightblue, 
        text width=12em, text centered, minimum height=3em},
    line/.style={draw, -latex'},
    cloud/.style={draw, ellipse, fill=lightyellow, minimum height=2em},
}

\resizebox{6cm}{!}{\begin{tikzpicture}[node distance=2cm, auto]

    \node [block] (opa) {One-shot Private Aggregation};
    \node [block, below left=4cm and -1cm of opa] (prg) {(LR) Seed-Homomorphic PRG + Secret Sharing};
    \node [block, below right=4cm and -1cm of opa] (prf) {Threshold Key Homomorphic PRFs};
    
    \node [cloud, below left=2cm and -0.5cm of prg] (lwr1) {LWR};
    \node [cloud, below right=2cm and -0.5cm of prg] (lwe) {LWE};
    \node [cloud, below left=2cm and -0.5cm of prf] (cl) {HSM};
    \node [cloud, below right=2cm and -0.5cm of prf] (lwr2) {LWR};
    
    \path [line] (opa) -- node[sloped, above, text width=15em, align=center] {\textcolor{darkgreen}{$o(n)$} Committee Work} 
                        node[sloped, below, text width=8em, align=center, midway, text=darkgreen] {$\caps$} (prg);
    \path [line] (opa) -- node[sloped, above, text width=15em, align=center] {\textcolor{red}{$o(nL)$} Committee Work}
                        node[sloped, below, text width=5em, align=center, midway, text=red] {$\caps'$} (prf);
    \path [line] (prg) -- (lwr1);
    \path [line] (prg) -- (lwe);
    \path [line] (prf) -- (cl);
    \path [line] (prf) -- (lwr2);
    

\end{tikzpicture}}
    \caption{Our Cryptographic Assumptions. Here, $n$ is the number of clients, and $L$ is the length of the input vector. Finally, ``LR'' refers to leakage-resilient. Note that the HSM assumption is parametrized by $\modulus$ and covers the case where $\modulus$ is a prime integer, $\modulus=\Prime^k$ for some prime $\Prime$ and integer $k$, $\modulus=2^k$, and $\modulus=\term{N}$ where $\term{N}$ is an RSA modulus (i.e., the \DCR\ assumption). }
    \label{fig:contrib}
    \end{subfigure}}
    \caption{$\caps$ communication model and summarized contributions.}
    \label{fig:combined}
    \end{figure*}

\noindent    \textbf{Other contributions beyond Secure Aggregation :}
    \begin{itemize} 
        \item New Key Homomorphic PRF Constructions: We develop the first Key Homomorphic PRF based on the \HSM~ assumption and based on the \LWE\ assumption.
        \item New Threshold Key Homomorphic PRF Constructions: We extend the \HSM\ construction to a distributed key homomorphic setting using a modified Shamir's Secret Sharing scheme over integers. We also extend the almost Key Homomorphic PRF based on the \LWR\ assumption~\cite{C:BLMR13,PETS:ErnKoc21} using Shamir Secret Sharing over prime-order fields. In doing so, we fix gaps in the prior Distributed Key Homomorphic PRF based on \LWR, as proposed by Boneh~\etal~\cite{C:BLMR13}.
        \item Packed Secret Sharing over Integers: We also extend Shamir Secret Sharing over Integers to a packed version, which enables packing more secrets in one succinct representation. 
        \item New Seed Homomorphic PRG Construction: Of independent interest, we also build a seed homomorphic PRG from the \HSM\ assumption.\footnote{It is important to note that one cannot build $\caps$ from this seed-homomorphic PRG. The main reason is that the seed space is in a group of an unknown order. Therefore, one requires an integer secret-sharing scheme. Unfortunately, because the server reconstructs an integer value, one cannot prove that this construction is leakage resilient, where the leakage is defined as $\seed+r$. This is not indistinguishable from $\seed'$, where they are all integers.}
    \end{itemize} 
    \textbf{Applications in Federated Learning:} Our motivating application is Federated Learning (FL). The dynamic participation feature of $\caps$, crucial for federated learning, facilitates secure federated learning, where participants speak only once, streamlining the process significantly. In contrast, prior works\cite{bonawitz2017practical,bell2020secure} involve eight rounds, and the work of \cite{SP:MWAPR23} requires seven rounds, including the setup. Our advantages extend beyond just round complexity. See Section~\ref{sec:detailed_Contributions} and Tables~\ref{table:asym-perf-computation} and \ref{table:asym-perf-communication} for a detailed comparison of asymptotics. We introduce another variant, $\caps'$, derived from a threshold key-homomorphic PRF, eliminating simulatable leakage. In $\caps'$, the committee's workload scales with the length of the input vector $L$. Given its linear performance scaling with $L$, this remains a theoretical result aimed at achieving secure aggregation from the \HSM\ assumption and for small $L$.\ifdefined\IsPRF{} In the full version of this work~\cite[\S G]{EPRINT:KarPol24}, \else In Section~\ref{sec:stronger},\fi we show how to achieve security of honest user's inputs in $\caps'$ against a corrupt committee (but an honest server), for $\caps'$ instantiated with $\HSM$ assumption (called $\caps_\cl$).

    \noindent\textbf{Implementation and Benchmarks:} We implement $\caps$ as a secure aggregation protocol and benchmark with several state-of-the-art solutions~\cite{bonawitz2017practical,bell2020secure,GPSBB22,SP:MWAPR23}. We benchmark $\caps_\LWR$, based on the \LWR-based seed homomorphic PRG, offering competitive performance over prior works. Concretely, at 1000 clients, $\caps_\LWR$ has a server computation time of 0.31s, orders of magnitude smaller than other protocols. Similar orders of magnitude improvement are observed in client performance over other protocols' running time. 


To further demonstrate the feasibility of our protocol, we train a binary classification model using logistic regression in a federated manner for two datasets. Our protocol carefully handles floating point values (using two different methods of quantization - multiplying by a global multiplier vs representing floating point numbers as a vector of integer values), and the resulting model is shown to offer performance close to simply learning in the clear. 
        We also trained a neural network-based MLP classifier over popular machine learning datasets, including MNIST, CIFAR-10, and CIFAR-100. More details can be found in Section~\ref{sec:exp}. 

\noindent\textbf{Resilience against Data Heterogeneity and Malicious Clients.} While $\caps$ is a tool that achieves properties expected of standard secure aggregation, we show how to combine $\caps$ with existing solutions from the machine learning community to handle the heterogeneity of data distribution, i.e., when the clients; datasets are non-iid distributed and when there are poisoning attacks from the clients to their contributed inputs. This is detailed in \ifdefined\IsPRF{} the full version of our paper~\cite[\S F]{EPRINT:KarPol24}\else Section~\ref{sec:byzantine}\fi, where we are the first to work on secure aggregation to replace FedAvg with FedOPT~\cite{fedopt}, which performs better accuracy and convergence when data is non-iid. We also show how to combine $\caps$ with Byzantine Robust Stochastic Aggregation~\cite{brsa}, which is secure against poisoning attacks of clients. Meanwhile, we also describe how to use lattice-based zero-knowledge proofs to ensure that a client can prove their honest behavior. This is illustrated in Section~\ref{sub:mal-clients}.

\subsection{Detailed Contributions in Federated Learning}\label{sec:detailed_Contributions}
Next, we compare our protocol with all efficient summation protocols in terms of high-level features listed in Table~\ref{tab:psa}, focusing on those that accommodate dynamic participation, a key feature shared by all federated learning methodologies. Regarding performance 
in Table~\ref{table:asym-perf-computation}, we list the communication complexity, computational complexity, and round complexity per participant. Notably, our protocols are setup-free, eliminating any need for elaborate initialization procedures. Furthermore, they are characterized by a streamlined communication process, demanding just a single round of interaction from the participants. 
 \paragraph{Asymptotic Comparison.}
More concretely, based on Table~\ref{table:asym-perf-computation}, our approach stands out by significantly reducing the round complexity, ensuring that each participant's involvement is limited to a single communication round, i.e., each participant speaks only once. That is, users speak once and committee members speak once too. On the contrary, previous works\cite{bonawitz2017practical,bell2020secure}\footnote{\cite{bell2020secure} offer a weaker security definition from the other works: for some parameter $\alpha$ between $[0, 1]$, honest inputs are guaranteed to be aggregated at most once with at least $\alpha$ fraction of other inputs from honest users.} require $8$ rounds and the work of \cite{SP:MWAPR23} requires $7$ rounds in total, including the setup. This reduction in round complexity serves as a significant efficiency advantage. 

Despite our advantage in the round complexity, our advantages extend beyond just round complexity (see Table~\ref{table:asym-perf-computation}). Notably, our protocol excels in computational complexity as the number of participants ($n$) grows. While previous solutions exhibit complexities that are quadratic~\cite{bonawitz2017practical,SASH} or linearithmic~\cite{SP:MWAPR23} in $n$, our approach maintains a logarithmic complexity for the users, which is noteworthy when considering our protocol's concurrent reduction in the number of communication rounds. Furthermore, our committee framework demonstrates a sublinear relationship with $n$ for the committee members, a notable improvement compared to the linearithmic complexity and setup requirement in the case of~\cite{SP:MWAPR23} which considers a stateful set of decryptors (committee), as opposed to our stateless committee.\footnote{
Flamingo~\cite{SP:MWAPR23} employed \emph{decryptors}, a random subset of clients the server chose to interact with it to remove masks from masked data sent by the more extensive set of clients. LERNA\cite{AC:LLPT23} also requires a fixed, stateful committee (like Flamingo) to secret share client keys, whereas we support dynamic, stateless committees that can change in every round.} 

When it comes to user communication and message sizes, previous solutions entail user complexities that either scale linearly~\cite{bonawitz2017practical,SASH} or linearithmically~\cite{SP:MWAPR23} with the number of participants ($n$) according to Table~\ref{table:asym-perf-communication}. However, in our case, user communication complexity is reduced to a logarithmic level ($m\approx \log n$). Furthermore, as the number of users $n$ increases, the communication load placed on the server is also effectively reduced compared to other existing protocols. That said, the above advantages underline the scalability and efficiency of our protocols in the federated learning context, which typically requires a very large number of $n$ and $L$ where $L$ is the input length per client. It is important to stress that the number of committee members a client speaks to ($\csize$) purely depends on $\eta$, the fraction of malicious clients in the entire population of size $N$. We discuss sampling these clients in Section~\ref{sec:cons-caps}. Jumping ahead, we sample $\csize \log n$ clients for the committee and ensure that we assign clients to committee members, such that each client speaks with at most $\csize$ clients while each committee member receives communication from at most $\log n$ clients (via the server). This reduces the computation and the communication for the committee member to be $O(\log n)$.

The works of \cite{bonawitz2017practical,bell2020secure,SP:MWAPR23} address a malicious adversary that can provide false information regarding which users are online on behalf of the server by requiring users to verify the signatures of the claimed online set. This approach introduces an additional two rounds into each protocol, resulting in 10 rounds in \cite{bonawitz2017practical,bell2020secure} and 5 rounds in~\cite{SP:MWAPR23} with 10 rounds of setup. The setup communication complexity of ~\cite{SP:MWAPR23} also increases to $O(k \log^2 n)$. Prior work often required a two-thirds honest majority for the malicious setting. One can leverage this to introduce a gap between the reconstruction threshold, $\crec$, and the corruption threshold, $\cthr$. Specifically, if $\crec>(\csize+\cthr)/2$, then a malicious server can only reconstruct for a unique set, and any other information is purely random and unhelpful to the server. Furthermore, the usage of signatures typically adds a round of communication between the committee and the server. The combined technique of signatures and the reconstruction-corruption threshold gap ensures that the committee speaks once, while guaranteeing security against residual attacks. See Remark~\ref{rem:signatures} for more details.
This is the added communication and computation cost:
\begin{itemize}\itemsep0em
    \item Client: It sends an additional $O(\csize)$ elements and incurs a computation cost of $O(\csize)$ to perform one additional secret-sharing. 
    \item Committee: Each committee member receives an additional $o(n)$ elements from the clients, via the server. It decrypts and forwards them, as is, to the server. There is no additional computation cost. 
    \item Server: The server receives an additional $o(n)$ communication from each of the $\csize$ committee members and incurs a computation cost of $O(n\csize \allowbreak \log \csize)$.
\end{itemize}
To guarantee security against malicious clients, we propose zero-knowledge proofs based on the work of Lyubashevksy~\etal\cite{C:LyuNguPla22} and is described in Section~\ref{sub:mal-clients}. Critically, our proofs are designed such that the server performs the bulk of the verification with the committee only checking if the share it decrypts is a valid opening to the commitment for the underlying commit-and-prove proof system. 

\ifdefined\IsLLNCS
\begin{table*}[!tb]
\caption{
   Total asymptotic computation cost for all rounds per aggregation with semi-honest security.
    $n$ denotes the total number of users per iteration, with committee size $\csize$ and
    $L$ is the length of the input vector. 
    The ``Rounds'' column
    indicates the number of rounds in the setup phase (on the left, if applicable) and each aggregation iteration (on the right). A ``round of communication'' refers to a discrete step within a protocol during which a participant or a group of participants send messages to another participant or group of participants, and participants from the latter group must receive these messages before they can send their own messages in a subsequent round. ``fwd'' means that
    the server only forwards the messages from the users. The second column in the User Aggregation phase refers to the cost of the committee members.  
    }
\label{table:asym-perf-computation}
\small
    \centering

    \resizebox{\columnwidth}{!}{\begin{tblr}{
      colspec={|r|c|c|c|c|c|c|c|l|},
      row{1}={font=\bfseries},
      column{1}={font=\itshape},
      row{even}={bg=gray!10},
    }
\hline
\SetCell[r=3]{c} Protocol & \SetCell[c=3,r=2]{c} Rounds & & &\SetCell[c=5]{c}Computation Cost&&&&\\
\hline
&&&&\SetCell[c=2]{c}Server&&\SetCell[c=3]{c}User&&\\\hline
     & Setup & \SetCell[c=2]{c} Agg.&&  Setup  &  Agg. & Setup &\SetCell[c=2]{c} Agg.&    \\
\hline
\gglbasic\cite{bonawitz2017practical} & - & \SetCell[c=2]{c} 8& & - & $O(n^2 L)$ & - &\SetCell[c=2]{c} $O(n^2+nL)$& \\
\hline
\gglgroup\cite{bell2020secure}  & - & \SetCell[c=2]{c} 8& & - & $O(n \log^2 n +nL\log n)$ & - &\SetCell[c=2]{c} $O(\log^2 n + L\log n)$ & \\\hline
\flmgo\cite{SP:MWAPR23} & 4 & \SetCell[c=2]{c} 3& & fwd & $O(nL + n\log^2 n)$ & $O(\log^2 n)$ & $O(L + n\log n)$ & $O(\csize^2+ n)$\\\hline
LERNA~\cite{AC:LLPT23} & 1 & 1 & 1 & fwd & $O((\kappa+n)L+\kappa^2)$ & $O(\kappa^2)$ & $O(L)$ & $O(L+n)$\\\hline
SASH\cite{SASH} & - & \SetCell[c=2]{c} 10& & - & $O(L + n^2)$ & - &\SetCell[c=2]{c} $O(L + n^2)$ & \\\hline
$\caps$ & - & 1 & 1 & - & $O(nL+\csize \log \csize)$ & - & $O(L+\csize)$ & $O(\log n)$\\\hline
$\caps'$ & - & 1 & 1 & - & $O(nL+L(\csize\log\csize))$ & - & $O(L + \csize L)$ & $O(\log n\cdot L)$\\\hline
\end{tblr}}
\end{table*}

\begin{table*}[!tb]
   \caption{
    Total received and sent asymptotic communication cost for all rounds per aggregation with semi-honest security.
    $n$ denotes the total number of users, $\csize$ is the size of the committee,
    $k$ the security parameter, $L$ the length of the input vector, and $\ell$ the bit length of each element. 
    }
\label{table:asym-perf-communication}
\small
    \centering

    \resizebox{\columnwidth}{!}{\begin{tblr}{
      colspec={|r|c|c|c|c|l|},
      row{1}={font=\bfseries},
      column{1}={font=\itshape},
      row{even}={bg=gray!10},
    }
\hline
\SetCell[r=3]{c} Protocol & \SetCell[c=5]{c}Communication Cost&&&&\\
\hline
&\SetCell[c=2]{c}Server&&\SetCell[c=3]{c}User&&\\\hline
     &  Setup  &  Agg. & Setup &\SetCell[c=2]{c} Agg.&    \\
\hline
\gglbasic\cite{bonawitz2017practical} & - & $O(nL\ell+n^2k)$ & - &\SetCell[c=2]{c} $O(L\ell + nk)$& \\
\hline
\gglgroup\cite{bell2020secure} & - & $O(nL\ell + nk\log n)$&-&\SetCell[c=2]{c} $O(L\ell + k\log n)$ & \\\hline
\flmgo \cite{SP:MWAPR23} & $O(k\log n)$ & $O(nL\ell + nk\log^2 n)$& $O(k\log n)$ & $O(L\ell + nk\log n)$ & $O(\csize^2+n\log n)$\\\hline
LERNA~\cite{AC:LLPT23} & - & $O(Lnk+\csize^2\cdot k)$ & $O(\kappa^2k)$ & $O(\kappa L +L\log n+k)$ & $O(L\kappa+L\log n+k)$\\\hline
SASH \cite{SASH} & - & $O(nL\ell + \kappa^2 k)$ & - &\SetCell[c=2]{c} $O(L\ell + nk)$ & \\\hline
$\caps$ & -& $O(Ln+\csize)$ & - & $O(kL + \csize)$ & $O(\log n)$\\\hline
$\caps'$ & -& $O(kL(n+\csize))$ & - & $O(k(L + \csize L))$ & $O(k(\log n\cdot L+L))$\\\hline
\end{tblr}}

\end{table*}
\else
\begin{table*}[!tb]
\caption{
   Total asymptotic computation and communication costs for all rounds per aggregation with semi-honest security.
   $n$ denotes the total number of users per iteration, with committee size $\csize$ and
   $L$ is the length of the input vector. $k$ is the security parameter, and $\ell$ is the bit length of each element. The ``Rounds'' column
    indicates the number of rounds in the setup phase (on the left, if applicable) and each aggregation iteration (on the right). A ``round of communication'' refers to a discrete step within a protocol during which a participant or a group of participants send messages to another participant or group of participants, and participants from the latter group must receive these messages before they can send their own messages in a subsequent round. ``fwd'' means that
    the server only forwards the messages from the users. The second column in the User Aggregation phase refers to the cost of the committee members.
}
\label{table:combined-perf}
\label{table:asym-perf-computation}
\label{table:asym-perf-communication}
\small
    \centering
    \resizebox{\textwidth}{!}{\begin{tblr}{
      colspec={|r|c|c|c|c|c|c|c|c|c|c|c|c|c|},
      row{1}={font=\bfseries},
      column{1}={font=\itshape},
      row{even}={bg=gray!10},
    }
\hline
\SetCell[r=3]{c} Protocol & \SetCell[c=3,r=2]{c} Rounds & & &\SetCell[c=5]{c}Computation Cost&&&&&\SetCell[c=5]{c}Communication Cost&&&&\\
\hline
&&&&\SetCell[c=2]{c}Server&&\SetCell[c=3]{c}User&&&\SetCell[c=2]{c}Server&&\SetCell[c=3]{c}User&&\\\hline
     & Setup & \SetCell[c=2]{c} Agg.&&  Setup  &  Agg. & Setup &\SetCell[c=2]{c} Agg.& & Setup & Agg. & Setup &\SetCell[c=2]{c} Agg.&    \\
\hline
\gglbasic\cite{bonawitz2017practical} & - & \SetCell[c=2]{c} 8& & - & $O(n^2 L)$ & - &\SetCell[c=2]{c} $O(n^2+nL)$& & - & $O(nL\ell+n^2k)$ & - &\SetCell[c=2]{c} $O(L\ell + nk)$& \\
\hline
\gglgroup\cite{bell2020secure}  & - & \SetCell[c=2]{c} 8& & - & $O(n \log^2 n +nL\log n)$ & - &\SetCell[c=2]{c} $O(\log^2 n + L\log n)$ & & - & $O(nL\ell + nk\log n)$&-&\SetCell[c=2]{c} $O(L\ell + k\log n)$ & \\\hline
\flmgo\cite{SP:MWAPR23} & 4 & \SetCell[c=2]{c} 3& & fwd & $O(nL + n\log^2 n)$ & $O(\log^2 n)$ & $O(L + n\log n)$ & $O(\csize^2+ n)$ & $O(k\log n)$ & $O(nL\ell + nk\log^2 n)$& $O(k\log n)$ & $O(L\ell + nk\log n)$ & $O(\csize^2+n\log n)$\\\hline
LERNA~\cite{AC:LLPT23} & 1 & 1 & 1 & fwd & $O((\kappa+n)L+\kappa^2)$ & $O(\kappa^2)$ & $O(L)$ & $O(L+n)$ & - & $O(Lnk+\csize^2\cdot k)$ & $O(\kappa^2k)$ & $O(\kappa L +L\log n+k)$ & $O(L\kappa+L\log n+k)$\\\hline
SASH\cite{SASH} & - & \SetCell[c=2]{c} 10& & - & $O(L + n^2)$ & - &\SetCell[c=2]{c} $O(L + n^2)$ & & - & $O(nL\ell + \kappa^2 k)$ & - &\SetCell[c=2]{c} $O(L\ell + nk)$ & \\\hline
$\caps$ & - & 1 & 1 & - & $O(nL+\csize \log \csize)$ & - & $O(L+\csize)$ & $O(\log n)$ & -& $O(Ln+\csize\cdot n\log n)$ & - & $O(kL + \csize)$ & $O(\log n)$\\\hline
$\caps'$ & - & 1 & 1 & - & $O(nL+L(\csize\log\csize))$ & - & $O(L + \csize L)$ & $O(nL)$ & -& $O(kL(n+\csize))$ & - & $O(k(L + \csize L))$ & $O(k(nL+L))$\\\hline
\end{tblr}}
\end{table*}
\fi

\paragraph{Comparison with Willow~\cite{Willow}.} The concurrent work, Willow, employs a static, stateful committee of members, with non-committee members (clients) communicating only once, offering improvements over previous approaches like LERNA. While Willow supports client-to-committee communication that remains independent of committee size—compared to LERNA and $\caps$ —this advantage comes with trade-offs: 
\begin{enumerate}\itemsep0em
    \item Each committee members undergoes a two-round setup procedure. While Willow focuses on the single-iteration setting, it is unclear how Willow can be extended to multi-iteration setting without requiring a two-round setup procedure, at the beginning of each iteration. Specifically, to achieve constant server to committee communication, the committee members always provide the secret key corresponding to the public key encryption (PKE) scheme. This PKE is used to encrypt the second mask that is used for malicious privacy. Further, even in the semi-honest setting, the secret key of the committee members who have dropped out is revealed to the server.
    \item Each committee member also participates in two-rounds during the decryption phase: one for threshold decryption of the aggregated ciphertext and another to supply key shares for any missing members—similar to LERNA, where committee members also communicate multiple times
    \item They additionally require parties called verifiers who are tasked with validating the behavior of a possibly untrusted server, though these verifiers can be committee members. 
\end{enumerate}

In contrast, $\caps$ employs stateless committee members that can dynamically change across iterations. All client parties, including committee members, communicate only once per aggregation round, eliminating the need for a setup procedure. Furthermore, $\caps$ requires no additional verifier committee. Last but not least, unlike Willow, $\caps$ offers input validation and malicious client security with abort.

%% file: related.tex
\section{Related Work}\label{sec:cont}
\label{sec:relatedwork}
\ifdefined\IsFull
 First, we compare with other communication models that bear similarities with $\caps$ in Section~\ref{sec:other}. Second, we give an overview and compare $\caps$ with other tools employed for aggregation, with privacy, under various models in Section~\ref{sec:fl-other}, while summarizing our comparison in Table~\ref{tab:psa}.
 \input{other-related}
\else Due to space constraints, we defer an exposition on related work to the full version of the paper~\cite[\S 2.3, 2.4, 2.5]{EPRINT:KarPol24}
while summarizing our comparison of $\caps$ with other tools employed for aggregation, with privacy, under various models  in Table~\ref{tab:psa}.
\fi


\newcommand{\xmark}{{\color{red}\ding{55}}}
\definecolor{applegreen}{rgb}{0.01, 0.75, 0.24}
\newcommand{\cmark}{{\color{applegreen}\ding{51}}}
\newcommand{\TD}{{\color{red}\texttt{TD}}}
\newcommand{\DS}{{\color{orange}\texttt{DS}}}
\renewcommand{\PP}{{\color{applegreen}\texttt{NA}}}

\begin{table}[!tb]
  \caption{Comparison of Various Private Summation Protocols. \TD\ stands for trusted dealer/trusted setup, \DS\ stands for multi-round distributed setup. Note that \DS\ implies several rounds of interaction. Here, efficient aggregate recovery refers to whether the server can recover the aggregate efficiently. For example, \cite{NDSS:SCRCS11,BJL16,GPSBB22} require restrictions on input sizes to recover the aggregate due to the discrete logarithm computation.}
  \label{tab:psa}
  \centering
  \resizebox{0.8\columnwidth}{!}{\begin{tblr}{
      colspec={|l|l|l|l|l|l|},
      row{1}={font=\bfseries},
      column{1}={font=\itshape},
      row{even}={bg=gray!10},
    }
    \toprule
              & Efficient Aggregate & \begin{tabular}[c]{@{}l@{}}Dynamic \\ Participation\end{tabular} & \begin{tabular}[c]{@{}l@{}}\TD\\ vs\\ \DS\end{tabular} & Assumptions  \\
    \toprule
\cite{NDSS:SCRCS11}   & \xmark             & \xmark                                                               & \TD                                                                                                           & DDH                                                                  \\ \hline
\cite{FC:JoyLib13}       & \cmark             & \xmark                                                               & \TD                                                                                                       & DCR                                                                               \\ \hline
\cite{BJL16}    & \xmark/\cmark             & \xmark                                                               & \TD                                                                                                      & DDH/DCR                                                               \\ \hline
\cite{NDSS:BecGuaZim18}    & \cmark            & \xmark                                                               & \TD                                                                                                            & LWE/R-LWE                                                           \\ \hline
\cite{JC:TKGJ23}     & \cmark             & \xmark                                                               & \TD         & R-LWE                                                               \\ \hline
\cite{EPRINT:WMSA21}    & \cmark            & \xmark                                                               & \TD           & AES         \\ \hline
\cite{SC:TCKJ22}    & \cmark            & \xmark                                                               & \TD       & RLWE                                                                           \\ \hline
\cite{CANS:LeoElkMol14}   & \xmark             & \cmark                                                              &      \TD    & DCR                                                                           \\ \hline
\cite{PETS:ErnKoc21}    & \cmark            & \xmark                                                               & \TD                           & LWR                                                                               \\ \hline
\cite{BroGun23} & \xmark             & \xmark                                                               & \DS                                                                                       & LWR                                                                              \\ \hline
\cite{bonawitz2017practical,bell2020secure} & \cmark & \cmark & \DS$^\ast$ & DDH  \\\hline
\cite{GPSBB22} & \cmark & \cmark & \DS & DDH  \\\hline
\cite{SP:MWAPR23} & \cmark & \cmark & \DS & DDH  \\\hline
\cite{AC:LLPT23} & \cmark & \cmark & \DS & DDH  \\\hline
Our Work & \cmark & \cmark & \PP  & HSM, LWR, (R)LWE  \\
    \bottomrule
  \end{tblr}}
\end{table}

%% file: other-related.tex
\ifdefined\IsFull
\subsection{$\caps$ vs Other Communication/Computation Models}
\else 
\section{$\caps$ vs Other Communication/Computation Models}
\fi
\label{sec:other}
\paragraph{Shuffle Model.} Note that our model bears similarities to the shuffle model, in which clients dispatch input encodings to a shuffler or a committee of servers responsible for securely shuffling before the data reaches the server for aggregation, as in the recent work of Halevi et al.~\cite{HaleviIKR23}. Nonetheless, it is important to note that such protocols typically entail multiple rounds among the committee servers to facilitate an efficient and secure shuffle protocol.

\paragraph{Multi-Server Secure Aggregation Protocols.} It's worth emphasizing that multi-server protocols, as documented in~\cite{DBLP:conf/eurocrypt/GilboaI14,prio,DBLP:conf/scn/AddankiGJOP22,10179468,EPRINT:ZhaZhoWan24}, have progressed to a point where their potential standardization by the IETF, as mentioned in~\cite{patton-cfrg-vdaf-01}, is indeed noteworthy. In the multi-server scenario, parties can share their inputs securely among a set of servers, which then collaborate to achieve secure aggregation. Some of the works in this domain include two-server solutions Elsa~\cite{SP:RSWP23} and SuperFL~\cite{EPRINT:ZhaZhoWan24} or the generic multi-server solution Flag~\cite{ASIACCS:BSHV23}. Unfortunately, in the case of federated learning, which involves handling exceptionally long inputs, the secret-sharing approach becomes impractical due to the increased communication complexity associated with each input. Furthermore, these servers must have heavy computation power and be stateful (retaining data/state from iteration to iteration). In our protocol, the ephemeral parties are neither stateful nor require heavy computation. Finally, there must be non-collusion among a majority of the servers. Ensuring this constraint is logistically complicated to maintain.

\paragraph{Commmittee-Based MPC} Committee-based MPC is widely used for handling scenarios involving a large number of parties. However, it faces a security vulnerability known as adaptive security, where an adversary can corrupt parties after committee selection. The YOSO model, introduced by Gentry \etal~\cite{C:GHKMNR21}, proposes a model that offers adaptive security. In YOSO, committees speak once and are dynamically changed in each communication round, preventing adversaries from corrupting parties effectively. The key feature of YOSO is that the identity of the next committee is known only to its members, who communicate only once and then become obsolete to adversaries. YOSO runs generic secure computation calculations, and aggregation can be one of them. However, its efficiency is prohibitive for secure aggregation. In particular, the communication complexity of YOSO in the computational setting scales quadratic with the number of parties $n$ (or linear in $n$ if the cost is amortized over $n$ gates for large circuits). Additionally, an expensive role assignment protocol is applied to select the committees. Like LERNA, in YOSO, specific committee sizes need to be fulfilled to run a protocol execution. Lastly, our protocol does not rely on any secure role assignment protocol to choose the committees since even if all committee members are corrupted, privacy is still preserved. Fluid MPC~\cite{C:CGGJK21,C:BieEscPol23} also considers committee-based general secure computation. However, like YOSO, it is not practical. Unlike YOSO, it lacks support for adaptive security.



 Moreover, SCALES~\cite{AcharyaHKP22} considers ephemeral servers a la YOSO responsible for generic MPC computations where the clients only provide their inputs. This approach is of theoretical interest as it is based on heavy machinery such as garbling and oblivious transfer if they were to be considered for secure aggregation. Moreover, SCALES needs to make extra effort to hide the identities of the servers, which we do not require. 

 \ifdefined\IsFull
 \subsection{$\caps$ vs Other Secure Aggregation Algorithms}
 \else \section{$\caps$ vs Other Secure Aggregation Algorithms}
 \fi
 \label{sec:fl-other}

\paragraph{Multi-Round Private Summation.} We begin by revisiting the concept outlined in~\cite{HaleviLP11}. The first multi-round secure aggregation protocol, designed to enable a single server to learn the sum of inputs $x_1, \ldots, x_n$ while hiding each input $x_i$, is based on the following idea.  Each user $i$ adds a mask $r_i$ to their private input $x_i$. This mask remains hidden from both the server and all other users it exhibits the property of canceling out when combined with all the other masks, i.e., $\sum_{i \in [n]} r_i = 0$. Subsequently, each user forwards $X_i = x_i + r_i$ to the server. By aggregating all the $X_i$ values, the server can then determine the sum of all $x_i$. More specifically, to generate these masks, a common key $k_{ij} = k_{ji} = \mathrm{PRG}(g^{s_i s_j})$ is established by every pair of clients ${i, j}$. Here, $g^{s_i}$ is an ephemeral ``public key'' associated with each client $i \in [n]$. This public key is shared with all other clients during an initial round and facilitated through the server. Importantly, the value $s_i$ remains secret by each client $i$. Then, each client $i \in [n]$ computes the mask \(r_i = \sum_{j < i} k_{ij} - \sum_{j >i} k_{ij}\; \) and due to the cancellation property the server outputs $\sum_i X_i = \sum_i x_i$. In this protocol, users must engage in multiple rounds of communication, where each user communicates more than once. Moreover, the protocol prevents users from dropping out from an aggregation iteration.

\paragraph{Non-Interactive Private
Summation with trusted setup.} If we were to require users to communicate only once in a protocol iteration, we would encounter the challenge of mitigating residual attacks. In a prior study conducted by~\cite{NDSS:SCRCS11}, a solution based on DDH was proposed to mitigate residual attacks by involving a trusted setup that assumes the generation of the common keys $k_{ij}=k_{ji}$ into the protocol. However, it is essential to highlight that this setup lacked the necessary mechanisms to accommodate dropouts and facilitate dynamic participation for multiple aggregation iterations. 
Furthermore, to ensure that the server cannot recover the masking key given a client's masked inputs, the work relies on the DDH Assumption. An unfortunate consequence is that the server has to compute the discrete logarithm to recover the aggregate, a computationally expensive operation, particularly when dealing with large exponents. Numerous other works within this framework have emerged, each relying on distinct assumptions, effectively sidestepping the requirement for laborious discrete logarithm calculations. These include works based on the DCR assumption~\cite{FC:JoyLib13,BJL16}, and lattice-based cryptography~\cite{NDSS:BecGuaZim18,PETS:ErnKoc21,EPRINT:TKGJ20,SC:TCKJ22,EPRINT:WMSA21,SCN:OttKoc24}. 


\paragraph{Multi-Round Private
Summation without Trusted Setup.}
Another research direction focuses on removing the need for a trusted setup by developing multi-round decentralized reusable setups that generate masks while ensuring the essential cancellation property. However, akin to the previously mentioned approaches, these protocols come with a caveat—they do not accommodate scenarios involving dropouts or dynamic user participation across multiple iterations. Dipsauce~\cite{BroGun23} is the first to formally introduce a definition for a distributed setup PSA with a security model based on $k$-regular graph, non-interactive key exchange protocols, and a distributed randomness beacon~\cite{SP:ChoManBon23,drand,ACISP:RaiGli22} to build their distributed setup PSA. Meanwhile, the work of Nguyen~\etal~\cite{AC:NguPhaPoi23}, assuming a PKI (or a bulletin board where all the public keys are listed), computed the required Diffie-Hellman keys on the fly to then build a one-time decentralized sum protocol which allowed the server to sum up the inputs \emph{one-time}, with their construction relying on class group-based cryptography. To facilitate multiple iterations of such an aggregation, they combined their one-time decentralized sum protocol with Multi-client Functional Encryption (MCFE) to build a privacy-preservation summation protocol that can work over multiple rounds without requiring a trusted setup and merely requiring a PKI. Unfortunately, per iteration, the clients need to be informed of the set of users participating in that round, and unfortunately, they cannot drop out once they are chosen. 
\paragraph{Non-Interactive Private
Summation with a collector.} 
To circumvent the need for a trusted setup and multi-round decentralized arrangements, an approach is presented in the work of ~\cite{CANS:LeoElkMol14}, which introduces an additional server known as the ``collector''. The fundamental premise here is to ensure that the collector and the evaluation server do not collude, thus effectively mitigating the risks associated with residual attacks. This protocol does allow dynamic participation and dropouts per iteration. 

\paragraph{Multi-Round Private Summation with Dynamic Participation (aka Secure Aggregation).} 

Secure aggregation of users’ private data with the aid of a server has been well-studied in the context of federated learning. Given the iterative nature of federated learning, dynamic participation is crucial. It enables seamless integration of new parties and those chosen to participate in various learning iterations while also addressing the challenge of accommodating parties that may drop out during the aggregation phase due to communication failures or delays. Furthermore, an important problem in federated learning with user-constrained computation and wireless network resources is the computation and communication overhead, which wastes bandwidth, increases training time, and can even impact the model accuracy if many users drop out. Seminal contributions by Bonawitz \etal~\cite{bonawitz2017practical} and Bell \etal~\cite{bell2020secure} have successfully proposed secure aggregation protocols designed to cater to a large number of users while addressing the dropout challenge in a federated learning setting. However, it's important to note that these protocols have a notable drawback—substantial round complexity and overhead are incurred during each training iteration. Even in the extensive corpus of research based on more complex cryptographic machinery (see~\cite{DBLP:journals/ftml/KairouzMABBBBCC21} for a plethora of previous works) such as threshold additive homomorphic encryption etc., these persistent drawbacks continue to pose challenges. Notably, all follow up works~\cite{bell2020secure,AC:LLPT23,GPSBB22,SP:MWAPR23,ACORN,SASH,SP:LBVKH23} of \cite{bonawitz2017practical} require multiple rounds of interaction based on distributed setups. With their adaptable nature, secure aggregation protocols hold relevance across a wide array of domains. They are applicable in various scenarios, including ensuring the security of voting processes, safeguarding privacy in browser telemetry as illustrated in~\cite{Firefox}, and facilitating data analytics for digital contact tracing, as seen in~\cite{ENPA} besides enabling secure federated learning.

 It is also important to note that some of these works - ACORN~\cite{ACORN} and RoFL~\cite{SP:LBVKH23} build on top of the works of \cite{bonawitz2017practical,bell2020secure} to tackle the problem of ``input validation'' using efficient zero-knowledge proof systems. The goal is for the clients to prove that encrypted inputs are ``well-formed'' to prevent poisoning attacks.RoFL allows for detection when a malicious client misbehaves, while ACORN presents a \emph{non-constant} round protocol to identify and remove misbehaving clients. We leave it as a direction for future research on augmenting our protocol and supporting input validation. The above discussion is also summarized in Table~\ref{tab:psa}, by looking at four properties (a) whether the aggregate can be efficiently recovered, (b) whether it allows dynamic participation, (c) whether it requires trusted setup or multi-round distributed setup, and (d) the security assumptions. 
 
\paragraph{Comparison with LERNA~\cite{AC:LLPT23}.} LERNA requires a fixed, stateful committee (like Flamingo) to secret share client keys, whereas we support smaller, dynamic, stateless committees that can change in every round. Concretely, LERNA works by having \emph{each} client (in the entire universe of clients, not just for that iteration) secret-share the keys with the committee. Consequently, LERNA's committee needs to be much larger ($2^{14}$ members for $\kappa=40$ due to the number of shares they receive) and tolerate fewer dropouts than our approach. Furthermore, LERNA's benchmarks assume 20K+ clients, while real-world deployments have 50-5000 clients per iteration. When the client count is low, the committee has to do significantly more work to handle and store the required large number of shares. That said, LERNA is not suitable for traditional FL applications. In the table, we use the same notations, as LERNA refers to committee size by utilizing $\kappa$ in the committee calculations. LERNA could work for less than 16K parties, but the computation of committee members increases significantly as the number of parties decreases. Let us look at the concrete costs for the non-committee clients in both LERNA and $\caps_\LWR$. Assuming 20K clients and $L=50,000$, LERNA requires approximately 2GB of data followed by 0.91 MB per iteration.
Meanwhile, $\caps_\LWR$ requires each client to send 2.36 MB per iteration. As a result, LERNA only becomes cost-effective after more than 1300 iterations, requiring a fixed, stateful, and large committee to stay alive. The communication cost difference becomes starker for the committee clients between LERNA and $\caps_\LWR$ as LERNA requires every client, in the universe of clients, to communicate with the committee clients at setup; meanwhile, $\caps_\LWR$ only requires the chosen clients ($<<20K$) to speak in every iteration, and there is no setup. 

%% file: tech-overview.tex
\section{Technical Overview}
\label{sec:tech}

In this work, we focus on building a primitive, \CAPS\ ($\caps$), that enables privacy-preserving aggregation of multiple inputs across several aggregation iterations whereby a client only speaks once on his will per iteration.  

\paragraph{Seed-Homomorphic PRG (SHPRG)} \allowdisplaybreaks Recall that a pseudorandom generator PRG $\prg$ takes as input a random seed and outputs pseudorandom values. PRG $\prg:\cK\to\cY$ is seed-homomorphic if $\prg(s_1\oplus s_2)=\prg(s_1)\otimes \prg(s_2)$ where $(\oplus,\cK)$ and $(\otimes,\cY)$ are groups. While using the known \LWR-based construction~\cite{C:BLMR13}, we introduce the first \LWE-based SHPRG. Both are almost seed-homomorphic, with induced error handled via input encoding/decoding. For the \LWR\ construction (Construction~\ref{cons:shprg-lwr}): given random $\bfA\getsr\bbZ_q^{L\times n}$, $\allowdisplaybreaks \prg_{\LWR,\bfA}(\bfs)=\lfloor \bfA\bfs\rfloor_p$ where $p<q$, with error

\noindent $\allowdisplaybreaks \prg_{\LWR,\bfA}(\bfs_1+\bfs_2)=\prg_{\LWR,\bfA}(\bfs_1)+\prg_{\LWR,\bfA}(\bfs_2)+\bfe$ for $\bfe\in\set{0,1}^L$. Algorithm $\prg.\prge$ computes the PRG output on input seed $\seed$. It is important to note that the \LWE\ based SHPRG requires careful consideration of seed and error spaces and accounting for ensuing error in homomorphism while secret-sharing over appropriate seed spaces.  
We also provide the first construction of SHPRG based on \HSM\ assumption as an independent result, which does not have any error in the homomorphism.

\paragraph{Secret Sharing over Finite Fields.} 
In standard Shamir secret sharing~\cite{Shamir79}, a secret $s$ is shared via polynomial $f(X) = \sum_{i=0}^{\crec-1} a_i X^i$ where $a_0=s$ and $a_1,\ldots,a_{\crec-1}$ are random field elements. For $\csize$ parties, party $i \in [\csize]$ receives share $f(i)$, allowing any $\crec$ parties to reconstruct $s$ while any $\crec-1$ shares remain uniformly random. The corruption threshold can be lowered from $\crec-1$ to $\cthr$ for additional properties. Packed secret sharing~\cite{STOC:FraYun92} enables hiding multiple secrets in a single polynomial. A secret sharing scheme $\ss$ consists of sharing algorithm $\{\seed^{(j)}\}_{j\in[\csize]} \getsr \ss.\share\allowbreak(\allowbreak\seed,\cthr,\crec,\csize)$ taking secret $\seed$, thresholds $\cthr,\crec$, party count $\csize$ as input and outputting share $\seed^{(j)}$ for each party $j$, and reconstruction algorithm $\seed\getsr\ss.\recons(\{\seed^{(j)}\}_{j\in \cS})$ taking shares $\{\seed^{(j)}\}_{j\in \cS}$ as input and returning secret $\seed$ when $|\cS|\geq \crec$. Finally, Shamir's Secret Sharing is linearly homomorphic, which we leverage. Specifically, given two secrets $\seed_1,\seed_2$ with $j$-th share $\seed_1^{(j)}$ and $\seed_2^{(j)}$, then $\seed_1^{(j)}+\seed_2^{(j)}$ is a valid share of $\seed_1+\seed_2$. Note that since our corruption threshold $\cthr<\csize/3$, we can use reconstruction techniques with error correction to guarantee that a malicious share holder does not thwart the reconstruction.

\subsection{$\caps$ based on seed-homomorphic PRG}
\label{sub:caps-overview}

Every client needs to ensure the privacy of their input. Therefore, a client has to mask their input. In iteration $\lab$, if client $i$ has input $\bpsain_{i,\lab}$, it chooses a mask of the same length to ``add'' to the inputs. Let the mask be $\bmask_{i,\lab}$ and the ciphertext is defined as $\bpsact_{i,\lab}=\bpsain_{i,\lab}+\bmask_{i,\lab}$. To ensure privacy, we need the mask chosen uniformly randomly from a large distribution. Furthermore, by performing the addition with respect to a modulus $\modulus$, we get the property that for a random $\bmask_{i,\lab}$, $\bpsact_{i,\lab}$ is identically distributed to a random element from $\bbZ_\modulus$. The client $i$ sends $\bpsact_{i,\lab}$ to the server. This is Message 1a in Figure~\ref{fig:comm-model}. 

The server, upon receiving the ciphertexts, can add up the ciphertexts. This leaves it with $\sum_i \bpsact_{i,\lab}=\sum_i \bpsain_{i,\lab}+\sum_i \bmask_{i,\lab}$. The goal of the server is to recover $\sum_i \bpsain_{i,\lab}$. Therefore, it requires $\sum_i \bmask_{i,\lab}$ to complete the computation. In works on Private Stream Aggregation~\cite{NDSS:SCRCS11,BJL16,FC:JoyLib13,PETS:ErnKoc21}, the assumption made is that $\sum_i \bmask_{i,\lab}=0$. However, this requires all the clients to participate, which is a difficult requirement in federated learning. Instead, our approach follows the prior works in federated learning~\cite{bonawitz2017practical,bell2020secure,SP:MWAPR23,AC:LLPT23,GPSBB22,Willow}\footnote{\cite{bonawitz2017practical} had all the clients be the intermediate parties. In contrast, \cite{bell2020secure} created neighborhoods of clients where the intermediate parties were only the neighbors. \cite{SP:MWAPR23,AC:LLPT23,Willow} explicitly defined a committee of chosen clients, who may or may not be participating in that round.} to enlist the help of some intermediate parties to provide the server with $\sum_i \bmask_{i,\lab}$, for only the participating clients. However, most previous works \cite{bonawitz2017practical,bell2020secure} mandate that the masks must sum to zero, requiring mask reconstruction if a party drops out. This approach, often dependent on secret-sharing of the masks, results in multiple interaction rounds. Another recent line of works~\cite{GPSBB22,AC:LLPT23,SP:MWAPR23} necessitate an expensive setup for mask generation but still demands Shamir reconstruction by stateful committee members to handle dropout. Our approach is the first to break away from these paradigms, removing the need for zero-sum masks and setup, thereby significantly streamlining the process while working with ephemeral committee. In $\caps$, following Flamingo~\cite[Figure 1, Lines 2-7]{SP:MWAPR23}, $\csize$ committee members are ephemerally chosen using randomness beacon~\cite{drand}. More details can be found in Section~\ref{sub:caps-lwr}. 
The list of committee members can be included in the opening message from server to client, to begin the iteration. 

\paragraph{Working with the Committee- First Attempt.} Each client $i$, therefore, has to communicate information about its respective $\bmask_{i,\lab}$ to the committee members. We refer to this as ``Aux info'', and we route it through the server to the committee member. The information is encrypted under the public key of the respective committee member. This is Message 1b in Figure~\ref{fig:comm-model}. This ensures that the server cannot recover the auxiliary information. Eventually, each committee member ``combines'' the aux info it has received to the server (Message 2 in Figure~\ref{fig:comm-model}), with the guarantee that this is sufficient to reconstruct $\sum_i \bmask_{i,\lab}$. We rely on Shamir's secret-sharing~\cite{Shamir79} to distribute a secret $s$ to a committee of $\csize$ such that as long as $\crec$ number of them participate, the server can learn the secret $s$. The security of the secret is guaranteed even if the server colludes with $\cthr$ committee members. We require $\crec> (\csize+\cthr)/2$. We also rely on the homomorphism property, which ensures that if a committee member receives shares of two secrets $s_1,s_2$. Then, adding up these shares will help reconstruct the secret $s_1+s_2$. While prior works relied on committees, our contribution is in (a) ensuring that the committee's computation and communication cost is minimal and (b) ensuring that the client and committee speak once per iteration.

\paragraph{Optimizing Committee Performance} The solution laid out above requires the clients to ``secret-share'' $\bmask_{i,\lab}$. However, note that $\bmask_{i,\lab}$ is as long as $\bpsain_{i,\lab}$. Therefore, each committee member will receive communication $O(L)$ where $n$ is the number of clients, and $L$ is the length of the vector. It must also perform computations proportional to $O(nL)$. We can further reduce the complexity to $O(\log n\cdot L)$ by sampling a much larger committee such that each committee member receives communication from $\log n$ clients while each each client speaks with $\csize$ committee members. 
$\caps$ reduces the burden of the committee by introducing a succinct communication independent of $L$ to the committee. We rely on a structured pseudorandom generator (PRG) called a ``seed-homomorphic PRG''. A seed-homomorphic PRG ensures that $\prg(\seed_1+\seed_2)=\prg(\seed_1)+\prg(\seed_2)$. Now, each client $i$ secret-shares their respective seeds $\seed_{i,\lab}$ (and not $\bmask_{i,\lab})$, while setting $\bmask_{i,\lab}=\prg(\seed_{i,\lab})$. By the homomorphism of the secret-sharing scheme, the server can reconstruct $\sum_i \seed_{i,\lab}$, and then expand it as $\prg(\sum_i \seed_{i,\lab})$. By seed-homomorphism, this is also equal to $\sum_i \prg(\seed_{i,\lab})=\sum_i \bmask_{i,\lab}$. 

Three key challenges arise: (1) server learns $\sum_{i=1}^{n} \seed_{i,\lab}$, requiring leakage-resilient SHPRG (satisfied by our \LWR/\LWE\ constructions), (2) almost-homomorphic \LWR/\LWE\ PRGs need careful encoding/decoding, and (3) active security simulation requires valid ciphertext simulation for honest clients before finalizing dropout set. For (3), we resolve this in programmable random oracle model: each client adds a mask from hashing seed $\hseed_{i,\lab}$, shares $\{\hseed_{i,\lab}^{(j)}\}_{j\in[\csize]}\getsr\ss.\allowbreak\share(\hseed_{i,\lab},\cthr,\crec,\csize)$, which committee decrypts and forwards for server reconstruction. Notably, committee information remains vector-length independent. 

Though the underpinning idea of $\caps$ combines seed-homomorphic PRG and an appropriate secret-sharing scheme, there are technical issues with presenting a generic construction. We begin by presenting the construction based on the Learning with Rounding Assumption. We present constructions from both \LWR\ and \LWE\ assumptions in Section~\ref{sub:caps-lwr} and Section~\ref{sub:caps-lwe}, respectively.

\paragraph{Malicious Client Behavior and Heterogeneity.} In applications of secure aggregation to Federated Learning, it is important for the model updates to be tolerant to any heterogeneity in data distribution. In \ifdefined\IsPRF{}the full version of our paper~\cite[\S F.1]{EPRINT:KarPol24}\else Section~\ref{sub:fedopt}\fi, we show how to extend FedOpt~\cite{fedopt} to the setting involving secure aggregation. This two pronged approach of first securely aggregating the updates and then using an optimization technique to update the model weights, as a function of the aggregated updated, helps bring prior research on FedOpt to the secure aggregation setting. Additionally, we also detail how to combine $\caps$ with Byzantine-Robust Stochastic Aggregation (bRSA). The latter was conceived to handle both heterogeneity in data and client's malicious behavior when it came to providing inputs. Finally, we also present a protocol where a server can abort when malicious behavior is encountered. Specifically, when a client sends inconsistent information to the server and the committee members. We eschew the computation heavy Feldman's Verifiable Secret Sharing (VSS) and instead rely on SCRAPE~\cite{ACNS:CasDav17}, along with other zero-knowledge proof techniques based on lattices~\cite{C:LyuNguPla22}. Critically our proofs ensure that the bulk of the verification is done by the server with only \emph{constant} work done by a client. 
We refer readers to Section~\ref{sub:mal-clients} for more details. 
\subsection{$\caps'$ based on Threshold, Key-Homomorphic PRF}
We present an alternative approach to building $\caps'$ based on a threshold, key-homomorphic PRF. We present this alternative approach that is suitable for small $L$. 

{\paragraph{$\cl$ Framework.}  We use the generalized version of the framework, as presented by Bouvier~\etal\cite{JC:BCIL22}. Broadly, there exists a group $\Ghat$ whose order is $\modulus\cdot \shat$ where $\gcd(\modulus,\shat)=1$ and $\shat$ is unknown while $\modulus$ is a parameter of the scheme. Then, $\Ghat$ admits a cyclic group $\bbF$, generated by $f$ whose order is $\modulus$. Consider the cyclic subgroup $\bbH$, generated by $h=x^\modulus$, for a random $x\in\Ghat$. Then, one can consider the cyclic subgroup $\bbG$ generated by $g=f\cdot h$ with $\bbG$ factoring as $\bbF\cdot \bbH$. The order of $\bbG$ is also unknown. The $\HSM$ assumption states that an adversary cannot distinguish between an element in $\bbG$ and $\bbH$, while a discrete logarithm is easy in $\bbF$, $\sbar$, an upper bound for $\shat$ is provided as input. Note that for $\modulus=N$ where $N$ is an RSA modulus, the \HSM\ assumption reduces to the \DCR\ assumption. Therefore, the \HSM\ assumption can be viewed as a generalization of the \DCR\ assumption. }

\paragraph{Secret Sharing over Integers.} Braun~\etal~\cite{C:BraDamOrl23} was the first to identify how to suitably modify Shamir's secret sharing protocol to ensure that the operations can work over a group of unknown order, such as the ones we use on the $\cl$ framework. This stems from two reasons. The first is leakage in that a share $f(i)$ corresponding to some sharing polynomial $f$ always leaks information about the secret $\secret$, $\bmod i$ when the operation is over the set of integers. Meanwhile, the standard approach to reconstructing the polynomial requires the computation of the Lagrange coefficients, which involves dividing by an integer, which again needs to be ``reinterpreted'' to work over the set of integers. The solution to these problems is multiplying with an offset $\Delta=\csize!$ where $\csize$ is the total number of shares.

\paragraph{(Almost) Key Homomorphic Pseudorandom Functions} 
Naor~\etal~\cite{EC:NaoPinRei99} introduced the concept of key homomorphic PRFs (KH-PRFs), demonstrating that $H(x)^k$ is a secure KH-PRF under the DDH assumption in the Random Oracle Model, where $H$ is a hash function, $k$ is the key, and $x$ is the input. \cite{C:BLMR13} later constructed an almost KH-PRF under the Learning with Rounding assumption \cite{EC:BanPeiRos12}, which was formally proven secure by \cite{PETS:ErnKoc21}. Our work leverages almost KH-PRF constructions from both \LWR\ and \LWE\ assumptions in the Random Oracle Model. Additionally, we present a novel KH-PRF construction in the $\cl$ framework, yielding new constructions under the HSM assumption (including DCR-based constructions). We adapt the DDH-based construction to the $\cl$ framework and prove that $F(k,x)=\hash(x)^k$, where $k \getsr \cK$ and $\hash:\set{0,1}^* \to \bbH$ is modeled as a random oracle, is a secure KH-PRF under the HSM assumption. This adaptation requires careful consideration of appropriate groups, input spaces, and key spaces.

\paragraph{Distributed Key Homomorphic Pseudorandom Functions (KHPRF)} 
\cite{C:BLMR13} presented generic constructions of Distributed PRFs from any KH-PRF using secret sharing techniques. However, the $\cl$ framework's use of groups with unknown order necessitates working over integer spaces. While Linear Integer Secret Sharing \cite{PKC:DamTho06} exists, it can be computationally expensive. Instead, we utilize Shamir Secret Sharing over Integer Space as described by \cite{C:BraDamOrl23}, refining it with appropriate offsets to construct Distributed PRFs from our HSM-based construction. In other words, rather than reconstructing the secret, the solution reconstructs a deterministic function of the secret, which is accounted for in the PRF evaluation. This induces complications in reducing the security of our distributed key homomorphic PRF to that of the HSM-based KH-PRF. Finally, we demonstrate that our construction maintains the key homomorphism property, allowing combinations of partial evaluations of secret key shares to match the evaluation of the sum of keys at the same input point.


\paragraph{Constructions of \emph{almost} Distributed KHPRFs.}
\cite{C:BLMR13} introduced a generic construction of distributed KHPRFs from almost key homomorphic PRFs. They proposed $F_\LWR(\bfk,x):=\floor{\angle{\hash(x),\bfk}}_\Prime$ as an almost key homomorphic PRF, where $q>\Prime$ are primes, $\bfk\getsr\bbZ_q^{\rho}$, and $\hash$ is a suitably defined hash function. Their reduction to build a distributed PRF utilizes standard Shamir's Secret Sharing over fields, simplifying the process compared to integer secret sharing due to the prime nature of $q$ and $\Prime$. However, their proposed construction contained shortcomings affecting both correctness and security proofs. We will now provide a brief overview of these issues.
An almost KH-PRF satisfies $F(k_1+k_2,x)=F(k_1,x)+F(k_2,x)-e$ for some error $e$. For the \LWR~ construction, $e\in{0,1}$. However, this implies $F(T\cdot k,x)=T\cdot F(k,x)-e_T$ where $e_T\in{0,\ldots,T-1}$ for any integer $T$, leading to error growth and affecting Lagrange interpolation. The authors proposed multiplying by an offset $\offset=\csize!$ to bound the error and use rounding to mitigate its impact by ensuring that the error terms are ``eaten'' up. However, their security reduction works by reducing the security of the Distributed KHPRF (DKHPRF) to the underlying KHPRF. For a particular $i^\ast$, the adversary uses its oracle access to the KHPRF security game to obtain partial evaluations. In other words, the key of the KHPRF is implicitly set as the share of the key corresponding to some $i^\ast$. However, to use the KHPRF Challenger's response to construct the actual evaluation of the DKHPRF (via Lagrange interpolation), one has to rely on the ``clearing out the denominator'' technique by multiplying with some offset $\Delta$. Unfortunately, this induces additional errors when dealing with an \emph{almost} KH-PRF such as the one based on the \LWR\ assumption. 
%
%
Thus, their definition of partial evaluation function needs to be updated to be consistent with what is simulatable. We identified further issues in their rounding choices. Specifically, partial evaluations should be rounded down to $u$ where $\floor{p/u}>(\offset+1)\cdot \crec\cdot\offset$ ($\crec$ being the reconstruction threshold). Moreover, their framework only addressed single-key PRFs, not vector-key cases like $\LWR$-based construction. We address these issues in constructing a distributed, almost key homomorphic PRF based on LWR. We formally prove that
$F_\LWR(\bfk,x)=\floor{\Delta\floor{\Delta\floor{\angle{\hash(x),\bfk}_\Prime}_u}}_v$ for appropriate choices of $u$ and $v$ is a secure, distributed, almost key homomorphic, PRF.  

\paragraph{\CAPS\ without Leakage Simulation.} The construction of $\caps'$ is similar to the earlier ones based on seed-homomorphic PRG. There exist the following differences:
\begin{itemize}
    \item A client $i$ has to sample $L$ different keys. Each key $in$ is used to evaluate the PRF at a point $\lab$ and is used to mask the input $\psain_{i}^{(in)}$.
    \item It then secret shares each of the $L$ keys by running the $\tprf.\tpshare$, the algorithm to generate the shares of the $\tprf$ key. 
    \item Each share for committee member $j$ is evaluated at $\lab$. This evaluation is sent to the committee member. 
    \item Finally, the server runs $\tprf.\tpcombine$ to combine the information from the committee members to get the PRF evaluation at $\lab$ under the sum of the keys for each index $in$. $\tpcombine$ is the algorithm that helps reconstruct the evaluation from partial evaluations. 
\end{itemize}

\paragraph{Stronger Security Definitions and Construction.} We also present a stronger security guarantee, which was not provided by the committees of Lerna and Flamingo, whereby the committee members can all collude and observe all encrypted ciphertexts and all auxiliary information and cannot mount an IND-CPA-style attack. Unfortunately, our current construction, where the inputs are solely blinded by the PRF evaluation, which is also provided to the committee members in shares, can be unblinded by the committee leaking information about the inputs.  
Instead, we will have the server choose its key material $\psask_0$, and the corresponding key for that iteration communicated with the ``initiate transaction'' message. The client can then compute an ``ephemeral'' Diffie-Hellman key on the fly and use this to mask the input. Since the adversary does not receive $\psask_0$, it intuitively provides a sufficient mask for the inputs and the actual key $\psask_i$ of the honest client. Note that the adversary only receives the honest client's ``iteration'' public key. The actual construction is presented in \ifdefined\IsPRF{}the full version of our paper~\cite[\S G.1]{EPRINT:KarPol24}\else Section~\ref{sub:stronger}\fi.

%% file: crypto.tex
\section{Preliminaries and Cryptographic Building Blocks}
\label{sec:prelims}
\paragraph{Notations.} For a distribution $X$, we use $x\getsr X$ to denote that $x$ is a random sample
drawn from the distribution $X$. We denote by $\bfu$ a vector and by $\bfA$ a matrix. For a set $S$ we use $x\getsr S$ to denote
that $x$ is chosen uniformly at random from the set $S$. By $[n]$ for some integer $n$, we denote the set $\set{1,\ldots,n}$. For $x\in\bbR$, $\floor{x}$ (resp. $\ceil{x}$) refers to $y\in\bbZ$ such that $y\leq x<y+1$ (resp. $y-1<x\leq y$). 

\begin{definition}[Hypergeometric Distribution]
\label{def:hypergeo}
A Hypergeometric Distribution $\term{HyperGeom}(N,\eta,\csize)$ is a discrete probability distribution
that describes the probability of successes in $\csize$ draws, without
replacement from a finite population of size $N$ that contains $\eta$ fraction with that feature. We use the following tail bounds for $X \sim \term{HyperGeom}(N,\eta,\csize)$
$\forall d>0$, $\Pr[X\leq (\eta-d)\csize]\leq e^{-2d^2\cdot \csize}$ and  $\Pr[X\geq (\eta+d)\cdot \csize]\leq e^{-2d^2\cdot \csize}$  
\end{definition}
\subsection{Cryptographic Preliminaries}
We begin by discussing lattice-based assumptions in Section~\ref{sub:lattice}. Later, in Section~\ref{sec:shprg}, we introduce the syntax and security of a seed-homomorphic pseudorandom generator and present our lattice-based constructions. 

Due to space constraints, we defer \ifdefined\IsPRF{}a discussion of various constituent cryptographic primitives to the full version of our paper~\cite{EPRINT:KarPol24}. However, for completeness, we present the Packed Secret Sharing Construction based on Shamir's Secret Sharing here. 
\newcommand{\pos}{\term{pos}}
\begin{construction}[Packed Secret Sharing over $\bbF_q$]
\label{cons:pssf}
Consider the following $(\cthr,\crec,\csize)$ Secret Sharing Scheme where $\csize$ is the total number of parties, $\cthr$ is the corruption threshold, $\crec$ is the threshold for reconstruction. Further, let $\rho$ be the number of secrets being packed which are to be embedded at points $\pos_1,\ldots,\pos_\rho$ where $\pos_i=\csize+i$. Here, $\cthr:=\crec-\rho$. Then, we have the following scheme:

    \centering
 \begin{pchstack}[center,space=2\fboxsep]
\resizebox{0.7\columnwidth}{!}{\procedure{$\lshare(\bfs=(\secret_1,\ldots,\secret_\rho),\cthr,\crec,\csize)$}{
         (\coeff_0,\ldots,\coeff_{\crec-\rho-1})\getsr \bbF_q\\
         q(X):=\sum_{i=0}^{\crec-\rho-1} X^{i}\cdot \coeff_i\\
         \pos_i=\csize+i~\cFor i=1,\ldots,\rho\\
         \cFor i\in[\rho]~\cDo\\
        \pcind L_i(X):=\prod_{j\in[\rho]\setminus i}  \frac{X-\pos_j}{\pos_i-\pos_j}\cdot \Delta\\
         f(X):=q(X)\prod_{i=1}^{\rho} (X-\pos_i)+\sum_{i=1}^{\rho}{\secret}_i\cdot L_i(X)\\
          \pcreturn \set{\secret^{(i)}}_{i\in[\csize]}
        }
        
         \procedure{$\lcomb(\set{\secret^{(i)}}_{i\in\cS})$}{
            \pcif |\cS|< \crec ~\pcreturn \bot\\
            \text{Parse}~\cS:=\set{i_1,\ldots,i_{\cthr},\ldots}\\
            \cFor k\in [\rho]\\
            \pcind \cFor j\in[\cthr]\\
            \pcind\pcind \Lambda_{i_j}(X):=\prod_{\zeta\in[\cthr]\setminus{j}} \frac{i_\zeta-X}{i_\zeta-i_j}\\
             \pcind \pcind \secret_k':=\sum_{j\in[t]} \Lambda_{i_j}(m+k)\cdot \secret^{(j)}\\
             \pcind \pcreturn \bfs':=(\secret_1',\ldots,\secret_\rho')\\
            }
            }
        \end{pchstack}
\end{construction}
\else discussion on the syntax of secret-sharing schemes and constructions over fields and integers in Section~\ref{sec:secret}. This is followed by an exposition of various pseudorandom functions' syntax and security definitions through Sections~\ref{sub:prf} and \ref{sub:dkhprf}. Finally, we introduce the $\cl$ Framework in Section~\ref{sub:class}. 
\fi

\input{lattice-prelims}


\input{opa-shprg}


%% file: lattice-prelims.tex
\subsection{Lattice-Based Assumptions}
\label{sub:lattice}
This section will look at constructions based on three different lattice-based assumptions. 

\subsubsection{Learning with Rounding Assumption}
\label{sub:lwr}
We will begin by defining the learning with rounding (\LWR) assumption, which can be
viewed as a deterministic version of the learning with errors (\LWE) assumption~\cite{ACM:Regev09}. 
\LWR\ was introduced by Banerjee~\etal\cite{EC:BanPeiRos12}.  

\begin{definition}[Learning with Rounding]
\label{def:lwr}
    Let $\lambda,q,\Prime\getsr\lwr\Gen(1^{\rho})$ be functions of the security parameter $\rho$, with $L,\lambda,q,\Prime\in\bbN$ such that $q>\Prime$. Then, the Learning with Rounding assumption ($\LWR_{\lambda,L,q,p}$) states that for all PPT adversaries $\cA$, there exists 
    a negligible function $\negl$ such that:
\[
\Pr\left[b=b'~~
\begin{array}{|c}
      \bfs\getsr\bbZ_q^\lambda,\bfA\getsr\bbZ_q^{\lambda\times L}, \\
      \bfu_0:=\floor{\bfA^\top\cdot \bfs}_\Prime,
      \bfu_1\getsr\bbZ_p^L\\
      b\getsr\{0,1\},b'\getsr\cA(\bfA,\bfu_b)
\end{array}
\right]=\frac{1}{2} + \negl(\rho)
\]
where $\floor{x}_\Prime:=\lfloor x\cdot p/q\rfloor$. 
\end{definition}

\subsubsection{Learning with Errors Assumption}
\label{sub:cons-prflwe}
\begin{definition}[Learning with Errors Assumption (LWE)] 
\label{def:lwe} 
Consider integers $\lambda,L,~q\in\bbN$ that are functions of the security parameter $\rho$,
and a probability distribution $\chi$ on $\bbZ_q$, typically taken to be a normal
distribution that has been discretized. Then, the $\LWE_{\lambda,L,q,\chi}$ assumption states
that for all PPT adversaries $\cA$, there exists a negligible function $\negl$ such that: 
\[
\Pr\left[b=b'~~
\begin{array}{|c}
      \bfA\getsr \bbZ_q^{L \times \lambda},\bfx\getsr \bbZ_q^\lambda,\bfe\getsr \chi^L\\
      \bfy_0:=\bfA\bfx+\bfe\\
      \bfy_1\getsr\bbZ_q^L,
      b\getsr\{0,1\},b'\getsr\cA(\bfA,\bfy_b)
\end{array}
\right]=\frac{1}{2} + \negl(\rho)
\]
\end{definition}

\begin{definition}[Hint-LWE~\cite{ACISP:CKKLSS21,EC:DKLLMR23}]
\label{def:hint-lwe}
Consider integers $\lambda, L, q$
and a probability distribution $\chi$ on $\bbZ_q$, typically taken to be a normal
distribution that has been discretized. Then, the Hint-LWE assumption\footnote{Kim~\etal~\cite{C:KLSS23b} demonstrates that the Hint-LWE assumption is computationally equivalent to the standard LWE assumption. This assumption posits that $\mathbf{y}_0$ maintains its pseudorandom properties from an adversary's perspective, even when provided with certain randomized information about the secret and error vectors. Consider a secure LWE instance defined by parameters $(\lambda, m, q, \chi)$, where $\chi$ represents a discrete Gaussian distribution with standard deviation $\sigma$. The corresponding Hint-LWE instance, characterized by $(\lambda, m, q, \chi')$, where $\chi'$ denotes a discrete Gaussian distribution with standard deviation $\sigma'$, remains secure under the condition $\sigma'=\sigma/\sqrt{2}$. As a result, we can decompose any $\mathbf{e}\in\chi$ into the sum $\mathbf{e}_1+\mathbf{e}_2$, where both $\mathbf{e}_1$ and $\mathbf{e}_2$ are drawn from $\chi'$.} states
that for all PPT adversaries $\mathcal{A}$, there exists a negligible function $\mathsf{negl}$ such that: 
\[
\Pr\left[b=b'~~
\begin{array}{|c}
      \bfA\getsr \bbZ_q^{L\times \lambda},\bfk\getsr \bbZ_q^{\lambda},\bfe\getsr \chi'^L\\
      \bfr\getsr\bbZ_q^{\lambda},\bff\getsr\chi'^L\\
      \bfy_0:=\bfA \bfk+\bfe,
      \bfy_1\getsr\bbZ_q^L,
      b\getsr\{0,1\}\\
      b'\getsr\cA(\bfA,(\bfy_b,\bfk+\bfr,\bfe+\bff))
\end{array}
\right]=\frac{1}{2} + \mathsf{negl}(\rho)
\]
Where $\rho$ is the security parameter. 
\end{definition}

%% file: opa-shprg.tex
\subsection{Seed Homomorphic PRG}
\label{sec:shprg}
\label{sub:shprg-prf}
\label{sub:def-shprg}
\begin{definition}[Seed Homomorphic PRG (SHPRG)]
\label{def:shprg}
Consider an efficiently computable function $\prg:\cK\to\cY$, parametrized by $\prg=(\prg.\tpgen,\prg.\prge)$ where $(\cK,\oplus),(\cY,\otimes)$ are groups. Then $(\prg,\oplus,\otimes)$ is said to be a secure seed homomorphic pseudorandom generator (SHPRG) if: 
\begin{itemize}
    \item $\prg$ is a secure pseudorandom generator (PRG), i.e., for all PPT adversaries $\cA$, there exists a negligible function $\negl$ such that: \begin{gather*}
	\Pr\left[b=b'~~
	\begin{array}{|c}
     \pp_\prg\getsr\prg.\prgg,
	 b\getsr\set{0,1},\seed\getsr\cK\\
	 Y_0=\prg.\prge(\pp_\prg,\seed), Y_1\getsr cY \\
	 b'\getsr\cA(Y_b)
	\end{array}
	\right]\leq \frac{1}{2}+\negl(\kappa)
    \end{gather*}
    \item For every $\seed_1,\seed_2\in \cK$, we have that $\prg.\prge(\seed_1)\otimes \prg.\prge(\seed_2)=\prg.\prge(\seed_1\oplus \seed_2)$
\end{itemize}
\end{definition}
We abuse notation and drop $\pp_\prg$ from the input of $\prge$. 
\section{Constructions of SHPRG}
\label{sec:cons-shprg}
In this section, we first recap the construction of SHPRG based on \LWR\ assumption as presented by Boneh~\etal~\cite{C:BLMR13}. We then show that it is leakage-resilient. Later, we present the \emph{first} construction based on \LWE\ assumption and then prove that it is indeed leakage resilient under Hint-\LWE\ assumption.
\subsection{Construction from \LWR\ Assumption}
\begin{construction}[SHPRG from \LWR\ Assumption]
\label{cons:shprg-lwr}
    Let $\pp_\prg:=\bfA\getsr\bbZ_q^{\lambda\times L}$ be the output of $\prg.\prgg$, then consider the following: $\prg_\LWR:\bbZ_q^{\lambda}\to \bbZ_p^{L}$ where $L>\lambda$ is defined as $\prg.\prge(\seed=\bfs)=\lfloor \bfA^\top \cdot \bfs\rfloor_p$ where $q>p$ with $\lceil x \rfloor_p=\lfloor x\cdot p/q\rfloor$ for $x\in\bbZ_q$.  
\end{construction}
It is easy to see that Construction~\ref{cons:shprg-lwr} is a secure PRG under the \LWR\ assumption, which gives us the following theorem
\begin{theorem}[PRG Security of Construction~\ref{cons:shprg-lwr}]
    \label{thm:shprg-lwr-prg}
    If $\LWR_{\lambda,L,q,p}$ assumption holds, then $\prg_\LWR$ is a secure PRG. 
\end{theorem}
This is almost seed homomorphic in that: $\prg.\prge(\bfs_1+\bfs_2)=\prg.\prge(\bfs_1)+\prg.\prge(\bfs_2)+\bfe$ where $\bfe\in \set{0,1}^{L}$.


\begin{restatable}[Leakage Resilience of Construction~\ref{cons:shprg-lwr}]{theorem}{lrprglwr}
\label{thm:lr-prg-lwr}
    Let $\prg_\LWR$ be the PRG defined in Construction~\ref{cons:shprg-lwr}. Then, it is leakage resilient in the following sense:
    \begin{gather*}
    \begin{array}{cc}
        \left\{ 
        \begin{array}{cc}
          \prg.\prge(\seed)\bmod p   &   \seed\getsr\bbZ_q^{\lambda}\\
            \seed+\seed'\bmod q & \seed'\getsr\bbZ_q^{\lambda}
        \end{array}
        \right\} & \approx_c  \left\{ 
        \begin{array}{cc}
          \bfy   &   \seed,\seed'\getsr\bbZ_q^{\lambda}         \\   \seed+\seed'\bmod q & \bfy\getsr\bbZ_p^L
        \end{array}
        \right\} 
    \end{array}
    \end{gather*}
\end{restatable}
\ifdefined\IsPRF
Due to space constraints, we refer the readers to the full version of this paper for the proof~\cite[\S 5.1]{EPRINT:KarPol24}.
\else
\begin{proof}
The proof proceeds through a sequence of hybrids. 
\begin{gamedescription}[name=Hybrid,nr=-1]
    \describegame The left distribution is provided to the adversary. In other words, the adversary gets: 
    \[
   \set{\prg_\LWR(\seed)\bmod p, \seed+\seed'\bmod q:\seed,\seed'\getsr\bbZ_q^{\lambda}}
    \]
    \describegame In this hybrid, we replace $\seed+\seed'\bmod q$ with a uniformly random value $\seed''\getsr \bbZ_q^{\lambda}$.
    \[
    \set{\prg_\LWR(\seed)\bmod p,{\color{blue}\seed''}:\seed,\seed''\getsr\bbZ_q^{\lambda}}
    \]
    Note that $(\seed+\seed')\bmod q$ and $\seed''$ are identically distributed. Let us assume that there exists a leakage function oracle $L$ that can be queried with an input $\seed$, for which it either outputs $\seed+\seed'\bmod q$ for a randomly sampled $\seed'\getsr\bbZ_q^{\lambda}$ or outputs $\bfs'\getsr\bbZ_q^{\lambda}$. Then, an adversary $\cA$ that can distinguish between $\Hybrid_0$ ad $\Hybrid_1$ can be used by an adversary $\cB$ to break the leakage function as follows: $\cB$ samples $\seed$. It sends this value to the leakage function. Meanwhile, it also computes $\prg_\LWR(\seed)$ and appends the outputs of the leakage function. Note that if the leakage output was $\seed+\seed'$, the $\Hybrid_0$ is simulated by $\cB$, and if it was $\seed''$ then $\Hybrid_1$ is simulated by $\cB$. However, the output of the leakage function is an identical distribution. Therefore, $\Hybrid_0,\Hybrid_1$ are also identically distributed.
    \describegame In this hybrid, we will replace the PRG computation with a random value from the range. 
    \[
    \set{{\color{blue}\bfy},{\seed''}:\bfy\getsr\bbZ_p^L,\seed''\getsr\bbZ_q^{\lambda}}
    \]
    Under the security of the PRG, we get that $\Hybrid_1,\Hybrid_2$ are computationally indistinguishable. 
    \describegame We replace $\seed''$ with $(\seed+\seed')\bmod q$. 
    \[
    \set{\bfy,{\color{blue}(\seed+\seed')\bmod q}:\bfy\getsr\bbZ_p^L,\seed,\seed''\getsr\bbZ_q^{\lambda}}
    \]
    As argued before $\Hybrid_2,\Hybrid_3$ are identically distributed. 
\end{gamedescription}
Note that $\Hybrid_3$ is the right distribution from the theorem statement. This completes the proof. 
\end{proof}
\fi
\subsection{Construction from \LWE\ Assumption}
\begin{construction}[SHPRG from \LWE\ Assumption]
\label{cons:shprg-lwe}
    Let $\bfA\getsr \mathbb{Z}_q^{L\times\lambda}$ be the output of $\prg.\prgg$. Then, consider the following seed homomorphic PRG $\prg_\LWE:\bbZ_q^\lambda\times \chi^L \to \mathbb{Z}_q^L$ is defined as:$\prg.\prge((\bfs,\bfe))=\bfA \bfs+\bfe$
\end{construction}
The proof of security directly applies the \LWE\ Assumption. However, we now prove the (almost) seed homomorphic property. 
\begin{align*}
    &\prg.\prge(\bfs_1,\bfe_1)=\bfA \bfs_1+\bfe_1;
    \prg.\prge(\bfs_2,\bfe_2)=\bfA \bfs_2+\bfe_2\\
    & \prg.\prge(\bfs_1,\bfe_1)+\prg.\prge(\bfs_2,\bfe_2)=\bfA (\bfs_1+s_2)+(\bfe_1+\bfe_2)\\
    &\hspace{5.3cm} =\prg.\prge(\bfs_1+s_2)+\bfe
\end{align*}
where $\bfe\leq \bfe_1+\bfe_2+1$. 

Looking ahead, when we use this PRG to mask the inputs, we will do the following: $\prg.\prge(s_i,e_i)=\bfA \bfs_i+\bfe_i+\lfloor q/p \rfloor \cdot \bfx_i$ where $\bfx_i\in \bbZ_p^m$. Upon adding $n$ such ciphertexts $(\bmod q)$, we get:
\[
\bfA  \sum_{i=1}^{n}{\bfs_i} +\bfe + \lfloor q/p \rfloor \cdot \sum_{i=1}^{n}\bfx_i
\]
where $\bfe\leq 1+\sum_{i=1}^{n} \bfe_i$. Therefore, to eventually recover $\sum_{i=1}^{n} \bfx_i \bmod p$ from just the value of  $\sum_{i=1}^{n}{\bfs_i}$, we will require $||\bfe||_{\infty}< \frac{q}{2p}$. 
Looking ahead, we will rely on the Hint-LWE Assumption~\ref{def:hint-lwe} to show that it is leakage resilient when we use it to build our aggregation tool. 

\newcommand{\RLWE}{\texttt{R-LWE}}
\begin{remark}[Construction based on Ring-LWE]
    The above LWE construction can be extended to the Ring-LWE~\cite{EC:LyuPeiReg10} setting. 
    
    Recall the R-LWE Assumption. Let $N$ be a power of two and $m>0$ be an integer. Let $R$ be a cyclotomic ring of degree $N$, and let $R_{q}$ be its residue ring modulo $q>0$. Then, 
    the following holds: 
\begin{gather*}
\left\{(\mathbf{a}, \mathbf{a} \cdot k+\mathbf{e}): \mathbf{a} \getsr R_{q}^{m}, k\getsr R_q, \mathbf{e}\getsr \chi^{m}\right\} \approx_c
\left\{(\mathbf{a}, \mathbf{u}): \mathbf{a} \getsr R_{q}^{m}, \mathbf{u} \getsr R_{q}^{m}\right\}
\end{gather*}
This gives us the following construction:
$\prg_\RLWE((k,e)):\bfa k+\bfe$
\end{remark}

%% file: caps-vector.tex
\section{One-shot Private Aggregation}
\label{sec:caps}
In this section, we construct One-shot Private Aggregation ($\caps$). Broadly speaking, the goal of this primitive is to support a server (aka aggregator) to sum up the clients' inputs encrypted to a particular label (which can be a tag, timestamp, etc.), without it learning any information about the inputs beyond just the sum. 
\input{simulation-def}

%% file: simulation-def.tex
\subsection{Simulation-Based Privacy}
\label{sec:sim-def}
Our proof approach is based on the standard simulation-based framework~\cite{Goldreich,Lindell} where we demonstrate that any attacker against our protocol can be simulated by an attacker $\tSim$ in an ideal world where a trusted party $\cT$ computes a function $F$ on the clients' inputs $X$. In our case, this function is that of vector summation. We consider an attacker $\cA$ that controls at most $\eta\cdot n$ clients and possibly the server. Further, $\eta_C$ will be the proportion of corrupt clients controlled by the adversary and those belonging to the committee. The ideal world consists of the following steps, which are adapted to the more straightforward setting where only one party, i.e., the server, has the output:
\begin{enumerate}[label=(\alph*)]
    \item \label{stepa}The honest clients provide the inputs to the trusted party $\cT$. 
    \item  \label{stepb}$\tSim$ chooses which corrupted clients send the input and which abort. 
    \item  \label{stepc}If the server is corrupted, then $\tSim$ can either choose to abort the protocol or continue. 
    \item  \label{stepd}If the protocol is not aborted, then $\cT$ computes the function $F(X)$ and sends to the server. 
    \item  \label{stepe}Finally, if the server is not corrupted, it outputs what it has received from $\cT$. 
\end{enumerate}
Our function $F$ is parametrized by the following: (a) the set of inputs $X=\set{\bpsain_{i,\lab}}_{i\in [n]}$, (b) the set of client dropouts $D\subseteq [n]$, and (c) $\delta\in[0,1)$ which is the maximum fraction of dropouts permitted. In addition, we also parametrize it by the iteration ID $\lab$. Then, 

\begin{equation}
\label{eq:functionality}
F_{D,\delta}^\lab=\begin{cases}
    \sum_{i\in [n]\setminus D} \bpsain_{i,\lab} & \textbf{if } |D|\leq \delta n\\
    \bot & \textbf{otherwise}
\end{cases}
\end{equation}

%% file: caps-cons.tex
\renewcommand{\SS}{\term{SS}}
\subsection{Our Construction of \CAPS~Scheme}
\label{sec:cons-caps}
%
\subsubsection{Construction of $\caps_\LWR$}
\label{sub:caps-lwr}
We now present $\caps_\LWR$ and prove its correctness and security. We rely on the seed-homomorphic PRG (Construction~\ref{cons:shprg-lwr}) $\prg$ with seed space $\prg.\cK:=\bbZ_q^\lambda$ and combine it with Shamir's Secret Sharing Scheme over $\bbF_q$ (Construction~\ref{cons:pssf}). We will also employ a public key encryption scheme $\pkee$, which clients use to encrypt the shares to the committee. To facilitate, we assume a PKI or any other mechanism, such as authenticated channels, to ensure that every client knows a public key $\pkepk_j$ to encrypt to committee member $j$.\footnote{Recent work by Pasquini \textit{et} 
\textit{al.}~\cite{CCS:PasFraAte22} describes an attack where the server can send different models to different client updates with the goal that the model sent to a particular client can negate the training done by other clients on different models. In our case, this attack can be easily remedied with no overhead. Rather than evaluating the $\prg$ with just the public matrix $\bfA$, one can first compute a hash $\hash(iteration,model)$ and multiply it with $\bfA$. If $\hash$ was preimage-resistant, then the adversary cannot find a suitable model to force computation. As a result, the client's computation is tied to the model update sent. If different clients use different $\bfA$, the seed homomorphism fails. This would make it difficult for a malicious server to send different models to different clients to ensure that a particular client's contributions are not aggregated and, therefore, can be recovered.  We can also switch to the key-homomorphic PRF constructions described in \ifdefined\IsPRF{} the full version of the paper~\cite{EPRINT:KarPol24}\else Section~\ref{sec:khprf-cons}\fi.}

For ease of presentation, we will abuse notation and drop the packing factor $\rho$ (of the secret-sharing scheme) and the consequent vector slicing needed for the seed. We chunk the vector of length $\lambda$ into smaller vectors, each of length $\rho$, before secret-sharing each smaller vector. 

\begin{figure}[!tb]
\centering
	\resizebox{0.9\textwidth}{!}{
    \begin{minipage}{0.98\textwidth}
    \begin{protocolbox}{Construction of $\caps_\LWR$}
		\scalebox{0.85}{
			\begin{minipage}[t]{1.05\linewidth}
			    \algoHead{One-Time System Parameters Generation}
			    \begin{algorithmic}
                       \State Run $\pp\getsr\SS.\psas(1^\kappa,1^\cthr,1^\crec,1^\csize)$
                       \State Run $\pp_\prg\gets \allowbreak\prg.\prgg(1^\kappa)$ 
                \State Set  $\pp=(\pp_\prg,\cthr,\crec,\csize)$
                       \State \Return Committee of size $\csize$ and $\pp$. 
			    \end{algorithmic}
			     \algoHead{Data Encryption Phase by Client $i$ in iteration $\lab$}
			    \begin{algorithmic}
                    \State Sample $\seed_{i,\lab}\getsr \prg.\cK,\Active{\hseed_{i,\lab}\getsr\{0,1\}^{\log q}}$
                    \State Compute $(x_1,\ldots,x_L)\gets\term{Encode}(\bpsain_i)$
                    \State Compute $\bmask_{i,\lab}=\prg.\prge(\seed_{i,\lab}),\Active{ \bmask_{i,\lab}':=\hash(\hseed_{i,\lab})}$
                    \State Compute $\bpsact_{i,\lab}={(\psain_{1},\ldots\psain_{L})}+ \bmask_{i,\lab}$ \Active{$+\bmask_{i,\lab}'$}
                    \State $(\seed_i^{(j)})_{j\in[\csize]}\getsr\SS.\lshare(\seed_{i,\lab},\cthr,\crec,\csize)$
                    \State \Active{$(\hseed_{i,\lab}^{(j)})_{j\in[\csize]}\getsr\SS.\lshare(\hseed_{i,\lab},\cthr,\crec,\csize)$}
                    \For{$j=1,\ldots,\csize$}
                    \State $\psaaux_{i,\lab}^{(j)}\gets \seed_i^{(j)}\Active{,\hseed_{i,\lab}^{(j)}}$
                    \EndFor
                    \State Send $\bpsact_{i,\lab}$ to the Server
                    \State Send $c_{i,\lab}^{(j)}=\pke.\term{Enc}_{\upkepk_j}(\lab,i,\psaaux_{i,\lab}^{(j)})$ to committee member $j$ for each $j\in[\csize]$, via server. 
			    \end{algorithmic}
                \algoHead{Set Intersection Phase by Server in iteration $\lab$}
			    \begin{algorithmic}
                    \State For $j\in[\csize]$, let ${\cC^{(j)}:=\set{i:c_{i,\lab}^{(j)}\text{was received by server}}}$
                    \State Let $C^{(0)}:=\set{i:\bpsact_{i,\lab}\text{was received by the server}}$ 
                    \State Compute $\cC:=\cap_{j\in\cS\cup\set{0}}\cC^{(j)}$ \Comment{This is bit-wise AND operation of the $\crec+1$ bit strings.}
                    \State \textbf{assert} $|\cC|\geq (1-\delta)n$
                   \State Send $\cC,\set{c_{i,\lab}^{(j)}}_{i\in\cC}$ for every $j$ in $[\csize]$ 
			    \end{algorithmic}
               \algoHead{Data Combination Phase by Committee Member $j$ in iteration $\lab$}
			    \begin{algorithmic}
                    \State Decrypt $(i,\lab,\set{\psaaux_{i,\lab}^{(j)}=(\seed_{i,\lab}^{(j)}\Active{,\hseed_{i,\lab}^{(j)})}})$ from $\set{c_{i,\lab}^{(j)}}_{i\in\cC}$ using $\term{sk}_j$
                    \State Verify $i\in\cC$ and $\lab$ is the current iteration, else abort.
                    \State Compute $\AUX^{(j)}\gets\sum_{i\in\cC} \seed_{i,\lab}^{(j)}$
                    \State Send $\AUX^{(j)}\Active{,\set{\hseed_{i,\lab}^{(j)}}_{i\in\cC}}$ to server
			    \end{algorithmic}
                     \algoHead{Data Aggregation Phase by Server in iteration $\lab$}
			    \begin{algorithmic}
                    \State Let $\set{\AUX_\lab^{(j)}}_{j\in\cS},\set{\bpsact_{i,\lab}}_{i\in\cC}$ be the inputs received by the server with $|\cS|\geq\crec$. 
                    \State Run $\seed_\lab\gets \SS.\lcomb(\set{\AUX_\lab^{(j)}}_{j\in\cS})$
                    \State $\AUX_\lab=\prg.\prge(\seed_\lab)$
                     \For{\Active{$i\in\cC$}}
                        \State \Active{$\hseed_{i,\lab}\gets\SS.\lcomb(\set{\hseed_{i,\lab}^{(j)}}_{j\in\cS})$}
                    \EndFor
                    \State Compute $\boldsymbol{\term{CT}}_\lab\gets\sum_{i\in\cC}\bpsact_{i,\lab}$
                    \State Compute $(X_1,\ldots,X_L)=\boldsymbol{\term{CT}}_\lab-\AUX_\lab\Active{-\sum_{i\in\cC}\hash(\hseed_{i,\lab})}$
                    \State Compute $\bpsaout_\lab\gets\term{Decode}(\pp,(X_1,\ldots,X_L))$
                    \State \Return $\bpsaout_\lab$
			    \end{algorithmic}
		\end{minipage}
		}
	\end{protocolbox}
    \end{minipage}
    }
	\caption{Our Construction of $\caps$ built from $\prg=(\prg.\tpgen,\prg.\prge)$ with key space $\cK=\bbZ_q^\lambda$ (Construction~\ref{cons:shprg-lwr}) and the $(\cthr,\crec,\csize)$-secret sharing scheme $\term{SS}=(\SS.\tpshare,\allowbreak \SS.\recons)$ (Construction~\ref{cons:pssf}).
    Here, $\term{Encode}(\bpsain_{i,\lab}):=\Delta\cdot \bpsain_{i,\lab}+\bfr_i+2^{\kappa_s}$ and $\term{Decode} (X_i):=\lceil{X_i/\Delta}\rceil-1$ with $\Delta=2^{\kappa_s}\cdot n$ and $\bfr_i\getsr \set{0,\ldots,2^{\kappa_s}-1}^L$ for statistical security parameter $\kappa_s$ and $\delta$ is the protocol's dropout parameter. The \Active{lines} are for security against an active server. The second mask $\bmask_{i,\lab}'$, as the output of a $\hash$ is for simulation proof, for an active server. We will model $\hash$ as a programmable random oracle.  } 
	\label{fig:opa-lwr}
 \label{fig:caps-lwr}
 \label{fig:caps-active-lwr}
 \label{cons:caps-lwr}
\end{figure}

\paragraph{Correctness.} First, recall that Construction~\ref{cons:shprg-lwr} is only almost seed homomorphic. In other words, 
\[
\prg(\seed_1+\seed_2)=\prg(\seed_1)+\prg(\seed_2)+e
\]
where $e\in\{0,1\}$. 

For ease of presentation, our correctness proof is for $L=1$, but it extends to any arbitrary $L$. Therefore, while the correctness of Shamir's Secret Sharing scheme guarantees that $\seed_\lab$, computed by the server, is indeed $\sum_{i=1}^{n}\seed_i\bmod q$, there is an error growth in $\AUX_\lab$. 
Specifically, we get that:
\[
\AUX_\lab:=\prg.\prge\left(\seed_\lab=\sum_{i=1}^{n}\seed_{i,\lab}\right)=\sum_{i=1}^{n} \prg.\prge(\seed_{i,\lab})+e'
\]
Or, 
\begin{equation}
\label{eq:aux}
\sum_{i=1}^{n} \prge(\seed_{i,\lab})=\AUX_\lab-e'
\end{equation}
where $e'\in\{0,\ldots,n-1\}$.
We know that $\term{Encode}(\psain_{i,\lab}):=\Delta\cdot \psain_{i,\lab}+r_i+1$ where $\Delta=n\cdot 2^{\kappa_s}$ and $r_i\getsr\set{0,\ldots,2^{\kappa_s}-1}$ which gets:

{\small\begin{align*}
    \sum_{i=1}^n \psact_{i,\lab}&=\left(\sum_{i=1}^n (\psain_{i,\lab}\Delta +r_i) + \prge\left(\seed_{i,\lab}\right)\right)\bmod p\\
    &=\Delta \cdot \sum_{i=1}^n \psain_{i,\lab} + \sum_{i=1}^n r_i +n+\sum_{i=1}^n \prge(\seed_{i,\lab})\bmod p\\
    &=\Delta \cdot \sum_{i=1}^n \psain_{i,\lab} + \sum_{i=1}^n r_i  +n+\AUX_\lab-e'\bmod p\\
   \sum_{i=1}^n \psact_{i,\lab}-\AUX_\lab &=\Delta \cdot \sum_{i=1}^n \psain_{i,\lab} + \sum_{i=1}^n r_i +n - e' \bmod p\\
    \bar{X}_{\lab} &=\Delta \cdot \sum_{i=1}^n \psain_{i,\lab} + \sum_{i=1}^n r_i +n- e'
\end{align*}}
The jump from the second to the third step follows from Equation~\ref{eq:aux}, and to make the last jump in the proof, we require:
\[
0\leq \Delta \cdot \sum_{i=1}^n \psain_{i,\lab} + \sum_{i=1}^n r_i +n  - e' < p
\]
First, $e'\leq n-1$. This guarantees that: $0\leq \Delta \cdot \sum_{i=1}^n \psain_{i,\lab} + \sum_{i=1}^n r_i +n - e'$. Now, if $\sum_{i=1}^n \psain_{i,\lab}<(p-\Delta)/\Delta$ then we also get: 
\[
\Delta \cdot \sum_{i=1}^n \psain_{i,\lab} + \sum_{i=1}^n r_i + n - e' < p
\]
Now, we show the correctness of $\term{Decode}$ algorithm to recover $\sum_{i=1}^{n} \psain_{i,\lab}$ from $\bar{X}_\lab$. 
\begin{itemize}
    \item $\bar{X}_\lab/\Delta= \sum {\psain_{i,\lab}}+ (\sum_{i=1}^n r_i+n-e')/\Delta$
    \item $0\leq \sum_{i=1}^n r_i <\Delta$ and $0\leq e'\leq n-1 \Rightarrow 1/\Delta\leq (\sum_{i=1}^n r_i+n-e')/\Delta\leq 1$
    \item Therefore, $\lceil \bar{X}_\lab /\Delta\rceil=\sum {\psain_{i,\lab}}+1$
\end{itemize}

\begin{restatable}{theorem}{thmopalwr}
    \label{thm:opa-lwr}
    Let $\delta,\eta$ be the dropout and corruption fraction among the universe of clients and let $\delta_C,\eta_C$ be the dropout and corruption fraction among the clients in committee. Let $\kappa$ be the security parameter. Let $N$ be the total universe of clients and $n$ be the number of clients chosen for summation in each iteration while $\csize$ be the number of committee clients chosen to help in each iteration. Let $L$ be the length of the vector. 

    Let $\prg=(\prg.\tpgen,\prg.\prge)$ be the leakage-resilient, seed-homomorphic PRG defined in Construction~\ref{cons:shprg-lwr} and $\ss=(\ss.\tpshare,\allowbreak\ss.\recons)$ be the $(\cthr,\crec,\csize)$-secret sharing scheme such that $\crec>(\csize+\cthr)/2)$ defined in Construction~\ref{cons:pssf}. Further, assuming a PKI (or authenticated channels) associated with an IND-CPA secure public key encryption scheme $\pke$. Then, if $\delta_C+\eta_C<1/3$,
    $\caps_\LWR$ securely realizes the functionality $F_{D,\delta}^\lab(X)$  (defined in Equation~\ref{eq:functionality}) with server malicious security 
    with abort where $X=\set{\bpsain_{i,\lab}}_{i\in[n]-\setminus \kcorr}$ and $\kcorr\subset [N]$ and $|\kcorr|\cap [n]\leq \eta n$, under the \LWR\ assumption.
\end{restatable}
\ifdefined\IsPRF{}
The proof is deferred to the full version of the paper found in~\cite[\S G]{EPRINT:KarPol24}. 
\else 
\begin{proof}
We will prove the theorem statement by defining a simulator $\tSim$, through a sequence of hybrids such that the view of the adversary $\cA$ between any two subsequent hybrids are computationally indistinguishable. Let $H=[n]\setminus \kcorr$, the set of honest clients. Further, let $\cC=[n]\setminus D$ where $D$ is the set of dropout clients. 

It is important to note that the server is semi-honest. Therefore, it is expected to compute the set intersection of online clients $\cC$, as expected. In other words, all committee members (and specifically the honest committee members) receive the same $\cC$. This is an important contrast from active adversaries as a corrupt and active server could deviate from expected behavior and send different $\cC^{(j)}$, for different committee members. This could help it glean some information about the honest clients. 


We now sketch the proof below:
\begin{gamedescription}[name=Hybrid,nr=-1]
    \describegame This is the real execution of the protocol where the adversary $\cA$ interacts with the honest parties. 
    \describegame In this hybrid, we will rely on the security of the secret sharing scheme to do two things:
    \begin{itemize}
        \item On the one hand, all corrupt committee members receive a random share from the honest client's seed. Note that there can be only a maximum of $\cthr$ corrupt committee members. By appropriately choosing $\csize$, conditioned on $\eta$, we can guarantee that this holds with overwhelming probability. Then, for an honest client $i$, these are the shares denoted by $\set{\seed_i^{(j)}}_{j\in[\csize]\cap \kcorr}$ and are generated randomly. 
        \item On the other hand, all the honest committee members receive a valid share of the honest client's seeds. However, each honest client $i$ need to generate this from a polynomial $p(X)$ that satisfies $p(0)=\seed_i$, while also ensuring $p(j)=\seed_{i}^{(j)}$ for $j\in[\csize]\cap \kcorr$. Note that this is a polynomial time operation and is similar to the way packed secret sharing is done, where multiple secrets are embedded at distinct points of the polynomial. See Construction~\ref{cons:pssf} for how to build such a polynomial.
    \end{itemize}
    It is clear that by relying on the privacy of the secret sharing scheme, $\Hybrid_0,\Hybrid_1$ are indistinguishable for the adversary. 
    \describegame In this hybrid, we change the definition of the last honest party's ciphertext. WLOG, let $n$ be the last honest party in $\cC$. Then, we will set $\bpsact_n:=\prg.\prge(\seed_\lab)+\bpsain_\tau-\sum_{i\in\cC\cap H} \bpsact_i$. Here, $\bpsain_\tau$ is the sum of the honest clients inputs. Note that we are still in the hybrid where $\tSim$ knows all the inputs. 

    It is clear that $\Hybrid_1,\Hybrid_2$ are identically distributed, by the almost seed-homomorphism property of $\prg$, provided $\tSim$ chooses the inputs for the honest parties such that they sum up to the value in $\bpsain_\tau$. 
    \describegame Again, without loss of generality, let client 1 be the first honest client in $\cC\cap H$. We will modify the way $\bpsact_{1,\lab}$ is generated. We will set it as $\bpsact_{1,\lab}:=\bpsain_{1,\lab}+u$ where $u\getsr\prg.\cY$. 

    $\Hybrid_2,\Hybrid_3$ are indistinguishable, provided Theorem~\ref{thm:lr-prg-lwr} holds. In the reduction, we will implicitly set $\seed_1+\seed_n$ to be the leakage obtained from the Theorem~\ref{thm:lr-prg-lwr}'s challenger. In this hybrid, $\tSim$ still knows all the inputs. If it was a real PRG output, then we can simulate $\Hybrid_2$, while simulating $\Hybrid_3$ in the random case. 
    \describegame In this hybrid, we will replace $\bpsact_{1,\tau}:=u'$ for $u'\getsr\prg.\cY$. It is clear that $\Hybrid_3,\Hybrid_4$ are identically distributed. 
\end{gamedescription}
We have successfully replaced the first honest client's ciphertext with a uniformly random value independent of its input. $\tSim$ will continue to modify this for every non-dropout honest client $i\in\cC\cap H$. This leaves the clients with all but the last honest clients' ciphertext to be independent of the input while leaving the previous honest client's ciphertext to be only a function of the sum of the inputs, which can be obtained by $\tSim$'s query to the functionality. 
$\tSim$ beings its interaction with the functionality. After all the honest clients have provided inputs to the trusted party $\cT$, in Step~\ref{stepb}, $\tSim$ does not instruct any corrupted client to abort but rather set their inputs to be 0. Then in Step~\ref{stepc}, $\tSim$ does not abort the server. Therefore, in Step~\ref{stepd}, $\tSim$ will learn the sum of the honest parties inputs. Denote it as $\bpsain_\lab$, which is also the sum of the inputs of the honest, surviving clients. With this information, $\tSim$ uses the last hybrid to interact with the adversary $\cA$, who's expecting the real world interaction. This will enable $\tSim$ to run $\cA$ internally. This is crucial to ensure that $\tSim$ can get the output of $\cA$, in the real world, which might depend on its view (including the output) of the server. This view will, in turn, depend on the the honest clients' inputs. Since $\tSim$ sets the honest inputs, in this internal execution, to match the sum of inputs in the real world, we can guarantee that the output of $\cA$ in the internal simulation is indistinguishable from $\cA$'s interaction in the real world by the aforementioned hybrid arguments. 
\end{proof}
Our constructions so far have relied on providing security against
a semi-honest server. Note that, as shown in the proof of security for Theorems~\ref{thm:opa-lwr}, we can use the functionality query to obtain the sum of all the honest non-drop out clients, as before. 
     
In the semi-honest setting, it is easy to see that the set $\cC$, with respect to which aggregation is performed, includes \emph{all} the honest, non-dropout clients' inputs. Therefore, querying the functionality, $\tSim$ does indeed get the sum of all the honest clients' inputs that are also included in the summation in the real world. This is imperative to ensure that $\tSim$, when internally invoking $\cA$, can get the output of $\cA$ which should be indistinguishable from $\cA$'s output in the real world. Specifically, this output of $\cA$ (in either the internal invocation or the actual execution) will depend on the view which consists of the output of the server. Therefore, if the output of the server in the real world does not include any of these honest clients' inputs, then the output produced by the internal invocation of $\cA$ can be different from that in the real world. 

Let us look at the case when the server is corrupted. Such a server can mount an attack whereby the real-world execution of the protocol may exclude inputs of some of those honest parties but actually included in the output of the ideal functionality. The proof of malicious security is tricky in this setting. Specifically, a malicious server can drop clients after seeing the honest input. This is an issue in the simulation as the simulator has to generate the masked inputs for the honest clients without knowing which of them would be dropped later.

\begin{proof}[Proof of Security against Active Server.]
    Let us look at the setting when the server is corrupted. Then, we need to look at the information that
the adversary $\cA$ gleans from this, in the real world. As before, let $\kcorr$ denote the corrupted clients.
Then, $\khon_{\term{Cli}}:=[n]\setminus \kcorr$ is the set of honest clients, $\khon_\term{Com}:=[\csize]\setminus\kcorr$ is the set of honest
committee members. Let $\kcorr_{\term{Cli}}:=[n]\cap \kcorr$ denote the corrupted clients and $\kcorr_{\term{Com}}:=[\csize]\cap \kcorr$ denotes the corrupted
committee members. 

\begin{itemize}
    \item It receives every honest client's masked input, i.e., $\set{\bpsact_i}_{i\in H}$
    \item It receives every honest client's shares sent to the corrupted committee members, i.e., $\set{\seed_{i}^{(j)}}_{i\in \khon_{\term{Cli}},j\in \kcorr_{\term{Com}}}$
    \item It receives the combined shares sent by every honest committee member $j$, $\sum_{i\in \khon^\ast} \seed_{i}^{(j)}$ where $\khon^\ast$ is the final set over which the aggregation occurs. 
\end{itemize}
\begin{remark}[Using Signatures]
    \label{rem:signatures}

Figure~\ref{fig:caps-lwr}, as previously described, does not incorporate the use of signatures. Nevertheless, the implementation of signatures is crucial to prevent a specific type of residual attack. Specifically, a malicious server can craft a message that encrypts 0 for all clients except one. Consequently, the committee inadvertently assists in reconstructing the information of the targeted client, which can then be exploited to reveal the inputs. This attack is successful due to the absence of a binding between the client's identity and the message it sends. It is important to note that this attack can succeed even when all clients are acting honestly. However, this issue can be effectively addressed by employing signatures, whereby each client signs each message. This solution presupposes the existence of a Public Key Infrastructure (PKI) setup for all clients, and aggregate signatures can be utilized to minimize communication overhead. Nonetheless, this requirement may be burdensome for certain applications, such as the one required in Willow~\cite{Willow}. If one chooses to forgo signatures, they must be prepared to endure this attack. While incorporating differential privacy, as suggested by Willow~\cite{Willow}, can mitigate data leakage, it cannot fully prevent the attack, as it only offers differentially private guarantees for the clients' inputs.

\end{remark}

\paragraph{Committee Speaks Once.} Our construction currently has the server identify $\cC^{(j)}$, for each $j\in[\csize]$, based on the information it has received from the client. Then, the server forwards the message to the committee member along with its computed intersection. This setting allows the server to selectively forward shares to committee member and also choose different sets for different committee members. We will show that if $\crec>(\csize+\cthr)/2$ where $\crec$ is the reconstruction threshold in the committee and $\cthr$ is the corruption threshold, then the server doing so will receive meaningless information. Formally, we will show that  there does not exist two sets of users $\cC\neq \cC'$ such that the server can reconstruct the shares over these two sets. 

Observe that the server controls $\cthr$ committee members. We require each honest committee member to participate once per iteration. This is easily enforced as the share from the honest client encrypts, along with the share for the honest committee member, the identity of the honest client and the iteration count. Therefore, a server cannot replay the same share, in another iteration. With this guarantee, a malicious server, in order to reconstruct the shares of two distinct sets $\cC,\cC'$, will require the cooperation of at least $\crec-\cthr$ honest users, while there are $\crec-\cthr$ honest users present. We will therefore need $2(\crec-\cthr)>\csize-\cthr$. Or, $\crec>(\csize+\cthr)/2$. This ensures that the server can only effectively reconstruct with respect to a unique set $\cC$ and $\khon^\ast$ is the set of honest users in this set. 
Note that the above inequality holds for $\crec=2*\csize/3, \cthr<\csize/3$. Indeed, prior works such as Bonawitz~\etal~\cite{bonawitz2017practical} and most recently LERNA~\cite{AC:LLPT23} also tolerated only upto a $\csize/3$ corruption threshold. 

While we have shown that there is a unique set $\khon^\ast$ of honest users, $\khon^\ast$ is only revealed after all the honest clients have sent their inputs. Therefore, the simulator, during its internal execution of $\cA$, needs to be able to generate the masked inputs for the honest users and it only knows the sum of \emph{all} the honest clients that have not dropped out. This set may be distinct from $\khon^\ast$. Therefore, we need a way for the simulator to generate masked inputs, independent of the sum of the inputs, and then ensure that the correct sum is computed during reconstruction.

The simulator does the following:
\begin{itemize}
    \item For every honest client that hasn't dropped out, i.e., for all $i\in \khon_{\term{Cli}}$, the simulator does the following:
    \begin{itemize}
        \item Samples $\hseed_{i,\lab}\getsr\{0,1\}^{\log q}$
        \item Samples $\seed_{i,\lab}\getsr\prg.\cK$
        \item It computes $\bmask_{i,\lab}':=\hash(\hseed_{i,\lab}$
        \item Computes $\bmask_{i,\lab}=\prg.\prge(\seed_{i,\lab})$
        \item Sets $\bpsact_{i,\lab}:=\bmask_{i,\lab}+\bmask_{i,\lab}'$
        \item Like shown in proof of Theorem~\ref{thm:opa-lwr}, the shares of $\seed_{i,\lab}^{(j)},\hseed_{i,\lab}^{(j)}$ for corrupt committee members $j\in \kcorr_{\term{Com}}$ are chosen at random. Meanwhile, the shares for the honest committee members are to be sampled in the second phase, with a specific purpose. However, the server still expects an encryption of shares from honest client to honest committee members. Therefore, it simply encrypts some random shares for the honest committee members too and sends it to the server.
        \item It sends to $\cA$, $\bpsact_{i,\lab}$ and $\seed_{i,\lab}^{(j)}$ and $\hseed_{i,\lab}^{(j)}$ for $j\in[\csize]$, which is encrypted appropriately. 
    \end{itemize}
    \item This concludes the client phase of the operation. Then, comes the interaction with the committee. Note that the simulator is also required to simulate the honest committee member $j$. 
    \item The simulator, which has received $\cC^{(j)}$ for each honest committee member $j$ does the check to make sure that there exists at least $\crec-\cthr$ such committee members with the same $\cC^{(j)}$. We will call this client set as $\cC$, while calling the set of these committee members to be $C_{\term{good}}$. 
    Meanwhile, it records those committee members with a different $\cC{(j)}$. We will call this as some set $C_{\term{bad}}$. 
    Looking ahead, for those honest committee members in $C_{\term{bad}}$, the shares of the honest clients that are to be added up is going to be random values. Note that $|C_\term{bad}|\leq \csize-\crec$. 
    \item The simulator now operates in two phases for honest committee member $j\in\khon_{\term{Com}}$. First is the share generation phase for honest clients $i$. It does the following:
    \begin{itemize}
        \item If $j\in C_{\term{bad}}$, then for honest client $i\in \cC^{(j)}\cap \khon_{\term{Cli}}$, set $\seed_{i,\lab}^{(j)},\hseed_{i,\lab}^{(j)}$ to be random values. 
        \item Now, the simulator computes the shares for all honest clients $i$ to $j\in C_{\term{good}}$. These are valid shares of $\seed_{i,\lab},\hseed_{i,\lab}$ subject to the constraint that random values were fixed for those $j\in C_{\term{bad}}$ where $i\in\cC^{(j)}$, and for those $j\in \kcorr_{\term{Com}}$. 
        \item The honest committee member $j$ receives from $\cA$, $\seed_{i,\lab}^{(j)}$ and $\hseed_{i,\lab}^{(j)}$ for $i\in\cC^{(j)}$. Note that the maximum number of prefixed values is $\csize-\crec+\cthr$, and by our constraint $\crec>\csize-\crec+\cthr$ which guarantees that these prefixed values cannot uniquely determine a polynomial of degree $\crec$. 
    \end{itemize}
    \item The second phase, is the combination phase. It responds, as expected, subject to the set $\cC^{(j)}$ that it receives. 
    \item $\tSim$ now queries the functionality. First, it provides $[n]\setminus\cC$ as the set of dropped out clients. Then, it sends for those corrupted clients $\kcorr\cap \cC$, input as 0 to the functionality. In response, it gets
    $\sum_{\cC\cap H} \bpsain_{i,\lab}$. Call this $\bpsain_{H}$.
    \item $\tSim$ now picks $i^\ast\in \khon^\ast$. It programs the random oracle by setting $\hash(\hseed_{i^\ast,\lab})=\bpsain_{H}-\bmask_{i^\ast,\lab}'$.
    \item $\tSim$ continues to respond, on behalf of the honest committee member, as expected. 
    \item Finally, $\cA$ (which controls the server) will make queries to random oracle and it answers as expected. $\tSim$ outputs whatever $\cA$ outputs at the end. 
\end{itemize}
We will now need to show that the above simulation is indistinguishable from the real world execution that $\cA$ expects when it is internally run. The hybrids proceeds as follows:
\begin{gamedescription}[name=Hybrid,nr=-1]
    \describegame This is the real world execution. 
    \describegame In this execution, we replace the shares sent by the honest client $i$ to honest committee member $j$, which are encrypted under $\upkepk_j$ with a random value. Under the semantic security of this encryption scheme, we can guarantee that this is indistinguishable from the previous hybrid. Meanwhile, these honest committee members (which the simulator controls) will receive the shares directly from the simulator. The view of $\cA$, in this hybrid, is indistinguishable from the real world execution, under the semantic security of the encryption scheme. 
    \describegame We will rely on the security of the secret sharing scheme to sample the shares for the honest clients, similar to $\Hybrid_1$ of semi-honest security. For those $j\in\kcorr_{\term{Com}}$, the shares are randomly chosen. 
    Furthermore, for those $j\in C_{\term{bad}}$ also the shares are randomly chosen. Finally, for those $j\in C_{\term{good}}$ it gets a valid share subject to those previously chosen random values. This is similar to $\Hybrid_1$ in the proof of semi-honest security. 
    \describegame In this hybrid, for all those honest clients $i$ that are not in $\cC$, we will set $\bpsact_{i,\lab}=\bmask_{i,\lab}+\bmask_{i,\lab}'$, effectively setting the input to be 0. Observe that the view of $\cA$ remains unchanged as these honest clients inputs were never incorporated in the final sum anyway. Furthermore, if any of these $i\in\cC^{(j)}$ for $j\in C_{\term{bad}}$, the shares from these honest clients $i$ to these $j$ are completely random \emph{and} independent of $\bmask_{i,\lab}$ and $\bmask_{i,\lab}'$. 
    \describegame In this hybrid, we pick an honest surviving client $i^\ast\in\khon^\ast$. It sets the inputs for all $i\neq i^\ast \in \khon^\ast$ to be 0. Then sets $\bpsain_{i^\ast,\lab}$ to be the sum of all the $i\in\khon^\ast$. Call this sum as $\bpsain_{H}$. Observe that the values are still correlated and pseudorandom. 
    \describegame In this hybrid, we will program $\hash(\hseed_{i^\ast,\lab})=\bpsain_H-\bmask_{i^\ast,\lab}'$, while setting $\bpsact_{i^\ast,\lab}=\bmask_{i^\ast,lab}+\bmask_{i^\ast,\lab}'$. Note that because $\hseed_{i^\ast,\lab}$ is chosen uniformly at random from $q$ values where $q$ is a large prime. The probability of collision is negligible. There is only negligible difference in the view of $\cA$.
    \describegame In this hybrid, we will set $\bpsact_{i,\lab}$ for $i\neq i^\ast$ to be some random term in the ciphertext space. Then, we will set $\bpsact_{i^\ast,\lab}=\hash(\hseed_{i^\ast,\lab})-\sum_{i\neq i^\ast}\bpsact_{i,\lab}$.
    
    Note that under the leakage resilience property of the seed-homomorphic PRG, we can conclude that the two hybrids are computationally indistinguishable. 
    \describegame In this hybrid, we will replace $\bpsact_{i,\lab}=\bmask_{i,\lab}+\bmask_{i,\lab}'$ for $i\neq i^\ast$. 
\end{gamedescription}
Observe that this last hybrid is exactly what the simulator produces. This concludes the proof. 
\end{proof}
\fi
Achieving malicious security with abort is outlined in Section~\ref{sub:mal-clients}. 
\paragraph{Sampling the Committee.} Let $C$ be a committee of size $\csize$. 
 Specifically, to ensure $\eta_C+\delta_C<1/3$, we need $\eta_C<1/3 - \delta_C$.  While the server controls $\delta_C$, $\eta_C$ is a random variable. Given the universe of $N$ clients where $\eta$ fraction of them are malicious, randomly sample $\csize$ clients from them. Then, the number of malicious clients $X$ in the committee should follow the tail bound of hypergeometric distribution (Definition~\ref{def:hypergeo}). 
 \[
 \Pr[X \geq (\eta + (1/3 - \delta_C - \eta))\csize] \le e^{-2\cdot \csize (1/3-\delta_C-\eta)^2}\leq 2^{-\gamma}
 \]
 for some $\gamma$. Observe that $X/\csize = \eta_C$. 
Assuming $\eta=\delta_C=0.01$, and committee of size 50, we get that $\Pr[\eta_C\geq 1/3-\delta_C]\leq 5\cdot 10^{-5}$. 
Extending it to the setting where we have a much larger pool of committee members, say $M\cdot \csize$, corresponding with $M$ groups each of size $\csize$. Call the smaller groups $C_1,\ldots,C_{M}$. We want to upper bound the probability that $\cup_{i=1}^{M} \Pr[\delta_{C_i}+\eta_{C_i}>1/3]$. Using Union Bound, we get that $\cup_{i=1}^{M} \Pr[\delta_{C_i}+\eta_{C_i}>1/3]\leq M\cdot 2^{-\gamma}$. 
Let $M=N/\log n$ (thereby assigning $\log n$ clients for each subcomittee of size $\csize$), then $\bigcup_{i=1}^{M} \Pr[\delta_{C_i}+\eta_{C_i}>1/3]\leq  2^{\log n -\log\log n-\gamma}$.

\subsubsection{Construction of $\caps_\LWE$}
\label{sub:caps-lwe}
\begin{construction}
\label{cons:caps-lwe}
    As alluded to before, $\caps_\LWE$ is largely similar to $\caps_\LWR$
with the following differences:
\begin{itemize}
    \item While the seed of $\prg_\LWE$ is $(\seed,\bfe)$, we will only secret share $\seed$. We will argue below that the correctness still holds for a suitable definition of $\chi$. 
    \item The plaintext space for $\caps_\LWE$, like the one for $\caps_\LWR$, is $\mathbb{Z}_p$. Meanwhile the seed space for both $\caps_\LWE$ and $\caps_\LWR$ will be $\mathbb{Z}_q$. Let $\Delta:=\lfloor{q/p}\rfloor$.
    \item We will use Shamir's Secret Sharing over $q$, as before, which is the seed space.
    \item There is a change in the server's last phase. To compute $\AUX_\lab$, the server uses the reconstructed seed $\seed_\lab$, and additionally sets the error component of the PRG seed to be 0. 
    \item $\term{Encode}(\bfx_{i,\lab}):=\Delta\cdot \bfx_{i,\lab}$
    \item $\term{Decode}(X_i):=\lceil X_i/\Delta\rceil -1$
\end{itemize}
\end{construction}
Due to the similarities, we do not present the construction entirely. Meanwhile, we present the proof of correctness and proof of security. 

\paragraph{Correctness.} First, observe that Construction~\ref{cons:shprg-lwr} is only almost seed-homomorphic, i.e.
\[
\prg((\seed_1+\seed_2,\bfe),\lab)=\prg((\seed_1,\bfe_1),\lab)+\prg((\seed_2,\bfe_2),\lab)+\bfe'
\]
for some error $\bfe'$. Indeed, assuming the correctness of Shamir's Secret Sharing, we get that the server computes:
\[
\AUX_\lab:=\prg.\prge\left(\left(\sum_{i=1}^{n} \seed_i,0\right),\lab\right):=\bfA\cdot \sum_{i=1}^{n} \seed_i 
\]
Meanwhile,
{\small\begin{align*}
    \sum_{i=1}^n \psact_{i,\lab}&=\sum_{i=1}^{n} \left(\bfA \seed_i + \bfe_i+ \Delta\cdot \psain_{i,\lab}\right)\\
    &=\bfA\sum_{i=1}^n \seed_i +\sum_{i=1}^n \bfe_i+\sum_{i=1}^{n} \Delta\cdot\psain_{i,\lab}\\
   \text{Let } X_{\lab}&:=\sum_{i=1}^n \psact_{i,\lab} - \AUX_\lab=\sum_{i=1}^n \bfe_i+\sum_{i=1}^{n} \Delta\cdot\psain_{i,\lab}\\
    \frac{X_\lab}{\Delta}&=\frac{\sum_{i=1}^n \bfe_i}{\Delta} + \sum_{i=1}^n \psain_{i,\lab}
\end{align*}}
If $\sum_{i=1}^{n} \bfe_i <\frac{\Delta}{2}$, then 
$ \lceil \frac{X_\lab}{\Delta}\rceil= \sum_{i=1}^n \psain_{i,\lab}+1$.
This shows the correctness of our algorithm. 

\ifdefined\IsPRF{} We get a similar theorem statement and proof as for the \LWR\ construction and is deferred to the full version of the paper~\cite[\S 6.2.2]{EPRINT:KarPol24}\fi
\ifdefined\IsSub{}
\else 
\begin{proof}
    The proof proceeds similar to that of Theorem~\ref{thm:opa-lwr}, through a sequence of hybrids. However, there are a few differences.
    Construction~\ref{cons:shprg-lwe} has the error vector $\bfe\getsr\chi$. However, we will replace $\bfe=\bfe'+\bff'$ where $\bfe',\bff'\getsr\chi'$, the distribution present in Hint-LWE Assumption (see Definition~\ref{def:hint-lwe}). The hybrid descriptions are similar, so we only specify the differences: 
    
    \begin{itemize}
        \item In $\Hybrid_2$ we will set: 
         $$\bpsact_{n,\lab}=\bfA\cdot \seed_\lab-\sum_{i\in(\cC\cap H)\setminus \set{n}} \bpsact_i +\bfe_\lab+\Delta \bpsain_{\lab}$$
         \item We will argue that $\Hybrid_2,\Hybrid_3$ are indistinguishable under Hint-LWE Assumption. We will sketch the reduction now. 
         \begin{itemize}
             \item  Recall that, from the Hint-LWE Challenge, we get $(\bfA,\bfu^\ast,\bfs^\ast:=\bfs+\bfr,\bfe^\ast:=\bfe'+\bff')$. 
             \item As done for the $\LWR$ construction, we will set the $\bfs_1+\bfs_n=\bfs^\ast$, the leakage on key.
             \item  For generating $\bpsact_{1,\lab}$ we will use $\bfu^\ast$, while also sampling a separate $\bff_1\getsr\chi'$. This gives us: $\bpsact_{1,\lab}=\bfu^\ast+\bff_{n,\lab}'+\Delta\cdot \bpsain_{1,\lab}$
             \item We will set $\bpsact_{n,\lab}:=\bfA\cdot \seed_\lab-\sum_{i\in(\cC \cap H)\setminus \set{n}} \bpsact_i+\bfe^\ast+\sum_{i\in(\cC \cap H)\setminus \set{1}} (\bfe_{i,\lab}'+\bff_{i,\lab}')+\Delta\cdot \bpsain_{\lab}$
             \item When $\bfu^\ast$ is the real sample, then $\bpsact_{1,\lab}$ satisfies $\Hybrid_2$'s definition. Meanwhile, the $\bpsact_{n,\lab}$ is also correctly simulated. Similarly, the case when it a random sample. 
         \end{itemize}
    \end{itemize}
    The proof of security against malicious servers also follows the previous theorem. 
    \end{proof}
\fi

%% file: key-reuse.tex
\subsection{\CAPS\ without Leakage Simulation}
\label{sec:key-reuse}
$\caps$ requires generating a new key in every iteration. We will now present a construction $\caps'$ such that: 
(a) a client's keys can be reused across multiple iterations, and (b) the server does not get the sum of the keys but rather a function of pseudorandom values, which can be argued as itself being pseudorandom. Our core technique in this work is a distributed, key-homomorphic PRF. We formally present constructions from the $\cl$ framework in \ifdefined\IsPRF{}the full version of the paper~\cite[\S B]{EPRINT:KarPol24}\else Section~\ref{sub:khprf-cons}\fi. Specifically, we defer $\caps_\cl$ to \ifdefined\IsPRF{}the full version of the paper~\cite[\S B.2]{EPRINT:KarPol24} and the \LWR-based construction in ~\cite[\S D]{EPRINT:KarPol24}\else  the appendix in Section~\ref{sub:async}. Similarly, we present \LWR\ based construction in Section~\ref{sec:lwr-cons}\fi. Next, we broadly describe the intuition behind our construction. It is important to emphasize that this is a purely theoretical construction as the committee's performance depends on $L$. However, we document this alternative approach for completeness.

A distributed-key-homomorphic PRF has three specific algorithms: $\tpeval$, which allows the evaluation of the PRF with key $\tpk_i$ at a point, $\tppeval$ allows the PRF to be evaluated at a share of the key $\tpk_i^{(j)}$ to get a partial evaluation, and $\tpcombine$ allows for the combination of partial evaluations to recover the actual evaluation. Key Homomorphism implies that both partial and actual evaluations are key homomorphic. Therefore, in our construction, the clients mask a vector of inputs $\bpsain_{i,\lab}$ by computing a pseudorandom evaluation of $\tprf.\tpeval(\tpk_{i,1},\lab,\ldots,\tprf.\tpeval(\tpk_{i,L},\lab)$. Meanwhile, the auxiliary information send to the committee member will be $\tppeval(\tpk_{i,k}^{(j)},\lab)$
for $k=1,\ldots,L$ and $j\in[\csize]$. Then, the committee members combine 
the auxiliary information. The server then reconstructs on its end using $\tpcombine$. Correctness follows from the key homomorphism. Note that the server only computes $\tprf.\tpeval(\sum_{i=1}^{n} \tpk_{i,k},\lab)$ for $k=1,\ldots,L$. This is a PRF evaluation; therefore, the leakage is pseudorandom and can be easily simulated, replacing it with random. 

Finally, we also strengthen the security of $\caps'$, instantiated with the $\HSM$$_{=\Prime}$ based construction (dubbed $\caps_\cl$), by offering the privacy of honest user's inputs when all the committee are corrupt with each other while the server is honest. This is covered in \ifdefined\IsPRF{}the full version of the paper~\cite[\S G]{EPRINT:KarPol24}\else Section~\ref{sec:stronger}\fi. 

%% file: malicious.tex
\section{Malicious Security with Abort}
\label{sub:mal-clients}
We detail how to provide security against clients who behave maliciously. We detail how: 
\begin{itemize}
    \item[(a)] a client proves that it has shared its key correctly. This is done so using SCRAPE Test~\cite{ACNS:CasDav17} which minimizes the computation on the part of the client (unlike traditional verifiable secret sharing),
    \item[(b)] a client proves that it has computed the masking correctly,
    \item[(c)] a client combines the previous two properties into one relation that is linear over $\bbZ_q$, and
    \item[(d)] a client proves that its masked input satisfies some particular constraints. 
\end{itemize}
We present instantiations of the above using the proof system from Lyubashevsky~\etal\cite{C:LyuNguPla22}. 
Importantly, $\caps$ with bRSA (described in Section~\ref{sub:brsa}) implies that the client's input is a binary vector. This shall simplify the zero-knowledge proof needed for input validation. 

\paragraph{Proof of Correct Sharing.} At its core, $\caps$ (both the \LWR\ and \LWE\ based version) involves the use of a $\seed_{i,\lab}$ that is secret-shared using Shamir's Secret Sharing. For simplicity, we will focus on the unpacked version where only one secret is shared. However, the proof extends to the packed setting as it only tests the existence of a polynomial. A standard 
approach to catch clients sending inconsistent shares would be for the client to rely on a verifiable secret sharing which would empower each committee member to verify if the share 
received by the committee member is consistent with the commitments to the polynomial that was sent by the client. However, this requires
each committee client to perform $n\cdot \crec$ exponentiations which can be expensive where $\crec$ is the reconstruction threshold. Instead, we take an alternative approach to verifying that the secret and the shares lie on the same polynomial, while ensuring that the bulk of the verification is done by the powerful server.

Our public verifiability will rely on a modification of SCRAPE~\cite{ACNS:CasDav17}. SCRAPE test is done to check if $(\seed_{i,\lab}^{(1)},\ldots,\seed_{i,\lab}^{(\csize)})$ is a Shamir sharing over $\bbZ_q$ of degree $d=\crec-1$ (namely there exists a polynomial $p$ of degree $\le d$ such that $p(i) = s_i$ for $i = 1, \ldots, n$), one can sample $w^{(1)}, \ldots, w^{(\csize)}$ uniformly from the dual code to the Reed-Solomon code formed by the evaluations of polynomials of degree $\le d$, and check if $\sum_{i=1}^{\csize} w^{(i)}\cdot \seed_{i,\lab}^{(i)} = 0$ in $\bbZ_q$. If the test passes, then $\seed_{i,\lab}^{(1)},\ldots,\seed_{i,\lab}^{(\csize)}$ are Shamir Shares, except with probability $1/|\bbZ_q|$. Recall that $q$ is a prime number. 

\begin{lemma}[SCRAPE Test~\cite{EC:CasDav24}]
\label{lem:scrape}
    Let $\bbZ_q$  be a finite field and let $d=\crec-1, \csize$ be parameters of the Shamir's Secret Sharing scheme such that $0\leq d\leq \csize-2$, and inputs $\seed_{i,\lab}^{(1)},\ldots,\seed_{i,\lab}^{(\csize)} \in \bbZ_q$. Define $v_i:=\prod_{j\in[\csize]\setminus \{i\} }(i-j)^{-1}$ and let $m^\ast(X):=\sum_{i=0}^{\csize-d-2} m_i\cdot X^i\gets_\$ \bbZ_q[X]_{\leq \csize-d-2}$ (i.e., a random polynomial over the field of degree at most $\csize-d-2$). Now, let $\bfw:=(v_1\cdot m^\ast (1),\ldots,v_n\cdot m^\ast(\csize))$ and $\bfs:=(\seed_{i,\lab}^{(1)},\ldots,\seed_{i,\lab}^{(\csize)})$. Then, 
\begin{itemize}
    \item If there exists $ p\in\bbZ_q[X]_{\leq d}$ such that $\seed_{i,\lab}^{(i)},=p(i)$ for all $i\in[n]$, then $\langle \bfw,\bfs\rangle=0$. 
    \item Otherwise, $\Pr[\langle \bfw,\bfs\rangle=0]=1/|\bbZ_q|. $
\end{itemize} 
\end{lemma}
Typically, we compute the polynomial $m^\ast(X)$ by using the Fiat-Shamir transform over public values. Then, the vector $\bfw$ is a public vector. One simply has to hide the vector $\bfs$. As a result, we have proved that the shares do lie on the same polynomial. Note that in standard Shamir's Secret Sharing, we set $p(0)=\seed_{i,\lab}$, i.e., the secret. Therefore, we will have to perform inner product over a vector of length $\csize+1$. 

Thus, we get the following relation for the proof of correct-sharing:
\[
\cR_{\term{Sharing}}:=\left\{
\begin{array}{c|c}
    (\bfw;\bfs) &  \langle \bfw,\bfs\rangle =0, \bfw\in \bbZ_q^{\csize+1}, \bfs\in \bbZ_q^{\csize+1}
\end{array}
\right\}
\]
While we will employ a commit-and-prove paradigm, observe that the server can only verify if the commitment satisfies the required proof. A malicious client can send a completely arbitrary share to the committee member, under the hood of decryption. To solve this problem, the server will forward the commitment to the share it has received. This will allow the committee member to receive if the share (and necessary opening information) received by the committee member opens the commitment forwarded by the server. If it fails, the committee member complains and the protocol aborts. This guarantees malicious security with abort. 

\paragraph{Proof of Correct Masking.} $\cR_{\term{Sharing}}$ guarantees that the shares are consistent. In contrast, the commitment opening by the committee member guarantees that the share received matches the share committed to; it is important to ensure that the correct seed is used in the computation of the seed-homomorphic PRG. For example, in $\caps_\LWR$ (for $L=1$), client $i$ computes
\[
\bpsact_{i,\lab} = \bfA\cdot \seed_{i,\lab} + \bfe_{i,\lab}+ \Delta\cdot \bpsain_{i,\lab}
\]
Here, $\bfA$ is public while $\seed_{i,\lab}\getsr\bbZ_q^\lambda,\bfe\getsr\chi^L$ are secret along with $\bpsain_{i,\lab}\in \bbZ_p^{L}$. Therefore, we get:
\[
\cR_{\term{Masking}}:=\left\{
\begin{array}{c|c}
    (\bpsact_{i,\lab}\bfA;\seed_{i,\lab},\bfe_{i,\lab},\bpsain_{i,\lab}) &  \bpsact_{i,\lab} =\bfA\cdot \seed_{i,\lab} + \bfe_{i,\lab}+ \Delta\cdot \bpsain_{i,\lab}, \\
    & \bpsact_{i,\lab} \in \bbZ_q^L,\seed_{i,\lab}\in\bbZ_q^\lambda,\bfe\in\chi^L,\bpsain_{i,\lab}\in \bbZ_p^{L}
\end{array}
\right\}
\]
Note that $\bfe\in\chi^L$ is typically proved by showing that its $L_2$ (or $L_\infty$) norm is bounded. 

\paragraph{Proof of Linear Relations over $\bbZ_q$.} It is critical to observe that the relations $\cR_{\term{Sharing}}$ and $\cR_{\term{Masking}}$ are both linear relations over $\bbZ_q$ (except showing that $\bfe$ is ``small'', i.e., $\Vert \bfe \Vert_2 \leq \beta_e$ for some parameter $\beta$. We will combine to get the following relation:
\[
\cR_{\caps}:=\left\{
\begin{array}{c|c}
 & \bpsact_{i,\lab} =\bfA\cdot \seed_{i,\lab} + \bfe_{i,\lab}+ \Delta\cdot \bpsain_{i,\lab}\\
\left(\bfw,\bfA,\bpsact_{i,\lab};\seed_{i,\lab},\set{\seed_{i,\lab}^{(j)}},\bfe_{i,\lab},\bpsain_{i,\lab}\right)      &  \left\langle{\bfw,\left(\seed_{i,\lab},\seed_{i,\lab}^{(1)},\ldots,\seed_{i,\lab}^{(\csize)}\right)=0}\right \rangle \\
     & \Vert \bfe_{i,\lab} \Vert_2 \leq \beta_e,\bfA\in\bbZ_q^{L\times\lambda},\seed_{i,\lab}\in \bbZ_q^{\lambda},\\
     & \bpsact_{i,\lab}\in\bbZ_p^L
\end{array}
\right\}
\]
One can employ the techniques of Lyubashevsky~\etal~\cite[Figure 5]{C:LyuNguPla22} to prove the linear relations between the various secrets. An alternative approach for $\cR_{\term{Sharing}}$ is described in \ifdefined\IsPRF{}the full version of the paper~\cite[Figure 11]{EPRINT:KarPol24}\else Figure~\ref{fig:mal-clients}\fi. Then, one can use Lyubashevsky~\etal~\cite[\S 4]{C:LyuNguPla22} to prove the $L_2$ norm on $\bfe$. 

\paragraph{Input Validation.} A malicious client can provide incorrect or malicious inputs under encryption. Therefore, adding a validation restriction on the input is also essential. Typically, this is done by proving that the $L_2, L_\infty$ norm is bounded. Therefore, it is clear that one can do techniques similar to bounding $\bfe$, to prove that $\Vert \bpsain_{i,\lab}\Vert$ also satisfies some prior bounds. However, looking ahead, in Section~\ref{sec:byzantine}, we describe how to combine bRSA with $\caps$, where each client's input to server is a vector of 0s and 1s. This proof technique is much simpler. Specifically, when combining $\caps$ with bRSA~\cite{brsa} as described in \ifdefined\IsPRF{}the full version of the paper~\cite[\S F]{EPRINT:KarPol24}\else Section~\ref{sec:byzantine}\fi, we require clients only to prove that their input is a binary vector. 

%% file: exp.tex
\section{Experiments}
\label{sec:exp}
\input{graph-figures/agg-graph-total}

\input{graph-figures/mcc-graph}
In this section, we benchmark $\caps_\LWR$, the construction based on leakage resilient seed-homomorphic PRG, combined with a secret-sharing scheme. \ifdefined\IsFull 
\else Due to space constraints some of the experiments are deferred to \ifdefined\IsPRF the full version of the paper~\cite[\S 8]{EPRINT:KarPol24}.\else Section~\ref{sec:other-exp}.
This includes benchmarking the computation time of various protocols (Figure~\ref{fig:client-comp}), the performance of multiple protocols for $L=1000$ (Figure~\ref{fig:client-1000}), and $\caps'$ instantiated with the \HSM$_{=\Prime}$ called $\caps_\cl$ (Table~\ref{tab:computation_times}).\fi \fi
We run our experiments on an Apple M1 Pro CPU with 16 GB of unified memory without multi-threading or related parallelization. We use the ABIDES simulation~\cite{Abides} to simulate real-world network connections. ABIDES supports a latency model, represented as a base delay and jitter, which controls the number of messages arriving within a specified time. Our base delay is set with the range from 21 microseconds to 100 microseconds), and we use the default parameters for the jitter. This delay is set to correspond to devices locally situated. This framework was used to measure the performance of other prior work, including \cite{SP:MWAPR23,GPSBB22}. More details on the framework can be found in \cite[\S G]{GPSBB22}. 

\paragraph{Parameter Choices     for $\caps_\LWR$.} $\caps_\LWR$ is parametrized by $\lambda,L,q,p$ (See Definition~\ref{def:lwr}). Given an LWR instance  \(\left(\mathbf{A},\mathbf{t} = \left\lfloor \frac{p}{q} \cdot \mathbf{A} \mathbf{s} \right\rfloor\right) \in \bbZ_q^{L\cdot \lambda}\times \mathbb{Z}_p^L
    \), one can convert it to an LWE instance as: $(\mathbf{A},q/p \cdot \mathbf{t} = \mathbf{A}\mathbf{s}+ \mathbf{e})$, where $\mathbf{e}=q/p\cdot \mathbf{e}'$ with $\mathbf{e}'=\mathbf{t}-p/q \mathbf{A\cdot s}$. Now, $\mathbf{e}'$ is uniformly distributed over $\left(\frac{-1}{2},\frac{1}{2}\right)$. Thus, we get that the variance of the error in the LWE samples ($\mathbf{e}$) is $\frac{q^2}{12p^2}$. This analysis is similar to the work of Okada~\etal~\cite[\S 5.5]{okada}. In other words, the Gaussian Error distribution with standard deviation $\alpha\cdot q$ has error rate $\alpha\approx 1/p$ (as used by \cite{PETS:ErnKoc21}). Using the LWE estimator, we set $\lambda:=2048$. $q$ to match the field used for Shamir's Secret Sharing, which is a 128-bit prime, and set $p=2^{53}$. The hardness estimated is $2^{129}$, i.e., we get the security of 129 bits. 





We use Packed Secret Sharing to benchmark the server and client computation cost. One can pack $\rho$ secrets into a single polynomial using packed secret sharing. However, an implicit trade-off exists as the total number of parties $\csize \geq 3/2\cdot \rho$. An increase in $\rho$ implies an increase in $\csize$. Here, we set $\rho=16$, corresponding to requiring $\lambda/16$ polynomials, i.e., each committee member receives 64 shares. Consequently, an $\csize=50$ satisfies the requirement. We then set the reconstruction threshold $\crec=34$, while Packed Secret Sharing guarantees a tolerance for a corruption threshold of $\csize-\rho$. However, we need to ensure this among the committee members, even in the face of dropouts. Specifically, if $\delta_C,\eta_C$ are dropout and corruption rates among the committee members, we need to compute the $\Pr[\delta_C+\eta_C>1/3]$. Using Hypergeometric Distribution (Definition~\ref{def:hypergeo}) and setting $\delta=0.01$, we get that $\csize=50$ ensures $\Pr[\delta_C+\eta_C>1/3]\leq 5\cdot 10^{-5}$. Note that this construction achieves committee performance independent of the vector length and is preferred. This also satisfies the constraint for reconstruction with error correction which requires that $2\csize\cdot \eta_C < (1-\delta_C) \csize - \rho+1$.

\paragraph{Microbenchmarking Secure Aggregation.} Our first series of experiments is to run $\caps_\LWR$ to build a secure aggregation protocol for $L=1$. We also compare with existing work, including \cite{bonawitz2017practical,bell2020secure,GPSBB22,SP:MWAPR23}\footnote{We do not benchmark LERNA~\cite{AC:LLPT23} given that it is not suitable for the FL setting.}. We vary the offline rates ($\delta$) and the ability to group clients ($\csize$), along with increasing the number of clients to study the performance of related work. It is important to note that $\caps$'s one-shot communication (client and committee member) implies that as $\delta$ increases, the client's performance remains the same while the server and the committee's performance improves as fewer clients participate in the aggregation. The results are plotted in Figure~\ref{fig:client}.


\paragraph{Performance of $\caps_\LWR$.} The running time performance of \(\caps_\LWR\) is notably superior to all other protocols tested across server and client computation times.
\begin{itemize}\itemsep0em
    \item \textbf{Server Total Time}: At \(1000\) clients, \(\caps_\LWR\) shows the most significant improvement, with a server computation time of just \(0.31\) seconds, drastically outperforming other protocols. For example, \cite{bonawitz2017practical} with \(\delta=0\) takes approximately \(71.3\) seconds, while \cite{bell2020secure} with \(\delta=0\) and \(\csize=50\) reaches \(10.1\) seconds. Even protocols from \cite{GPSBB22} with various settings, such as \(\delta=0\) and \(\delta=0.1\), take between \(16\) and \(19\) seconds. These results make \(\caps_\LWR\) an outstanding choice for minimizing server computation time.
    \item \textbf{Client Total Time}: Similarly, for client computation time at \(1000\) clients, \(\caps_\LWR\) again stands out with only \(38.8\) milliseconds, which is significantly faster than the other protocols. For instance, the protocol in \cite{bonawitz2017practical} with \(\delta=0\) reaches over \(5700\) milliseconds, while \cite{bell2020secure} with \(\delta=0, \csize=50\) results in \(146.2\) milliseconds. In contrast, \cite{GPSBB22}, even at lower offline rates (\(\delta=0\) and \(\delta=0.1\)), shows client computation times exceeding \(2000\) milliseconds, making \(\caps_\LWR\)'s performance at \(1000\) clients a clear advantage.
\end{itemize}
The remarkable speed of \(\caps_\LWR\) in server and client performance, especially as client count grows, emphasizes its efficiency and scalability in real-world applications. It significantly reduces the computational burden and communication overhead compared to existing protocols, making it an attractive solution for privacy-preserving aggregation tasks. 

\paragraph{Communication Cost of $\caps_\LWR$.} Let $k$ be the number of elements being shared. With naive secret sharing, this would be 1024. Instead, when we pack into 64 different polynomials, we get $k=64$. Thus, $\caps_\LWR$ has a communication cost, in terms of field elements, as follows: 
\begin{itemize}
    \item Total Sent/Received per Client: $L+k\cdot \csize/\bot$ field elements
    \item Total Sent/Received by Server: $nL+\csize \cdot n \cdot k/n+k$ field elements
    \item Total Sent/Received per Committee Member: $\log n\cdot k+k/\log n+ k$ field elements,
\end{itemize}

\ifdefined\IsFull
\else \paragraph{Benchmarking FL Models.} Due to space constraints, we refer readers to Section~\ref{sub:exp-ml}, which details the accuracy experiments conducted on various classifiers and datasets. Meanwhile, we plot the performance of $\caps_\LWR$ with respect to learning in the clear (i.e., without cryptographic protections) in Figure~\ref{fig:mcc}. Our results show that the accuracy of $\caps_\LWR$ is statistically indistinguishable from that of the classifier when learning in the clear.
\fi
\ifdefined\IsFull
\input{other-exp}
\fi

%% file: graph-figures/agg-graph-total.tex
\begin{figure*}[!tb]
\centering
\resizebox{\textwidth}{!}{
\begin{tikzpicture}

\begin{axis}[
    name=ax1,
    xlabel={\textbf{Number of Clients}},
    ylabel={\textbf{Total Running Time of Server (s)}},
    xmin=100, xmax=1000,
    ymin=0, ymax=25, 
    xtick={100,200,300,400,500,1000},
    ytick={5,10,15,20,25},
    legend style={at={(1,-0.3)}, anchor=north, legend columns=3,rounded corners=2pt, draw=gray!50, font=\small},
    legend cell align={left},
    tick label style={font=\footnotesize},
    label style={font=\bfseries},
    grid=major,
    grid style={dashed,gray!30},
    axis line style={thick,gray!70},
]

\addplot[smooth, thick, mark=*, mark options={scale=1.2,fill=white}, brown] plot coordinates {
(100,2.22210754)
(200,20.18119782)
(300,71.268430033)
(400,236.850921849)
(500,350.692884068)
};
\addlegendentry{\cite{bonawitz2017practical}, $\delta=0$};

\addplot[smooth, thick, mark=square, mark options={scale=1.2,fill=white}, brown] plot coordinates{
(100,2.189049914)
(200,17.626738704)
(300,60.535058121)
(400,143.704301151)
(500,287.997793231)
};
\addlegendentry{\cite{bonawitz2017practical}, $\delta=0.1$};

\addplot[smooth, thick, mark=*, mark options={scale=1.2,fill=white}, red] plot coordinates{
(100,0.58041229)
(200,1.258894755)
(300,2.029353019)
(400,2.773672778)
(500,3.67401742)
(1000,10.104327259)
};
\addlegendentry{\cite{bell2020secure}, $\delta=0, \csize=50$};

\addplot[smooth, thick, mark=square, mark options={scale=1.2,fill=white}, red] plot coordinates{
(100,1.881928585)
(200,4.18351766)
(300,6.319696365)
(400,8.667374274)
(500,10.981067395)
(1000,28.76563569)
};
\addlegendentry{\cite{bell2020secure}, $\delta=0.2, \csize=100$};

\addplot[smooth, thick, mark=*, mark options={scale=1.2,fill=white}, purple] plot coordinates{
(100,1.54882326)
(200,3.08866839)
(300,4.63147151666667)
(400,6.67002698)
(500,8.16298726)
(1000,19.047302576)
};
\addlegendentry{\cite{GPSBB22}, $\delta=0$};

\addplot[smooth, thick, mark=square, mark options={scale=1.2,fill=white}, purple] plot coordinates{
(100,1.50020380666667)
(200,3.09948501777778)
(300,4.7742555337037)
(400,6.42861823555556)
(500,8.31829591688889)
(1000,17.8900020293333)
};
\addlegendentry{\cite{GPSBB22}, $\delta=0.1$};

\addplot[smooth, thick, mark=o, mark options={scale=1.2,fill=white}, purple] plot coordinates{
(100,1.65940127)
(200,3.22059227)
(300,5.03315891666667)
(400,6.339051285)
(500,8.063034418)
(1000,16.03913433)
};
\addlegendentry{\cite{GPSBB22}, $\delta=0.1, \csize=50$};

\addplot[smooth, thick, mark=square, mark options={scale=1.2,fill=white}, magenta] plot coordinates{
(128,5.455737)
(256,9.713493)
(384,6.8808962)
(512,9.8518818)
(1000,16.5750486)
};
\addlegendentry{\cite{SP:MWAPR23}, $\csize=50,\delta=0$};

\addplot[smooth, thick, mark=ball, mark options={scale=1.2,fill=white}, blue] plot coordinates{
(100,0.30537408)
(200,0.307561519999999)
(300,0.05513964)
(400,0.05163816)
(500,0.06171612)
(1000,0.309924519999999)
};
\addlegendentry{$\caps_\LWR$, $\delta=0,\csize=50$};

\end{axis}

\begin{axis}[
    at={(ax1.south east)},
    xshift=2cm,
    xlabel={\textbf{Number of Clients}},
    ylabel={\textbf{Total Running Time of Client (ms)}},
    xmin=100, xmax=1000,
    ymin=0, ymax=500,
    xtick={100,200,300,400,500,1000},
    ytick={50,100,200,300,400,500},
    tick label style={font=\footnotesize},
    label style={font=\bfseries},
    grid=major,
    grid style={dashed,gray!30},
    axis line style={thick,gray!70},
]

\addplot[smooth, thick, mark=*, mark options={scale=1.2,fill=white}, brown] plot coordinates {
(100,78.367)
(200,183.80)
(300,358.80)
(400,646.56)
(500,1036.52)
(1000,5714.27)
};

\addplot[smooth, thick, mark=square, mark options={scale=1.2,fill=white}, brown] plot coordinates{
(100,80.926411)
(200,190.469919)
(300,377.026305)
(400,648.181712)
(500,1024.863026)
(1000,4473.49344044444)
};

\addplot[smooth, thick, mark=*, mark options={scale=1.2,fill=white}, red] plot coordinates{
(100,48.79917)
(200,55.102885)
(300,64.8912666666666)
(400,73.5988075)
(500,78.62315)
(1000,146.220567)
};

\addplot[smooth, thick, mark=square, mark options={scale=1.2,fill=white}, red] plot coordinates{
(100,78.83624)
(200,98.0515025)
(300,116.09904)
(400,135.137100625)
(500,154.083074499999)
(1000,247.31732725)
};

\addplot[smooth, thick, mark=*, mark options={scale=1.2,fill=white}, purple] plot coordinates{
(100,58.5905099999999)
(200,114.070695)
(300,203.933726666666)
(400,341.6750025)
(500,532.340492)
(1000,2119.96241)
};

\addplot[smooth, thick, mark=square, mark options={scale=1.2,fill=white}, purple] plot coordinates{
(100,58.981552222222)
(200,112.902698888889)
(300,203.769079259259)
(400,345.390370277778)
(500,529.657483111111)
(1000,2115.418889)
};

\addplot[smooth, thick, mark=o, mark options={scale=1.2,fill=white}, purple] plot coordinates{
(100,78.64522)
(200,88.21435)
(300,89.4430133333332)
(400,89.46912)
(500,90.193542)
(1000,90.73343)
};

\addplot[smooth, thick, mark=square, mark options={scale=1.2,fill=white}, magenta] plot coordinates{
(128,130.102799)
(256,275.559182)
(384,319.072848)
(512,425.929342)
(1000,658.105783)
};

\addplot[smooth, thick, mark=ball, mark options={scale=1.2,fill=white}, blue] plot coordinates{
(100,38.37018)
(200,39.400755)
(300,28.44341)
(400,28.285335)
(500,28.64661)
(1000,38.802717)
};

\end{axis}

\end{tikzpicture}}
\caption{We plot the total client and server running time as a function of client count across various protocols. Our running time includes both computation and communication time. Here, $\delta$ indicates the offline rate and $\csize$ is used to parametrize the number of parties each client has to communicate with (i.e., number of neighbors or committee size). In all our experiments on MicroSecAgg~\cite{GPSBB22}, the clients' maximum input is $10^4$. The running time of the committee members are added to the Client's running time in $\caps_\LWR$.} 
\label{fig:server}
\label{fig:client}
\end{figure*}

%% file: other-exp.tex
\ifdefined\IsFull
\else 
\section{More Experimental Results}
\label{sec:other-exp}
In this section, we look at other experimental benchmarking. 
\fi
\input{graph-figures/agg-graph}
\subsection{Computation Time}
\label{sub:exp-comp}
We also micro-benchmark in the same environment and setting as before to measure the cost of server and client computation alone. This is shown in Figure~\ref{fig:client-comp}. This does not account for the time of communication. As can be seen, $\caps_\LWR$ still significantly outperforms prior work. 

\subsection{Running Time vs length of vector $L$}
We also benchmark the performance of the protocols, with dropout rate $\delta=0$ for $L=1000$. This is covered in Figure~\ref{fig:client-1000}. The server performance of $\caps_\LWR$ significantly outperforms existing protocol. Meanwhile, we observe that Flamingo has a better client performance for smaller choices of client count, while $\caps_\LWR$ begins performing better for a larger number of clients. It is noted that $\caps_\LWR$'s unique one-shot design implies that the client performance is independent of the number of clients. Owing to the number of exponentiation that is proportional to $L$, MicroSecAgg~\cite{GPSBB22} performs degrades with $L$, and our benchmarks indicated that the performance bounds were outside the range of this graph and are omitted from our plot. 

\input{graph-figures/agg-graph-1000}
\subsection{Performance of $\caps_\cl$}
\label{sub:caps-cl-perf}
$\caps'$ is an alternate framework for instantiating \CAPS\  from threshold key-homomorphic PRF unlocking constructions based on additional assumptions, including class groups. 
While these contributions are practically more expensive, we include benchmarks for completeness. We use the Threshold PRF based on Class Groups (See Section~\ref{cons:caps-cl}). We rely on the BICYCL library~\cite{JC:BCIL22} and use pybind to convert the C++ code to Python. Our implementation will assume that the plaintext space is $\bbZ_p$ for a prime $p$. Our experiments will take that $\csize=50$. 

It is to be noted that the client's performance scales significantly, owing to 50 group exponentiations. However, the server performance is less than 2 seconds, outperforming several existing protocols' computations. 
\begin{table}[!tb]
\centering
\caption{Comparison of Computation Times for Server and Client between $\caps_\cl$ and $\caps_\LWR$ for $L=1$.}
\label{tab:computation_times}
\resizebox{!}{!}{
\rowcolors{2}{gray!10}{white}

  \begin{tabular}{c | c c |  c c}
        \toprule
        \textbf{Clients} & \multicolumn{2}{c|}{\textbf{Server Computation (s)}} & \multicolumn{2}{c}{\textbf{Client Computation (s)}} \\ 
        \cmidrule(lr){2-3} \cmidrule(lr){4-5}
        & $\caps_\LWR$  & $\caps_\cl$ & $\caps_\LWR$  & $\caps_\cl$ \\ 
        \midrule
        \textbf{100}  & 0.0262  & 1.9347  & 0.0276  & 1.6579  \\ 
        \textbf{200}  & 0.0268  & 1.9052  & 0.0286  & 1.6622  \\ 
        \textbf{300}  & 0.0260  & 1.9036  & 0.0273  & 1.6526  \\ 
        \textbf{400}  & 0.0266  & 1.9067  & 0.0273  & 1.6333  \\ 
        \textbf{500}  & 0.0267  & 1.9201  & 0.0273  & 1.6420  \\ 
        \textbf{1000} & 0.0292  & 1.9310  & 0.0280  & 1.6350  \\ 
        \bottomrule
    \end{tabular}

}
\end{table}
\subsection{Performance in Federated Learning Use-case}

\label{sub:exp-ml}
\ifdefined\IsFull
We discuss the accuracy experiments conducted on various classifiers and datasets and plot the performance of $\caps_\LWR$ with respect to learning in the clear (i.e., without cryptographic protections) in Figures~\ref{fig:mcc},\ref{fig:mlp_accuracy}. Our results show that the accuracy of $\caps_\LWR$ is statistically indistinguishable from that of the classifier when learning in the clear.
\fi
\begin{itemize}\itemsep0em
    \item Adult Census Dataset: We first run experiments on the adult census income dataset from \cite{ICAIF:,Jayaraman} to predict if an individual earns over \$50,000 per year. The preprocessed dataset has 105 features and 45,222 records with a 25\% positive class. We randomly split into training and testing, with further splitting by the clients. First, we train in the clear with weights sent to the server to aggregate. With 100 clients and 50 iterations, we achieve 82.85\% accuracy and 0.51 MCC. We repeat with $\caps_\LWR$ with 100 clients and 50 committee members. With 10 iterations, we achieve 82.38\% accuracy and 0.48 MCC. With 20 iterations, we achieve 82\% accuracy and 0.51 MCC. Our quantization technique divides weights into integer and decimal parts (2 integer and 8 decimal values per weight). Training with 50 clients takes under 1 minute per client per iteration with no accuracy loss. This quantization yields a vector size of 1050 (10 per feature).
    \item We use the Kaggle Credit Card Fraud dataset~\cite{Kaggle}, comprising 26 transformed principal components and amount and time features. We omit time and use the raw amount, adding an intercept. The goal is to predict if a transaction was indeed fraudulent or not. There are 30 features and 284,807 rows, with  $<$0.2\% fraudulent. Weights are multiplied by 10,000 and rounded to an integer, accounted for in aggregation. Figure \ref{fig:mcc} shows $\caps_\LWR$'s MCC versus clear learning for varying clients and iterations. With the accuracy multiplier, $\caps_\LWR$'s MCC is close to clear learning and sometimes outperforms. The highly unbalanced dataset demonstrates $\caps_\LWR$ can achieve substantial performance even in challenging real-world scenarios.
    \item We then train a vanilla multi-layer perceptron (MLP) classifier on three datasets: MNIST, CIFAR-10, and CIFAR-100. We quantize the weights by multiplying with $2^{16}$. The MLP accuracy, as a function of the iteration count, is plotted in Figures~\ref{fig:mlp_accuracy},\ref{fig:mcc}. Our experiments demonstrate that $\caps_\LWR$ preserves accuracy while ensuring the privacy of client data. Note that vanilla MLP classifiers do not typically offer good performance for CIFAR datasets, but note that our experiments aimed to show that $\caps_\LWR$ does not impact accuracy. 
\end{itemize}

\input{graph-figures/mlp-plot}

%% file: graph-figures/agg-graph.tex
\begin{figure*}[!tb]
\centering
\resizebox{\columnwidth}{!}{\begin{tikzpicture}
{\begin{axis}[
name=ax1,
    xlabel={Number of Clients},
    ylabel={Server Computation Time (s)},
    xmin=0, xmax=1000,
    ymin=0, ymax=110,
    xtick={100,200,300,400,500,1000},
    ytick={10,20,30,40,50,60,70,80,90,100,110},
            ]
\addplot[smooth,mark=*,brown] plot coordinates {
(100,2.204353)
(200,20.166329)
(300,71.249484)
(400,236.832975)
(500,350.675013)
};
\addplot[dashed,mark=x,brown] plot coordinates{
(100,1.91435)
(200,14.64972)
(300,49.959081)
(400,115.481861)
(500,228.532808)
}
;
\addplot[smooth,mark=o,red] plot coordinates{
(100,2.252246)
(200,5.003404)
(300,7.668508)
(400,22.312608)
(500,18.527716)
(1000,28.749468)
};
\addplot[dashed, mark=x,red] plot coordinates{
(100,1.864369)
(200,4.167624)
(300,6.303873)
(400,8.650845)
(500,10.964724)
(1000,23.040363)
};
\addplot[smooth, mark=*, purple] plot coordinates{
(100,1.537693)
(200,3.079383)
(300,4.621459)
(400,6.659946)
(500,8.152787)
(1000,19.037787)
};
\addplot[smooth, mark=o, purple] plot coordinates{
(100,1.490812)
(200,3.089455)
(300,4.765061)
(400,6.419494)
(500,8.308121)
(1000,17.880413)
};
\addplot[dashed, mark=o, purple] plot coordinates{
(100,1.574981)
(200,3.090684)
(300,4.830107)
(400,6.396579)
(500,7.617873)
(1000,15.336461)
};
\addplot[smooth, mark=square, magenta] plot coordinates{
(128,5.455737)
(256,9.713493)
(384,6.8808962)
(512,9.8518818)
(1000,16.5750486)
};

\addplot[double, mark=square, blue] plot coordinates{
(100,0.026)
(200,0.027)
(300,0.026)
(400,0.027)
(500,0.027)
(1000,0.029)
};
\end{axis}}
\begin{axis}[
at={(ax1.south east)},
        xshift=1.5cm,
    xlabel={Number of Clients},
    ylabel={Client Computation Time (ms)},
    xmin=0, xmax=1000,
    ymin=0, ymax=1100,
    xtick={100,200,300,400,500,1000},
    ytick={100,200,300,400,500,600,700,800,900,1000,1100},
    legend pos=outer north east,
    legend style={at={(0.5,-0.1)},
	anchor=north,legend columns=-1}
            ]
\addplot[smooth,mark=*,brown] plot coordinates {
(100,60.411)
(200,166.698585)
(300,343.14053)
(400,685.875)
(500,1019.111)
};
\addplot[dashed,mark=x,brown] plot coordinates{
(100,62.1171725)
(200,174.99744875)
(300,358.674363333332)
(400,633.701521875)
(500,980.982129)
}
;

\addplot[smooth,mark=o,red] plot coordinates{
(100,63.1444399999999)
(200,80.846)
(300,96.8828566666667)
(400,131.015685)
(500,183.513902)
(1000,230.217405999999)
};
\addplot[dashed, mark=x,red] plot coordinates{
(100,61.6356)
(200,81.28292875)
(300,100.2608125)
(400,138.259909375)
(500,197.617458499999)
(1000,250.71284075)
};
\addplot[smooth, mark=*, purple] plot coordinates{
(100,49.0913799999999)
(200,103.37621)
(300,193.365406666666)
(400,331.77315)
(500,522.55244)
(1000,2110.139399)
};
\addplot[smooth, mark=o, purple] plot coordinates{
(100,49.0786344444442)
(200,102.635808333333)
(300,193.940535925926)
(400,335.053793888889)
(500,519.855279777778)
(1000,2105.55258588889)
};
\addplot[dashed, mark=o, purple] plot coordinates{
(100,119.792054444444)
(200,141.352514444444)
(300,162.34831962963)
(400,163.938735277778)
(500,164.763129777777)
(1000,165.916322222222)
};
\addplot[smooth, mark=square, magenta] plot coordinates{
(128,130.102799)
(256,275.559182)
(384,319.072848)
(512,425.929342)
(1000,658.105783)
};

\addplot[double, mark=square, blue] plot coordinates{
(100,27.55)
(200,28.518)
(300,27.242)
(400,27.243)
(500,27.2229)
(1000,27.925)
};
\end{axis}
\end{tikzpicture}}\hfill
\resizebox{0.2\textwidth}{!}{\begin{tikzpicture}
    \begin{customlegend}[legend entries={{\cite{bonawitz2017practical} $,\eta=0$},{\cite{bonawitz2017practical}, $\eta=0.2$},{\cite{bell2020secure}, $\eta=0,\csize=100$},{\cite{bell2020secure}, $\eta=0.2, \csize= 100$}}]
    \addlegendimage{smooth,mark=*,brown}
    \addlegendimage{dashed,mark=x,brown}
    \addlegendimage{smooth,mark=o,red}
    \end{customlegend}
\end{tikzpicture}}
\resizebox{0.2\textwidth}{!}{\begin{tikzpicture}
    \begin{customlegend}[legend entries={{\cite{bell2020secure}, $\eta=0.2, \csize= 100$},{\cite{GPSBB22}, $\eta=0$},{\cite{GPSBB22}, $\eta=0.1$}}]
    \addlegendimage{dashed, mark=x,red}
    \addlegendimage{smooth, mark=*, purple}
    \addlegendimage{smooth, mark=o, purple}
    \end{customlegend}
\end{tikzpicture}}
\resizebox{0.2\textwidth}{!}{\begin{tikzpicture}
    \begin{customlegend}[legend entries={{\cite{GPSBB22}, $\eta=0.1,\csize=100$},{\cite{SP:MWAPR23}, $\eta=0$},{$\caps_\LWR$, $\delta=0,\csize=50$}}]
    \addlegendimage{dashed, mark=o, purple}
    \addlegendimage{smooth, mark=square, magenta}
    \addlegendimage{double, mark=star, blue}
    \addlegendimage{double, mark=square, blue}
    \end{customlegend}
\end{tikzpicture}}
\caption{Client and Server Computation Time as a function of client count across different algorithms.}
\label{fig:server-comp}
\label{fig:client-comp}
\end{figure*}

%% file: graph-figures/agg-graph-1000.tex
\begin{figure*}[!tb]
\centering
\resizebox{\columnwidth}{!}{
\begin{tikzpicture}

\begin{axis}[
    name=ax1,
    xlabel={\textbf{Number of Clients}},
    ylabel={\textbf{Total Running Time of Server (s)}},
    xmin=100, xmax=1000,
    ymin=0, ymax=250, 
    xtick={100,200,300,400,500,1000},
    ytick={50,100,150,200,250},
    legend style={at={(1,-0.3)}, anchor=north, legend columns=4,rounded corners=2pt, draw=gray!50, font=\small},
    legend cell align={left},
    tick label style={font=\footnotesize},
    label style={font=\bfseries},
    grid=major,
    grid style={dashed,gray!30},
    axis line style={thick,gray!70},
]

\addplot[thick, mark=*, mark options={scale=1.2,fill=white}, brown] plot coordinates {
(100,17.3739416500001)
(200,59.4470552499989)
(300,150.510430759989)
(400,295.376346170023)
(500,515.422273549964)
};
\addlegendentry{\cite{bonawitz2017practical}, $\delta=0$};

\addplot[thick, mark=*, mark options={scale=1.2,fill=white}, red] plot coordinates{
(100,8.47549649999998)
(200,16.99220663)
(300,24.3395425399999)
(400,42.30421542)
(500,45.7882414599998)
(1000,3428.44528893008)
};
\addlegendentry{\cite{bell2020secure}, $\delta=0, \csize=50$};


\addplot[thick, mark=square, mark options={scale=1.2,fill=white}, magenta] plot coordinates{
(128,2.746863666)
(256,5.070136332)
(512,8.468965)
(1000,16.7576)
};
\addlegendentry{\cite{SP:MWAPR23}, $\csize=50,\delta=0$};

\addplot[thick, mark=ball, mark options={scale=1.2,fill=white}, blue] plot coordinates{
(100,0.30537408)
(200,0.307561519999999)
(300,0.05513964)
(400,0.05163816)
(500,0.06171612)
(1000,0.309924519999999)
};
\addlegendentry{$\caps_\LWR$, $\delta=0,\csize=50$};

\end{axis}

\begin{axis}[
    at={(ax1.south east)},
    xshift=2cm,
    xlabel={\textbf{Number of Clients}},
    ylabel={\textbf{Total Running Time of Client (s)}},
    xmin=100, xmax=1000,
    ymin=0, ymax=2,
    xtick={100,200,300,400,500,1000},
    ytick={0.2,0.4,0.6,0.8,1,1.2,1.4,1.6,1.8,2},
    tick label style={font=\footnotesize},
    label style={font=\bfseries},
    grid=major,
    grid style={dashed,gray!30},
    axis line style={thick,gray!70},
]

\addplot[thick, mark=*, mark options={scale=1.2,fill=white}, brown] plot coordinates {
(100,3.58670212)
(200,6.99399983)
(300,10.93254365)
(400,14.93650793)
(500,18.96236001)
};

\addplot[thick, mark=*, mark options={scale=1.2,fill=white}, red] plot coordinates{
(100,1.83269845)
(200,1.82649965)
(300,1.83468779)
(400,1.89630629)
(500,1.84857408)
(1000,37.03698114)
};


\addplot[thick, mark=square, mark options={scale=1.2,fill=white}, magenta] plot coordinates{
(128,0.159602937)
(256,0.276509861)
(512,0.39406)
(1000,0.732817)
};

\addplot[thick, mark=ball, mark options={scale=1.2,fill=white}, blue] plot coordinates{
(100,0.38721672)
(200,0.39382147)
(300,0.37564388)
(400,0.38111999)
(500,0.38155926)
(1000,0.39875789)
};

\end{axis}

\end{tikzpicture}}
\caption{We plot the total client and server running time as a function of client count across various protocols. Our running time includes both computation and communication time. Here, $\delta=0$ indicates the offline rate and $\csize$ is used to parametrize the number of parties each client has to communicate with (i.e., number of neighbors or committee size) and $L=1000$. In all our experiments on MicroSecAgg~\cite{GPSBB22}, the clients' maximum input is $10^4$. The running time of the committee members are added to the Client's running time in $\caps_\LWR$.}
\label{fig:server-1000}
\label{fig:client-1000}
\end{figure*}

%% file: graph-figures/mlp-plot.tex
\begin{figure*}[!tb]
\centering
\resizebox{\textwidth}{!}{
\begin{tikzpicture}
    \begin{axis}[
        name=mnist,
        title={MNIST},
        title style={
                at={(axis description cs:0.5,0.9)}, 
                anchor=north},
        xlabel={Iterations},
        ylabel={MLP Accuracy},
        xmin=0, xmax=30,
        ymin=0.88, ymax=1,
        xtick={10,20,30},
        xticklabels={10,20,30},
        ytick={0.88, 0.90, 0.92, 0.94, 0.96, 0.98, 1.00},
        legend pos=south east,
        ymajorgrids=true,
        grid style=dashed,
    ]
    \addplot[
        color=red,
        mark=none,
        ]
        coordinates {
        (0, 0.8848800690249313) (5, 0.9703982958881384)(10, 0.949543233932329) (15, 0.9704503296316472) (20, 0.9711448703553986) (25, 0.961593629740354) (29, 0.97248235115229)

        };
    \addlegendentry{$\caps$}
    \addplot[
        color=blue,
        mark=none,
        ]
        coordinates {
        (0, 0.885)
        (5, 0.950)
        (10, 0.968)
        (15, 0.956)
        (20, 0.961)
        (25, 0.966)
        (29, 0.970)
        };
        \addlegendentry{Clear}
    \end{axis}
\end{tikzpicture}

\begin{tikzpicture}
    \begin{axis}[
        name = {cifar10},
        at={($ (mnist.east) + (1cm,0) $)},
            anchor=west,
        title={CIFAR-10},
        title style={
                at={(axis description cs:0.5,0.9)}, 
                anchor=north},
        xlabel={Iterations},
        ylabel={MLP Accuracy},
        xmin=0, xmax=30,
        ymin=0.1, ymax=0.7,
        xtick={10,20,30},
        xticklabels={10,20,30},
        ytick={0.1,0.2,0.3,0.4,0.5,0.6,0.7},
        yticklabels={0.1,0.2,0.3,0.4,0.5,0.6},
        legend pos=south east,
         ymajorgrids=true,
        grid style=dashed,
    ]
 \addplot[
        color=red,
        mark=none,
        ]
        coordinates {
     (0, 0.26572222222222222 )
(1, 0.47908888888888889)
(5, 0.49925555555555556)
(10, 0.45961111111111111)
(15,0.46833333333333334)
(20, 0.49788888888888887)
(25,0.4871111111111111)
(29,0.497)
        };
    \addlegendentry{$\caps$}
    \addplot[
        color=blue,
        mark=none,
        ]
        coordinates {
        (0, .25)(1, .48)(5, .49)(10, .46)(15, .48)
       (20, 0.47)(25, .48)
        (29, 0.50)
        };
    \addlegendentry{Clear}
    \end{axis}
\end{tikzpicture}
}
\caption{MLP Accuracy for different datasets: MNIST and CIFAR-10. We use $n=100$ clients in this training, and $\csize=50$. }
\label{fig:mlp_accuracy}
\end{figure*}
\begin{figure*}[!tb]
\centering
\resizebox{\textwidth}{!}{
\begin{tikzpicture}
    \begin{axis}[
        name = {cifar100},
        at={($ (mnist.south) + (1cm,0) $)},
            anchor=west,
        title={CIFAR-100},
        title style={
                at={(axis description cs:0.5,0.9)}, 
                anchor=north},
        xlabel={Iterations},
        ylabel={MLP Accuracy},
        xmin=0, xmax=30,
        ymin=0.0, ymax=0.4,
        xtick={10,20,30},
        xticklabels={10,20,30},
        ytick={0.1,0.2,0.3,0.4},
        yticklabels={0.1,0.2,0.3,0.4},
        legend pos=south east,
         ymajorgrids=true,
        grid style=dashed,
    ]
    
        \addplot[smooth,
        color=red,
        mark=none,
        ]
        coordinates {
        (0, 0.091)
        (5, 0.155)
        (10, 0.218)
        (15, 0.255)
        (20, 0.313)
        (25, 0.326)
        (29, 0.35)
        };
    \addlegendentry{$\caps$}
    \addplot[
        color=blue,
        mark=none,
        ]
        coordinates {
        (0, 0.07)
        (5, 0.1432)
        (10, 0.2123)
        (15, 0.2679)
        (20, 0.352)
        (25, 0.34)
        (29, 0.353)
        };
    \addlegendentry{Clear}
    \end{axis}
\end{tikzpicture}
\begin{tikzpicture}
        \begin{axis}[
                        xlabel={Iterations},
                        ylabel={MCC},
                        xmin=0, xmax=30,
                        ymin=0.06, ymax=0.24,
                        xtick={10,20,30},
                         title={Credit Card Fraud Dataset},
                        xticklabels={10,20,30},   
                        ytick={0.1,0.12,0.14,0.16,0.18,0.2,0.22,0.24},
                        legend pos=south east, ymajorgrids=true, grid style=dashed,
                    ]
\addplot[smooth,blue] plot coordinates {
(1,0.199035421011014)
(5,0.199035421011014)
(10,0.178021442325501)
(15,0.178021442325501)
(20,0.178021442325501)
(25,0.199035421011014)
(30,0.199035421011014)
};
\addlegendentry{Clear}

\addplot[smooth,color=red,]
    plot coordinates {
     (1,0.1888029824116)
(5,0.1888029824116)
(10,0.168869335284846)
(15,0.168869335384846)
(20,0.168869335284846)
(25,0.168869335284846)
(30,0.206824757269518)
    };
\addlegendentry{$\caps$}
\end{axis}
\end{tikzpicture}
}
\caption{MLP Accuracy for CIFAR-100. We use $n=100$ clients in this training, and $\csize=50$. In the last figure, we have the MCC score as a function of the number of iterations for the credit card fraud dataset.}
\label{fig:mcc}
\end{figure*}

%% file: conclusion.tex
\section{Future Work}

$\mathsf{OPA}$ offers malicious security with abort, yet achieving guaranteed output delivery (GoD) or robustness while ensuring that clients communicate only once remains an unresolved challenge. The first work addressing GoD is presented in \cite{ACORN}, which necessitates multiple rounds of interaction, scaling logarithmically with the number of clients in the worst-case scenario, with a minimum of 6 rounds. The recent advancement by \cite{armadillo} enhances the approach of \cite{ACORN} by introducing a more efficient 3-round protocol.

While we present lattice-based instantiations to achieve security in the face of malicious clients, we leave it as future work to implement and optimize the performance of these constructions. Another direction of research is to balance lattice-based zero-knowledge proofs with techniques from Bulletproof to achieve optimized performance.

%% file: neurips-app.tex
\section{Preliminaries}
\label{sec:crypto-app}
For completeness,
we discuss secret sharing in Section~\ref{sec:secret}. We discuss pseudorandom functions in Section~\ref{sub:prf}. We then introduce lattice-based cryptographic assumptions in Section~\ref{sub:lattice}.
\input{secret}
\input{ss}
\input{khprf}
\subsection{Distributed Key Homomorphic PRF}
\label{sub:dkhprf}
\input{dkhprf}
\input{class}
\input{cl-prf}
\input{tprf-cons}
\input{fl}
\input{lat}
\input{lattice}
\input{deferred-proofs}
\input{deferred-prfs}
\input{byzantine}
\input{app-malicious}
\input{stronger}

%% file: secret.tex
\subsection{Secret Sharing}
\label{sec:secret}
\label{sub:secret}
A key component of threshold cryptography is the ability to compute distributed exponentiation by sharing a secret. More formally, the standard approach is to compute $g^s$ for some $g\in\bbG$ where $\bbG$ is a finite group and $\secret$ is a secret exponent that has been secret-shared among multiple parties. This problem is much simpler when you assume that the group order is a publicly known prime $\Prime$ which then requires you to share the secret over the field $\bbZ_\Prime$. This was the observation of Shamir~\cite{Shamir79} whereby a secret $\secret$ can be written as a linear combination of $\sum_{i\in \cS} \alpha_i \secret_i\bmod\Prime$ where $\cS$ is a set of servers that is sufficiently large and holds shares of the secret $\secret_i$ and $\alpha_i$ is only a function of the indices in $\cS$. It follows that if each server provides $g_i=g^{\secret_i}$, then one can compute $g^s=g^{\sum_{i\in\cS} \alpha_i\cdot \secret_i}=\prod_{i\in\cS} g_i^{\alpha_i}$. Formally, this is defined below. 

\begin{construction}[Shamir's Secret Sharing over $\bbF_q$]
\label{cons:ssf}
Consider the following $(\cthr,\crec,\csize)$ Secret Sharing Scheme where $\csize$ is the total number of parties, $\cthr$ is the corruption threshold, $\crec$ is the threshold for reconstruction. Then, we have the following scheme:

~~
\begin{itemize}
    \item $\term{Share}(\secret,\cthr,\crec,\csize)$: Sample a random polynomial $f(X)\in \bbF_q[X]$ of degree $\crec-1$ such that $f(0)=\secret$. Then, \Return $\set{\secret^{(j)}:=f(j)}_{j\in[\csize]}$
    \item $\term{Coeff}(\cS)$: On input of a set $\cS=\set{i_1,\ldots,i_{\crec},\ldots}\subseteq [\csize]$ of at least $\crec$ indices, compute $\lambda_{i_j}=\prod_{\zeta\in[\crec]\setminus\set{j}}\frac{i_\zeta}{i_\zeta-i_j}$. Then, \Return $\set{\lambda_{i_j}}_{i_j\in\set{i_1,\ldots,i_\crec}}$ 
    \item $\term{Reconstruct}(\set{s^{(j)}}_{j\in\cS}$: If $|\cS|\geq \crec$, then output $\sum_{j\in\cS} \lambda_j\cdot \secret^{(j)}$ where $\set{\lambda_j}\gets\term{Coeff}(\cS)$.
    \end{itemize}
\end{construction}
The correctness of the scheme guarantees that the secret $\secret$ is correctly reconstructed. 
\newcommand{\pos}{\term{pos}}
\begin{construction}[Packed Secret Sharing over $\bbF_q$]
\label{cons:pssf}
Consider the following $(\cthr,\crec,\csize)$ Secret Sharing Scheme where $\csize$ is the total number of parties, $\cthr$ is the corruption threshold, $\crec$ is the threshold for reconstruction. Further, let $\rho$ be the number of secrets being packed which are to be embedded at points $\pos_1,\ldots,\pos_\rho$ where $\pos_i=\csize+i$. Here, $\cthr:=\crec-\rho$. Then, we have the following scheme:

    \centering
 \begin{pchstack}[center,space=2\fboxsep]
\resizebox{0.7\columnwidth}{!}{\procedure{$\lshare(\bfs=(\secret_1,\ldots,\secret_\rho),\cthr,\crec,\csize)$}{
         (\coeff_0,\ldots,\coeff_{\crec-\rho-1})\getsr \bbF_q\\
         q(X):=\sum_{i=0}^{\crec-\rho-1} X^{i}\cdot \coeff_i\\
         \pos_i=\csize+i~\cFor i=1,\ldots,\rho\\
         \cFor i\in[\rho]~\cDo\\
        \pcind L_i(X):=\prod_{j\in[\rho]\setminus i}  \frac{X-\pos_j}{\pos_i-\pos_j}\cdot \Delta\\
         f(X):=q(X)\prod_{i=1}^{\rho} (X-\pos_i)+\sum_{i=1}^{\rho}{\secret}_i\cdot L_i(X)\\
          \pcreturn \set{\secret^{(i)}}_{i\in[\csize]}
        }
        
         \procedure{$\lcomb(\set{\secret^{(i)}}_{i\in\cS})$}{
            \pcif |\cS|< \crec ~\pcreturn \bot\\
            \text{Parse}~\cS:=\set{i_1,\ldots,i_{\cthr},\ldots}\\
            \cFor k\in [\rho]\\
            \pcind \cFor j\in[\cthr]\\
            \pcind\pcind \Lambda_{i_j}(X):=\prod_{\zeta\in[\cthr]\setminus{j}} \frac{i_\zeta-X}{i_\zeta-i_j}\\
             \pcind \pcind \secret_k':=\sum_{j\in[t]} \Lambda_{i_j}(m+k)\cdot \secret^{(j)}\\
             \pcind \pcreturn \bfs':=(\secret_1',\ldots,\secret_\rho')\\
            }
            }
        \end{pchstack}
\end{construction}
\paragraph{Parameters.}\begin{itemize}
    \item For reconstruction, we require $\crec<\csize(1-\delta_C)$ where $\delta_C$ is the dropout rate within the committee.
    \item For security, we require that $(\crec-\rho)>\csize\cdot\eta_C$ where $\eta_C$ is the corruption rate. 
\end{itemize}
Combining, we get $\csize>\rho/(1-\delta_C-\eta_C)$. Recall that we need $\delta_C+\eta_C<1/3$, for byzantine fault tolerance. In other words, setting $\csize\geq 3\rho/2$ is sufficient. 
\begin{remark}[Optimizations for Packed Secret Sharing]
    Observe that the polynomial $L_i(X)$ is only dependent on points $\pos_i$, which are the points where the secrets are embedded. This can be pre-processed, and indeed, can be a part of the setup algorithm which distributes it to all the clients. Furthermore, rather than naively reconstructing the Lagrange polynomial, one can also rely on FFT techniques to achieve speed up. 
\end{remark}
Unfortunately, the above protocols do not extend to settings where the order of the group is not prime, not publicly known, or even possibly unknown to everyone. In this setting, the work of Damg{\aa}rd and Thorbek presents a construction to build Linear Integer Secret Sharing (LISS) schemes. In this work, we rely on the simpler scheme that extends Shamir's secret sharing into the integer setting from the work of Braun~\etal\cite{C:BraDamOrl23}. We also extend the Packed Secret Sharing scheme to this integer setting in Construction~\ref{cons:pss}. 
\begin{definition}[Secret Sharing over $\bbZ$]
A $(\cthr,\crec,\csize)$ Linear Integer Secret Sharing Scheme  $\liss$ is a tuple of PPT algorithms $\liss:=(\lshare,\lcoeff,\lcomb)$, with the following public parameters: the statistical security parameter $\kappa_s$, the number of parties $\csize$, the corruption threshold $\cthr$, and reconstruction threshold $\crec$ of secrets needed for reconstruction, the randomness bit length $\rsize$, the bit length of the secret $\ssize$ and the offset by which the secret is multiplied, denoted by $\offset=\csize!$, and the following syntax:
\begin{itemize}
    \item $(\secret_1,\ldots,\secret_\csize)\getsr\lshare(\secret,\csize,\crec,\cthr)$: On input of the secret $\secret$, the number of parties $\csize$, and the threshold $\cthr$, the share algorithm outputs shares $\secret_1,\ldots,\secret_\csize$ such that party $i$ receives $\secret_i$. 
    \item $\set{\lambda_{i}}_{i\in\cS}\gets\lcoeff(\cS)$: On input of a set $\cS$ of at least $\crec$ indices, the $\lcoeff$ algorithm outputs the set of coefficients for polynomial reconstruction. 
    \item $\secret'\gets\lcomb(\set{\secret_i}_{i\in\cS})$: On input of a set of secrets of at least $\crec$ shares, the reconstruction algorithm outputs the secret $\secret'$. 
\end{itemize}
We further require the following security properties. 
\begin{itemize}
    \item Correctness: For any $\csize,\kappa_s,\cthr,\crec,\ssize,\rsize\in\bbZ$ with $\cthr<\crec\leq\csize$, and any set $\cS\subseteq[m]$ with $|\cS|\geq \crec$, for any $\secret\in\bbZ$ such that $\secret\in[0,2^{\ssize})$ the following holds:
    \[
    \Pr\left[
    \secret'=f(\secret)~~\begin{array}{|c}
        (\secret_1,\ldots,\secret_m)\getsr\lshare(\secret,\csize,\crec,\cthr)\\
        \secret'\gets\lcomb(\set{\secret_i}_{i\in\cS})
    \end{array}
    \right]
    \]
    where $f$ is some publicly computable function, usually $f(\secret)=\offset^2\cdot \secret$.
    \item Statistical Privacy~\cite{PKC:DamTho06}: We say that a $(\cthr,\crec,\csize)$ linear integer secret sharing scheme $\liss$ is statistically private if for any set of corrupted parties $\cC\subset [\csize]$ with $|\cC|\leq \cthr$, and any two secrets $\secret,\secret'\in[0,2^{\ssize})$ and for independent random coins $\rho,\rho'$ such that $\set{\secret_i}_{i\in[\csize]}\getsr\lshare(\secret;\rho)$, $\set{\secret_i}'_{i\in[\csize]}\getsr\lshare(\secret';\rho')$ we have that the statistical distance between: $\set{\secret_i|i\in\cC}$ and $\set{\secret_i'|i\in\cC}$ is negligible in the statistical security parameter $\kappa_s$. 
\end{itemize}
\end{definition}


\begin{construction}[Shamir's Secret Sharing over $\bbZ$]
\label{cons:ss}
    Consider the following $(\cthr,\crec,\csize)$ Integer Secret Sharing scheme where $\csize$ is the number of parties, $\cthr$ is the corruption threshold, and $\crec$ is the threshold for reconstruction. Further, let $\kappa_{s}$ be a statistical security parameter. Let $\ssize$ be the bit length of the secret and let $\rsize$ be the bit length of the randomness. Then, we have the following scheme:
    \begin{center}
\begin{pchstack}[center,space=0.1cm]
    	\resizebox{0.7\columnwidth}{!}{     \procedure{$\lshare(\secret,\cthr,\crec,\csize)$}{
          \offset:=\csize!,\tilde{\secret}:=\secret \cdot \offset\\
         (\coeff_1,\ldots,\coeff_{\crec-1})\getsr [0,2^{\rsize+\kappa_s})\\
         f(X):=\tilde{\secret}+\sum_{i=1}^{\crec-1} \coeff_i\cdot X^{i}\\
          \pcreturn \set{\secret^{(i)}=f(i)}_{i\in[\csize]}
        }
        \procedure{$\lcoeff(\cS)$}{
             \pcif |\cS|\geq \crec\\
             \pcind \pcfor i\in\cS \pcdo\\
             \pcind \pcind \Lambda_i:=\prod_{j\in \cS\setminus\{i\}} \frac{x_j}{x_j-x_i}\cdot \offset\\
             \pcind \pcreturn \set{\Lambda_{i}}_{i\in\cS}\\
            }
                    \procedure{$\lcomb(\set{\secret^{(i)}}_{i\in\cS})$}{
             \pcif |\cS|\geq \crec\\
             \pcind \set{\Lambda_{i}}_{i\in\cS}\gets\lcoeff(\cS)\\
             \pcind \secret':=\sum_{i\in\cS} \Lambda_i\cdot \secret^{(i)}\\
             \pcind \pcreturn \secret'\\
            }}
        \end{pchstack}
        \end{center}
\end{construction}
We omit the proof of correctness as it is similar to the original Shamir's Secret Sharing scheme. However, we highlight the critical differences:
\begin{itemize}
    \item Unlike Shamir's Secret Sharing over fields, the secret here is already multiplied by the offset $\Delta$. Therefore, any attempt to reconstruct can only yield $\secret \cdot \Delta$ 
    \item However, note that the inverse of $x_j-x_i$ which was defined over the field $\bbZ_q$ might not exist or be efficiently computable in a field of unknown order. Instead, we multiply the Lagrange coefficients by $\offset$. Consequently, the reconstruction yields $\offset\cdot \tilde{\secret}$ which equals $\secret\cdot \offset^2$. 
\end{itemize} 
\begin{theorem}[\cite{C:BraDamOrl23}]
\label{thm:ss}
    Construction~\ref{cons:ss} is statistically private provided $\rsize\geq \ssize+\lceil\log_2(h_{\max}\cdot (\cthr-1))\rceil+1$ where $h_{\max}$ is an upper bound on the coefficients of the sweeping polynomial.
    \end{theorem}
    
We refer the readers to the proof in \cite[\S B.1]{C:BraDamOrl23}. The key idea behind the proof is first to show that there exists a ``sweeping polynomial'' such that at each of the points that the adversary has a share of, the polynomial evaluates to 0 while at the point where the secret exists, it contains the offset $\offset$. Implicitly, one can add the sweeping polynomial to the original polynomial whereby the sweeping polynomial "sweeps" away the secret information that the adversary has gained knowledge of. Meanwhile, in the later section, we present the proof for the generic construction that uses Shamir's Packed Secret Sharing over the integer space. This again uses the idea of a sweeping polynomial.

%% file: ss.tex
\begin{construction}[Shamir's Packed Secret Sharing over $\bbZ$]
\label{cons:pss}
    Let $\csize$ be the number of parties and $\rho$ be the number of secrets that are packed in one sharing. Further, let $\cthr$ denote the threshold for reconstruction (implies that corruption threshold is $\cthr-\rho$). Then,
    consider the following $(\csize,\cthr,\rho)$ Integer Secret Sharing Scheme with system parameters $\kappa_s$ as the statistical security parameter, $\ssize$ is the bit length of the a secret, and let $\rsize$ be the bit length of the randomness. Then, we have the following scheme:

    \centering
	\begin{pchstack}[center,space=0.1cm]
      \resizebox{0.7\columnwidth}{!}{  \procedure{$\term{Packed}\lshare(\bfs=(\secret_1,\ldots,\secret_\rho),\cthr,\csize)$}{
          \offset:=\csize!,\Tbs:=\bfs \cdot \offset\\
         (\coeff_0,\ldots,\coeff_{\cthr-\rho-1})\getsr [0,2^{\rsize+\kappa_s})\\
         q(X):=\sum_{i=0}^{\cthr-\rho-1} X^{i}\cdot \coeff_i\\
         \pos_i=\csize+i~\cFor i=1,\ldots,\rho\\
         \cFor i\in[\rho]~\cDo\\
        \pcind L_i(X):=\prod_{j\in[\rho]\setminus i}  \frac{X-\pos_j}{\pos_i-\pos_j}\cdot \Delta\\
         f(X):=q(X)\prod_{i=1}^{\rho} (X-\pos_i)+\sum_{i=1}^{\rho}\tilde{\secret}_i\cdot L_i(X)\\
          \pcreturn \set{\secret^{(i)}}_{i\in[\csize]}
        }
                    \procedure{$\lcomb(\set{\secret^{(i)}}_{i\in\cS})$}{
            \pcif |\cS|< \cthr ~\pcreturn \bot\\
            \text{Parse}~\cS:=\set{i_1,\ldots,i_{\cthr},\ldots}\\
            \cFor k\in [\rho]\\
            \pcind \cFor j\in[\cthr]\\
            \pcind\pcind \Lambda_{i_j}(X):=\prod_{\zeta\in[\cthr]\setminus{j}} \frac{i_\zeta-X}{i_\zeta-i_j}\cdot(\Delta)\\
             \pcind \pcind \secret_k':=\sum_{j\in[t]} \Lambda_{i_j}(m+k)\cdot \secret^{(j)}\\
             \pcind \pcreturn \bfs':=(\secret_1',\ldots,\secret_\rho')\\
            }
            }
        \end{pchstack}
\end{construction}
\paragraph{Correctness.}
    Observe that for all $i=1,\ldots,\rho$, we have the following:
    \begin{itemize}
        \item $L_i(\pos_i)=\offset$
        \item $L_j(\pos_i)=\offset$ for all $j\in[\rho],j\neq i$
        \item $f(\pos_i)=\Ts_i\cdot \Delta=s_i\cdot\Delta^2$
    \end{itemize} 
    Meanwhile,
     for $\lambda_{i_j}(X):=\prod_{\zeta\in[\cthr]\setminus[j]} \frac{i_\zeta-X}{i_\zeta-i_j}$,
     the polynomial we will be able to compute the polynomial $f(x)=\sum_{j\in[t]} \lambda_{i_j}\cdot s^{(j)}$ by correctness of Lagrange Interpolation. Consequently, $f(\pos_i)$ would return $\secret_i\cdot \Delta^2$. However, we compute $\Lambda_{i_j}$ instead, by multiplying with $\offset$ to remove need for division. Consequently, the resulting polynomial has $\Delta$ multiplied throughout yielding a $\Delta^{3}$ as the total offset. 
\begin{definition}[Vector of Sweeping Polynomials]
\label{def:sp}
Let $\cC\subset [\csize]$ such that $|\cC|=\cthr-\rho$. Then, we have a vector of sweeping polynomials, denoted by $\bsp_\cC(X)=(\sp_{1,\cC},\ldots,\sp_{\rho,\cC})$ where $\sp_{i,\cC}(X):=\sum_{j=0}^{\cthr-\rho} \sp_{i,j}\cdot X^j\in\bbZ[X]_{\leq \cthr-1}$ is the unique polynomial whose degree is at most $\cthr-1$ such that $\sp_{i,\cC}(\csize+i)=\Delta^2$, $\sp_{i,\cC}(\csize+j)=0$ for $j\in[\rho],j\neq i$, and $\sp_{i,\cC}(j)=0$ for all $j\in\cC$. Further, one can define $\sp_{\max}$ as the upper bound for the coefficients for the sweeping polynomials, i.e., $\sp_{max}:=\set{\sp_{i,j}|i\in \set{1,\ldots,\rho},j\in\set{0,\ldots,\cthr-1}}$
\end{definition}
\begin{lemma}[Existence of Sweeping Polynomial]
\label{lem:sp}
    For any $\cC\subset [\csize]$ with $|\cC|=\cthr-\rho$, there exists $\bsp_{\cC}\in(\bbZ[X]_{\leq \cthr-\rho})^{\rho}$ satisfying Definition~\ref{def:sp}. 
\end{lemma}
\begin{proof}
    For any $i=1,\ldots,\rho$, we have that $\sp_{i,\cC}(\csize+i)=\Delta^2$ and $\sp_{i,\cC}(j)=0$ for $j\in\cC$. Let $\cC:=(i_1,\ldots,i_{\cthr-\rho})$. In other words, we can use these evaluations to construct a polynomial as follows:
    \[
    \sp_{i,\cC}(X):=\Delta^2\cdot \prod_{j=1}^{\cthr-\rho} \frac{(X-i_j)}{(\csize+i)-i_j}\cdot \prod_{j\in[\rho]\setminus\set{i}} \frac{(X-(\csize+j))}{(i-j)}
    \]
    
    Note that $i_1,\ldots,i_j\in[\csize]$ and are distinct. Therefore, $\prod_{j=1}^{\cthr-\rho} (\csize+i)-i_j$ perfectly divides $\Delta$ and so does $\prod_{j\in[\rho]\setminus\set{i}} (i-j)$, which implies that the coefficients are all integers. Further, the degree of this polynomial is at most $\cthr-1$. Thus, $\sp_{i,\cC}(X)\in\bbZ[X]_{\cthr-1}$. This defines the resulting vector of sweeping polynomials $\bsp_{\cC}$. 
\end{proof}
\begin{restatable}{theorem}{pss}
\label{thm:pss}
  Construction~\ref{cons:pss} is statistically private provided $$\rsize\geq \ssize+\ceil{\log_2(\sp_{\max}\cdot (t-1)\cdot \rho)}+1$$
\end{restatable}
\begin{proof}
   Let $\bfs,\bfs'\in [0,2^{\ssize})^{\rho}$ be two vectors of secrets. Then, $\Tbs:=\bfs\cdot \Delta$ and $\Tbs':=\bfs'\cdot \Delta$. Let $\cC$ denote an arbitrary subset of corrupted parties of size $|\cC|=\cthr-\rho$. Further, let us assume that $\Tbs$ is shared using the polynomial $f(X)$ as defined below: 
   \[
    f(X):= q(X) \cdot \prod_{k=1}^{\rho} (X-\pos_i)+\sum_{k=1}^{\rho}\tilde{\secret}_i\cdot L_k(X)
   \]
   where $L_k(X):=\prod_{j\in[\rho]\setminus \set{k}} \frac{X-\pos_j}{\pos_k-\pos_j}\cdot \Delta$. and $q(X)$ is a random polynomial of degree $\cthr-\rho-1$. 

   Now observe that the adversary see $|\cC|=\cthr-\rho$ shares corresponding to $f(i_j)$ for $i_j\in\cC$. By Lagrange interpolation, this induces a one-to-one map from possible secrets to corresponding sharing polynomials. Specifically, we can use the vector of sweeping polynomials, as defined in Definition~\ref{def:sp} to explicitly map any secret vector $\bfs^\ast$ to its sharing polynomial defined by $f(X)+\angle{\bfs^\ast-\bfs,\bsp_{\cC}(X)}$

   In other words, the sharing polynomial to share $\bfs^\ast$ is defined by 
   \[
        f^\ast(X)=f(X)+\sum_{k=1}^{\rho} (\secret^\ast_k-\secret_k)\cdot \sp_{k,\cC}(X)
   \]
   One can verify the correctness. For example, to secret share $\secret_1^\ast$, at position $\csize+1$, we get: $$f^\ast(\csize+1)=f(\csize+1)+\sum_{k=1}^{\rho} (\secret_k^\ast-\secret_k)\cdot \sp_{k,\cC}(X)$$
   Now, observe that $f(\csize+1)=\secret_1\cdot (\Delta^2)$. Meanwhile, $\sp_{1,\cC}(\csize+1)=\Delta^2$ while $\sp_{j,\cC}(\csize+1)=0$ for $1<j\leq \rho$. This simplifies to:
   $f^{\ast}(\csize+1)=\secret_1\cdot\offset^2 + (\secret_1^\ast-\secret_1)\cdot \Delta^2=\secret_1^\ast\cdot\offset^2$. 
    However, while we have an efficient mapping, note that $f^\ast(X)$ could have coefficients that are not of the prescribed form, i.e., coefficients do not lie in the range $[0,2^{\rsize+\kappa_s})$. We will call the event $\term{good}$ if the coefficients lie in the range and $\term{bad}$ even if one of the coefficients does not lie in the range.
   
   Let us apply the above mapping to the secret $\bfs'$ and we have the resulting polynomial:
    \[
        f'(X)=f(X)+\sum_{k=1}^{\rho} (\secret_k'-\secret_k)\cdot \sp_{k,\cC}(X)
   \]
   Now, observe that if $f'(X)$ was a good polynomial, then $f'(j)=f(j)$ for every $j\in\cC$. It follows that if $f'$ was $\term{good}$, then an adversary cannot distinguish whether the secret vector was $\bfs$ or $\bfs'$. 
   
   We will now upper bound the probability that $f'$ was $\term{bad}$ in at least one of the coefficients. We know that $|\secret_k'-\secret_k|\in[0,2^{\ssize})$ for $k=1,\ldots,\rho$. Further all coefficients of $\sp_{k,\cC}(X)$ are upper bounded by $\sp_{\max}$. Therefore, to any coefficient of $f(X)$, the maximum perturbation in value is: $2^{\ssize}\cdot\sp_{\max}\cdot\rho$. Therefore, one requires that the original coefficients of $f$ be sampled such that they lie in $[2^{\ssize}\cdot\sp_{\max}\cdot\rho,2^{\rsize+\kappa_s}-2^{\ssize}\cdot\sp_{\max}\cdot\rho]$. In other words, the probability that one coefficient of $f'$ is bad is:
   \[
   \frac{2\cdot 2^{\ssize}\cdot\sp_{\max}\cdot\rho}{2^{\rsize+\kappa_s}}
   \]
   There are $\cthr-1$ such coefficients. This gives us that the probability is $\leq 2^{-\kappa_s}$ assuming that $\rsize\geq \ssize+\ceil{\log_2(\sp_{\max}\cdot (t-1)\cdot \rho)}+1$
\end{proof}

%% file: khprf.tex
\subsection{Pseudorandom Functions}
\label{sub:prf}
\begin{definition}[Pseudorandom Function (PRF)]
    \label{def:prf}
    A pseudorandom function family is defined by a tuple of PPT algorithms $\prf=(\tpgen,\tpeval)$ with the following definitions:
    \begin{itemize}
        \item $\tppp\getsr\tpgen(1^\kappa)$: On input of the security parameter $\kappa$, the generation algorithm outputs the system parameters required to evaluate the function $F:\cK\times\cX\to\cY$ where $\cK$ is the key space, $\cX$ is the input space, and $\cY$ is the output space. 
        \item $y\gets\tpeval(\tpk,x)$: On input of $x\in\cX$ and a randomly chosen key $\tpk\getsr\cK$, the algorithm outputs $y\in \cY$ corresponding to the evaluation of $F(k,x)$.
    \end{itemize}
    We further require the following security property that: for all PPT adversaries $\cA$, there exists a negligible function $\negl$ such that: 
    \begin{gather*}
	\Pr\left[b=b'~~
	\begin{array}{|c}
	 b\getsr\set{0,1},k\getsr\cK\\
	 \cO_0(\cdot):=F(k,\cdot), \cO_1(\cdot):=U(\cY) \\
	 b'\getsr\cA^{\cO_b(\cdot)}
	\end{array}
	\right]\leq \frac{1}{2}+\negl(\kappa)
    \end{gather*}
    where $U(\cY)$ outputs a randomly sampled element from $\cY$.  
\end{definition}
\begin{definition}[($\gamma$)-Key Homomorphic PRF]
    Let $\prf$ be a pseudorandom function that realizes an efficiently computable function $F:\cK\times\cX\to\cY$ such
    that $(\cK,\oplus)$ is a group. Then, we say that it is
    \begin{itemize}
        \item key homomorphic if: $(\cY,\otimes)$ is also a group and for every $\tpk_1,\tpk_2\in\cK$ and every $x\in\cX$ we get:
        $\tpeval(\tpk_1,x)\times\tpeval(\tpk_2,x)=\tpeval(\tpk_1\oplus\tpk_2,x)$. 
        \item $\gamma=1$-almost key homomorphic if: $\cY=\bbZ_\Prime$ if for every $\tpk_1,\tpk_2\in\cK$ and every $x\in\cX$, there exists an error $e\in\{0,1\}$ we get:
        $\tpeval(\tpk_1,x)\times\tpeval(\tpk_2,x)=\tpeval(\tpk_1\oplus\tpk_2,x)+e$. 
    \end{itemize}
\end{definition}

%% file: dkhprf.tex
\begin{definition}[Distributed Key Homomorphic PRF (DPRF)]
\label{def:tprf}
    A $(\cthr,\csize)$-Distributed PRF is a tuple of PPT algorithms $\tprf:=(\tpgen,\tpshare,\allowbreak\tpeval,\tppeval,\tpcombine)$ with the following syntax:
    \begin{itemize}
        \item $\tppp\getsr\tpgen(1^\kappa,1^\cthr,1^\csize)$: On input of the threshold $\cthr$ and number of parties $\csize$, and security parameter $\kappa$, the $\tpgen$ algorithm produces the system parameter $\tppp$ which is impliclty consumed by all the other algorithms.
        \item $\tpk^{(1)},\ldots,\tpk^{(\csize)}\getsr\tpshare(\tpk,\cthr,\crec,\csize)$: On input of the number of parties $\csize$, corruption threshold $\cthr$, reconstruction threshold $\crec$, and a key $\tpk\getsr\cK$, the share algorithm produces the key share for each party. 
        \item $Y\gets\tpeval(\tpk,x)$: On input of the PRF key $\tpk$ and input $x$, the algorithm outputs $y$ corresponding to some pseudorandom function $F:\cK\times\cX\to\cY$.
        \item $y_i\getsr\tppeval(\tpk^{(i)},x)$: On input of the PRF key share $\tpk^{(i)}$, the partial evaluation algorithm outputs a partial evaluation $y_i$ on input $x\in\cX$. 
        \item $Y'\getsr\tpcombine(\set{y_i}_{i\in \cS})$: On input of partial evaluations $y_i$ corresponding to some subset of shares $\cS$ such that $|\cS|\geq \cthr$, the algorithm outputs $Y'$.
    \end{itemize}
    We further require the following properties: 
    \begin{itemize}
        \item Correctness: We require that the following holds for any $\csize,\cthr,\crec,\kappa\in\bbZ$ with $\cthr<\crec\leq\csize$ and 
        any set $\cS\subseteq[m]$ with $|\cS|\geq \crec$, any input $x\in\cX$:
        \begin{gather*}
            \Pr\left[Y=Y'~~
            \begin{array}{|c}
            \tppp\getsr\tpgen(1^\kappa,1^\cthr,1^\crec,1^\csize), \tpk\getsr\cK\\
            \set{\tpk^{(i)}}_{i\in[\csize]}\getsr  \tpshare(\tpk), \set{y_i\gets\tppeval(\tpk^{(i)},x)}_{i\in\cS}\\
            Y'\gets\tpcombine(\set{y_i}_{i\in\cS}), Y=\tpeval(\tpk,x)
            \end{array}
            \right]=1
        \end{gather*}
        \item Pseudorandomness with Static Corruptions: We require that for any integers $\cthr,\crec,\csize$ with $\cthr<\crec\leq \csize$, and for all PPT adversary $\cA$, there exists a negligible function $\negl$ such that: 
        \begin{gather*}
            \Pr\left[~~
            \begin{array}{c|c}
            b=b' &\tppp\getsr\tpgen(1^\kappa,1^\cthr,1^\crec,1^\csize), \tpk\getsr\cK\\
            |\bfK\cup \set{j:(j,x^\ast)\in \bfE}|\leq \cthr&(\state,\bfK)\getsr\cA(\tppp)\\
            & \set{\tpk^{(i)}}_{i\in[\csize]}\getsr  \tpshare(\tpk,\cthr,\crec,\csize,)\\
           & (\state,x^\ast)\getsr\cA^{\oEval}(\set{\tpk^{(i)}}_{i\in\bfK})\\
           & b\getsr\{0,1\}, Y_0\getsr\tpeval(\tpk,x^\ast)\\
           & Y_1\getsr\cU(\cY), b'\getsr\cA^{\oEval}(\state,Y_b)
            \end{array}
            \right]=1
        \end{gather*}
        where:
        \begin{pchstack}[center,space=1cm]
        \procedure{$\oEval(i,x)$}{
             \bfE:=\bfE \cup \set{(i,x)}\\
             \pcreturn \tppeval(\tpk^{(i)},x)\\
            }
        \end{pchstack}
        \item Key Homomorphic: We require that if $(\cK,\oplus),(\cY,\otimes)$ are groups such that:
    \begin{itemize}
        \item $\forall~x\in\cX,\forall~\tpk_1,\tpk_2\in\cK, \tpeval(\tpk_1,x)\otimes \tpeval(\tpk_2,x)=\tpeval(\tpk_1\oplus \tpk_2,x)$, and
        \item $\forall~x\in\cX~$, $\forall~\tpk_1,\tpk_2\in\cK$, $\set{\tpk_b^{(j)}\getsr\tpshare(\tpk_b)}_{j\in[m],b\in\{1,2\}}$, 
        \item[] $\forall~j\in[m], y_{1,2}^{(j)}:=\left(\tppeval(\tpk_1^{(j)},x)\otimes \allowbreak \tppeval(\tpk_2^{(j)},x)\right)$, and $\forall\cS\subseteq [m]$ with $|\cS|\geq \crec$, $\tpcombine\allowbreak\left(\set{y_{1,2}^{(j)}}_{j\in\cS}\right)=\tpeval(\tpk_1\oplus \tpk_2,x)$
    \end{itemize}
    \end{itemize}
\end{definition}

%% file: class.tex
\subsection{Class Groups Framework}
\label{sub:class}
Class group-based cryptography is a cryptographic technique that originated in the late 1980s, with the idea that the class group of ideals of maximal orders of imaginary quadratic fields may be more secure than the multiplicative group of finite fields \cite{JC:BucWil88,Curley89}. The CL framework was first introduced by the work of \CL~ \cite{RSA:CasLag15}. This framework operates on a group where there exists a subgroup with support for efficient discrete logarithm construction. Subsequent works \cite{AC:CasLagTuc18,C:CCLST19,PKC:CCLST20,Tucker20,C:BraDamOrl23} have refined the original framework. The framework has been used in various applications over the years \cite{AC:CasLagTuc18,C:CCLST19,PKC:CCLST20,PKC:YueCuiXie21,AC:DMZWSX21,EPRINT:GMMMTT22,CCS:TCLM21,C:BraDamOrl23,EPRINT:KMMST23}. Meanwhile, class group cryptography itself has been employed in numerous applications \cite{EC:Wesolowski19,JC:Wesolowski20,ACNS:Lipmaa12,C:BonBunFis19,C:CasImbLag17,EC:ChaCou18,C:LaiMal19,EC:BunFisSze20,PKC:AGLMS23,C:ADOS22,EC:CKLR21,CCS:CGKR22,EPRINT:ACCDE22}.


Broadly, the framework is defined by two functions - $\cl\Gen,\cl\Solve$ with the former outputting a tuple of public parameters. The elements of this framework are the following:
\begin{itemize}
    \item \textit{Input Parameters}: $\kappa_c$ is the computational security parameter, $\kappa_s$ is a statistical security parameter, a prime $\Prime$ such that $\Prime>2^{\kappa_c}$, and uniform randomness $\rho$ that is used by the $\cl\Gen$ algorithm and is made public.
    \item \textit{Groups}: $\Ghat$ is a finite multiplicative abelian group, $\bbG$ is a cyclic subgroup of $\Ghat$, $\bbF$ is a subgroup of $\bbG$, $\Gp=\set{x^\Prime,x\in \bbG}$
    \item \textit{Orders}: $\bbF$ has order $\Prime$, $\Ghat$ has order $\Prime\cdot \shat$, $\bbG$ has order $\Prime\cdot s$ such that $s$ divides $\shat$ and $gcd(p,\shat)=1, gcd(p,s)=1$, $\Gp$ has order $s$ and therefore $\bbG=\bbF\times \Gp$.
    \item \textit{Generators}: $f$ is the generator of $\bbF$, $g$ is the generator of $\bbG$, and $\gp$ is the generator of $\Gp$ with the property that $g=f\cdot \gp$
    \item \textit{Upper Bound}: Only an upper bound $\sbar$ of $\shat$ (and $s$) is provided. 
    \item \textit{Additional Properties}: Only encodings of $\Ghat$ can be recognized as valid encodings and $s,\shat$ are unknown. 
    \item \textit{Distributions}:  $\cD$ (resp. $\cD_p$) be a distribution over the set of integers such that the distribution $\{g^x:x\getsr\cD\}$ (resp. $\{g_p^x:x\getsr\cD_p\}$) is at most distance $2^{-\kappa_s}$ from the uniform distribution over $\bbG$ (resp. $\Gp$).
    \item \textit{Additional Group and its properties}: $\Ghp=\set{x^\Prime,x\in\Ghat}$, $\Ghat$ factors as $\Ghp\times\bbF$.\footnote{Recall that $\Prime$ and $\shat$ are co-prime.} 
    Let $\baro$ be the group exponent of $\Ghp$. Then, the order of any $x\in\Ghp$ divides $\baro$.\footnote{This follows
        from the property that the exponent of a finite Abelian group is the least common multiple of its elements.}
\end{itemize}

\begin{remark}
    The motivations behind these additional distributions are as follows. One can efficiently recognize valid encodings of elements in $\Ghat$ but not $\bbG$. Therefore, a malicious adversary $\cA$ can run our constructions by inputting elements belonging to $\Ghp$ (rather than in $\Gp$). Unfortunately, this malicious behavior cannot be detected which allows $\cA$ to obtain information on the sampled exponents modulo $\baro$ (the group exponent of $\Ghp$). By requiring the statistical closeness of the induced distribution to uniform in the aforementioned groups allows flexibility in proofs. Note that the assumptions do not depend on the choice of these two distributions. Further, the order $s$ of $\Gp$ and group exponent $\baro$ of $\Ghp$ are unknown and the upper bound $\sbar$ is used to instantiate the aforementioned distribution. 
    Specifically, looking ahead we will set $\Dp$ to be the uniform distribution over the set of integers $[B]$ where $B=2^{\kappa_{s}}\cdot \sbar$. Using Lemma~\ref{lem:sample}, we get that the distribution is less than $2^{-\kappa_s}$ away from uniform distribution in $\Gp$. In our constructions we will set $\kappa_s=40$. 
    We will make this sampling more efficient for our later constructions. We refer the readers to Tucker~\cite[\S 3.1.3, \S 3.7]{Tucker20} for more discussions about this instantiation. 
    Finally, as stated we will also set $\Dhat=\cD$ and $\Dhp=\Dp$. 
\end{remark}
We also have the following lemma from Castagnos, Imbert, and Laguillaumie~\cite{C:CasImbLag17} which defines how to sample from a discrete Gaussian distribution. 
\begin{lemma}
\label{lem:sample}
Let $\bbG$ be a cyclic group of order $n$, generated by $g$. Consider the random variable $X$ sampled uniformly from $\bbG$; as such it satisfies $\Pr[X=h]=\frac{1}{n}$ for all $h\in\bbG$. Now consider the random variable $Y$ with values in $\bbG$ as follows: draw $y$ from the discrete Gaussian distribution $\cD_{\bbZ,\sigma}$ with $\sigma\geq n\sqrt{\frac{\ln(2(1+1/\epsilon))}{\pi}}$ and set $Y:=g^y$. Then, it holds that: 
\[
\Delta(X,Y)\leq 2\epsilon 
\]
\end{lemma}
\begin{remark}
\label{rem:1}
By definition, the distribution $\{g^x:x\getsr\cD\}$ is statistically indistinguishable from $\{g^y:y\getsr\{0,\ldots,\Prime\cdot s-1\}\}$. Therefore, it follows that $\{x\mod \Prime\cdot s:x\getsr\cD\}$ is statistically indistinguishable from $\{x:x\getsr\set{0,\ldots,\Prime\cdot s-1}\}$. Similarly, $\set{x\mod s:x\getsr\allowbreak \cD_\Prime}$ is statistically indistinguishable from $\{x:x\getsr\allowbreak\{\allowbreak 0\allowbreak,\allowbreak\ldots\allowbreak,\allowbreak s\allowbreak-\allowbreak1\}\}$. Furthermore, sampling a value $x$ corresponding to $\cD$ is statistically indistinguishable from the uniform distribution in $\set{0,\ldots,s-1}$ because $s$ divides
$\Prime\cdot s$. 
\end{remark}
\begin{definition}[Class Group Framework]
\label{def:class}
The framework is defined by two algorithms $(\cl\Gen, \cl\Solve)$ such that:
\begin{itemize}
    \item $\pp_\cl=(\Prime,\kappa_c,\kappa_s,\sbar,f,\gp,\Ghat,\bbF.\cD,\Dp,\Dhat,\Dhp,\rho)\getsr\allowbreak\cl\Gen(1^{\kappa_c},\allowbreak 1^{\kappa_s},\allowbreak\Prime;\allowbreak\rho)$
    \item The $\DL$ problem is easy in $\bbF$, i.e., there exists a deterministic polynomial algorithm $\cl\Solve$ that solves the discrete logarithm problem in $\bbF$: 
    \begin{gather*}
	\Pr\left[x=x'~~
	\begin{array}{|l}
	\pp_\cl=\getsr\cl\Gen(1^{\kappa_c},1^{\kappa_s},\Prime;\rho)\\ x\getsr \bbZ/p\bbZ, X=f^x; \\
    x'\gets\cl\Solve(\pp_\cl,X)\\
	\end{array}
	\right] =1 
    \end{gather*}
\end{itemize} 
\end{definition}

\begin{definition}[Hard Subgroup Membership Assumption (\HSM) Assumption~\cite{AC:CasLagTuc18}]
\label{def:hsm}
Let $\kappa$ be a the security parameter with prime $\Prime$ such that $|\Prime|\geq \kappa$. Let $(\cl\Gen,\cl\Solve)$ be the group generator algorithms as defined in Definition~\ref{def:class}, then the \HSM\ assumption requires that that \HSM\ problem be hard in $\bbG$ even with access to the $\Solve$ algorithm. More formally, let $\cD$ (resp. $\cD_p$) be a distribution over the set of integers such that the distribution $\{g^x:x\getsr\cD\}$ (resp. $\{g_p^x:x\getsr\cD_p\}$) is at most distance $2^{-\kappa}$ from the uniform distribution over $\bbG$ (resp. $\Gp$). Then, we say that the \HSM\ problem is hard if for all PPT adversaries $\cA$, there exists a negligible function $\negl(\kappa)$ such that:
\begin{gather*}
	\Pr\left[b=b'
	\begin{array}{c|c}
	& \pp_\cl\getsr\cl\Gen(1^{\kappa_c},1^{\kappa_s},\Prime;\rho)\\& x\getsr\cD,x'\getsr\cD_p\\
	& b\getsr \{0,1\}; Z_0=g^x; Z_1=g_p^{x'} \\
	& b'\getsr\cA^{\Solve(\pp_\cl,\cdot)}(\pp_\cl,Z_b)
	\end{array}
	\right]\leq \frac{1}{2}+\negl(\kappa)
    \end{gather*}
\end{definition}
When dealing with groups of known order, one can sample elements in a group $\bbG$ easily by merely sampling exponents modulo the group order and then raising the generator of the group to that exponent. Unfortunately, note that here neither the order of $\bbG$ (i.e., $ps$) nor that of $\Gp$ (i.e. $s$) is known. Therefore, we instead use the knowledge of the upper-bound $\sbar$ of $s $ to instantiate the distributions $\cD$ and $\cD_p$ respectively. This choice of choosing from the distributions $\cD$ and $\cD_p$ respectively allows for flexibility of various proofs.

%% file: cl-prf.tex
\section{Constructions in $\cl$ Framework}
\label{sub:khprf-cons}
\begin{construction}[PRF in $\cl$ Framework]
\label{cons:khprf}
    Let $(\cl\Gen,\cl\Solve)$ be the class group framework as defined in Definition~\ref{def:class}. Then, let 
    $\pp_\cl\getsr\cl\Gen(1^{\kappa_c},\allowbreak1^{\kappa_s})$. Further, let $\hash:\cX\to\bbH$ be a hash function. Then, consider the following definition of $\cK=\Dp,\cX=\{0,1\}^\ast,\cY=\bbH$, $F_\cl(k,x)=\hash(x)^{k}$. 
\end{construction}
\begin{restatable}{theorem}{khprf}
\label{thm:khprf}
  Construction~\ref{cons:khprf} is a secure PRF where $\hash$ is modeled as a random oracle under the \HSM\ assumption.
\end{restatable}
\ifdefined\IsSub{}
\begin{proof}
We denote the challenger by $\cB$. Let $S_j$ be the event that the adversary wins in $\Hybrid_j$ for each $j\in\{0,\ldots,2\}$. Let $q_e$ (resp. $q_h$) denote the number of evaluation queries (resp. hash oracle queries) that the adversary makes. We use an analysis similar to the technique by Coron~\cite{C:Coron00}.
\begin{gamedescription}[name=Hybrid, nr=-1]
\describegame Corresponds to the security game as defined for the security of PRF. It follows that the advantage of the adversary is $$\Adv_0=2\cdot |\Pr[S_0]-1/2]=\term{Adv}_\cA^{\term{PRF}}$$

\describegame This game is identical to $\Hybrid_1$ with the following difference. The challenger tosses biased coin $\delta_t$  for each random oracle query $H(t)$. The biasing of the coin is as follows: takes a value 1 with probability $\frac{1}{q_{e}+1}$ and 0 with probability $\frac{q_{e}}{q_{e}+1}$. Then, one can consider the following event $E$: that the adversary makes a query to the random oracle with $x_i$ as an input where $x_i$ was one of the evaluation inputs and for this choice we have that $\delta_t$ was flipped to 0. 

If $E$ happens, the challenger halts and declares failure. Then, we have that:
\[
\Pr[\neg E]=\left(\frac{q_{e}}{q_{e}+1}\right)^{q_{e}}\geq \frac{1}{e(q_{e}+1)}
\]
where $e$ is the Napier's constant. 
Finally, we get that:
\[
\Pr[S_1]=\Pr[S_0]\cdot \Pr[\neg E] \geq \frac{\Pr[S_0]}{e(q_{e}+1)}
\]
\describegame This game is similar to $\Hybrid_1$ with the following difference: we modify the random oracle outputs. 
\begin{itemize}
    \item If $\delta_t=0$, the challenger samples $w_t\getsr\cD_H$ and sets $H(t)=h^{w_t}$
    \item If $\delta_t=1$, the challenger samples $w_t\getsr\cD_H,u_t\getsr\allowbreak\bbZ/\modulus\bbZ$ and sets $H(t)=h^{w_t}\cdot f^{u_t}$
\end{itemize}
Note that, under the HSM assumption, an adversary cannot distinguish between the two hybrids. Therefore, we get:
\[
|\Pr[S_2]-\Pr[S_1]|\leq \epsilon_{\HSM}
\]
where $\epsilon_{\HSM}$ is the advantage that an adversary has in the $\HSM$ game. 
Note that $\Hybrid_2$ corresponds to the case where the outputs are all random elements in $\bbG$. Therefore, the inputs are sufficiently masked and leak no information about the key. Therefore, $\Pr[S_2]=0$
Then,
\[
\Adv_\cA^{\term{PRF}}\leq (e\cdot (q_e+1)\cdot \epsilon_{\HSM}
\]
\end{gamedescription}
\end{proof}
\begin{remark}
    Note that the above scheme is simply an adaptation of the famous DDH-based construction of a key-homomorphic PRF that was shown to be secure by Naor~\etal\cite{EC:NaoPinRei99}. It is easy to verify that our construction is also key homomorphic as $H(x)^{(k_1+k_2)}=H(x)^{k_1}\cdot H(x)^{k_2}$. 
\end{remark}
\else The proof is deferred to Section~\ref{sub:thm-1}
\fi

%% file: tprf-cons.tex
\subsection{Distributed PRF in $\cl$ Framework}
\label{sub:tprf-cons}
We build our construction of Distributed PRF from the Linear Secret Sharing Scheme $\liss:=(\lshare,\allowbreak\lcoeff,\allowbreak\lcomb)$ with $\pp_\sss$ denoting the public parameters of the $\liss$ scheme. We specifically employ the Shamir Secret Sharing scheme over the Integers, as defined in \cite{C:BraDamOrl23}. %
\begin{restatable}[Distributed PRF in $\cl$ Framework]{construction}{dprf}
\label{cons:tprf}
A $(\crec,\csize)$-Distributed PRF is a tuple of PPT algorithms $\tprf:=(\tpgen,\tpshare,\allowbreak\tpeval,\tppeval,\tpcombine)$ with the algorithms as defined in Figure~\ref{fig:tprf}. For simplicity, in the construction below we will set the corruption threshold $\cthr=\crec-1$. Though, the construction also holds for a lower $\cthr$.
\end{restatable}

\begin{figure}[!tb]
    \renewcommand{\algorithmiccomment}[1]{{\vspace*{-0.3em}\color{gray}// #1}}

	\centering
	\resizebox{\columnwidth}{!}{\begin{protocolbox}{Distributed PRF}
		\scalebox{0.8}{
			\begin{minipage}[t]{0.55\linewidth}
			
			    \algoHead{$\tpgen(1^\kappa,1^\crec,1^\csize)$}
			    \begin{algorithmic}
                \State Parse $\kappa=(\kappa_s,\kappa_c)$
                \State Run $\pp_\cl\getsr\cl\Gen(1^{\kappa_c},1^{\kappa_s})$
                \State Set $\ssize:=|\modulus|$
                \State Sample $\hash:\cX\to\bbH$
                \State \Return $\pp_{\prf}:=(\pp_\cl,\liss.\pp_\sss,\hash)$
			    \end{algorithmic}
       
			    \algoHead{$\tpshare(\tpk,\crec,\csize)$}
			    \begin{algorithmic}
                \State Run $\tpk^{(1)},\ldots,\tpk^{(\csize)}\getsr\liss.\lshare(\tpk,\crec,\csize)$
                \State \Return $\tpk^{(1)},\ldots,\tpk^{(\csize)}$
			    \end{algorithmic}

		\end{minipage}
		}
		\hfill\scalebox{.8}{
			\begin{minipage}[t]{0.55\linewidth}	
   
       			\algoHead{$\tpeval(\tpk,x)$}
			    \begin{algorithmic}
                \State Compute $Y=\hash(x)^{\Delta^3\cdot \tpk}$
                \State \Return $Y$
			    \end{algorithmic}
			 \algoHead{$\tppeval(\tpk^{(i)},x)$}
			    \begin{algorithmic}
                \State Compute $y_i=\hash(x)^{\Delta\cdot \tpk^{(i)}}$
                \State \Return $y_i$
			    \end{algorithmic}
			    \algoHead{$\tpcombine(\set{y_i}_{i\in \cS})$}
			    \begin{algorithmic}
                    \State Run $\Lambda_{i\in\cS}=\liss.\lcoeff(\cS)$
			        \State  Compute $Y'=\prod_{i\in\cS}y_i^{\Lambda_i}$
                    \State \Return $Y'$
			    \end{algorithmic}
			\end{minipage}
		}
	\end{protocolbox}}
	\vspace*{-0.7em}
	\caption{Construction of Distributed PRF based on the $\liss:=(\lshare,\lcoeff,\lcomb)$ scheme of \cite{C:BraDamOrl23}, with $\pp_\sss$ denoting the public parameters of the $\liss$ scheme. Recall that the offset $\offset:=\csize!$ where $\csize$ is the number of shares generated.}
	\label{fig:tprf}

\end{figure}
\begin{restatable}{theorem}{thmdprf}
    \label{thm:tprf}
    In the Random Oracle Model, if Construction~\ref{cons:khprf} is a secure pseudorandom function if Integer Secret Sharing is statistically private, then Construction~\ref{cons:tprf} is pseudorandom in the static corruptions setting. 
\end{restatable}
\ifdefined\IsSub{}
\paragraph{Correctness.} For a polynomial $f\in\bbZ[X]$, every $f(i)$ leaks information about the secret $\secret\bmod i$ leading to a choice of polynomial $f$ such that $f(0)=\Delta\cdot \secret$. For our use case, the secret is the PRF key $\tpk$. Let us consider a set $\cS=\set{i_1,\ldots,i_\crec}$ of indices and corresponding evaluations of the polynomial $f$ at $i_1,\ldots,i_\crec$ giving us key shares: $\tpk^{(i_1)},\ldots,\tpk^{(i_\crec)}$. To begin with, one can compute the Lagrange coefficients corresponding to the set $\cS$ as: $\forall~i\in\cS,\lambda_i(X):=\prod_{j\in\cS\setminus\set{i}} \frac{x_j-X}{x_j-x_i}$. This implies that the resulting polynomial is $f(X):=\sum_{j=1}^{\crec} \lambda_{i_j}(X)\cdot \tpk^{(i_j)}$. 

However, $\lambda_i(X)$ requires one to perform a division $x_j-x_i$ which is undefined as $\hash$ hashes to $\bbG$ whose order is unknown. To avoid this issue, a standard technique is to instead compute coefficient $\Lambda_i(X):=\Delta\cdot \lambda_i(x)$. Thereby, the resulting polynomial that is reconstructed if $f'(X)=\Delta\cdot f(X)=\sum_{j=1}^{\crec} \Lambda_{i_j}(X)\cdot \tpk^{(i_j)}$. Consequently, 

\begin{align*}
    \hash(x)^{\Delta^3\cdot \tpk}=\hash(x)^{\Delta\cdot f'(0)}&=\hash(x)^{\Delta\cdot \sum_{j=1}^{\crec} \Lambda_{i_j}(0)\cdot \tpk^{(i_j)}}\\&=\prod_{j=1}^\crec \left(\hash(x)^{\Delta\cdot \tpk^{{(i_j)}}}\right)^{\Lambda_{i_j}(0)}\\
     &=\prod_{j=1}^\crec \left(\tppeval(\tpk^{(i_j)},x)\right)^{\Lambda_{i_j}(0)}
\end{align*}
Thus, our protocol is correct. 

\paragraph{Pseudrandomness.} Next, we consider the pseudorandomness property of our construction. 
 \thmdprf*
Boneh~\etal\cite{C:BLMR13} showed that from any Key Homomorphic PRF (which Construction~\ref{cons:khprf}), one can build a Distributed PRF. The proof of the following theorem follows the template of this scheme with certain important adaptations as our secret sharing scheme is over integers. The proof technique is to show that if there exists an adversary $\cA$ that can break the $\tprf$ security, one can then use it to build an adversary $\cB$ to break the pseudorandomness of our original PRF, as defined in Construction~\ref{cons:khprf}. The idea behind the proof is for $\cB$, upon receiving choice of $\crec-1$ corruptions as indices $i_1,\ldots,i_{\crec-1}$, to then choose a random index $i_\crec$ and implicitly set $\tpk^{(i_\crec)}$ to be the PRF key chosen by its challenger. Therefore, the $\cB$ now has knowledge of $\crec$ indices, with which it can sample the Lagrange coefficients as before: 
\begin{code}
    \cFor $j=1,\ldots,\crec$ \cDo\\
        \> $\Lambda_{i_j}(X):=\prod_{\zeta\in\set{1,\ldots,\crec}\setminus{j}}^{\crec} \frac{i_\zeta-X}{i_\zeta-i_j}\cdot(\Delta)$
    \end{code}
Now, $\cB$ with knowledge of the keys for indices $i_1,\ldots,i_{\crec-1}$ along with access to Oracle needs to simulate valid responses to $\tppeval$ queries for an unknown index. Call this index $i^\ast$. Then, we have:
\begin{align*}
    \tppeval(\tpk^{(i^\ast)},x):
    &=\hash(x)^{\Delta\cdot \tpk^{(i^\ast)}}=\hash(x)^{\sum_{j=1}^{\crec} \Lambda_{i_j}(i^\ast)\cdot \tpk^{(i_j)}}\\
    &=\hash(x)^{\sum_{i=1}^{\crec-1}\Lambda_{i_j}(i^\ast)\cdot \tpk^{(i_j)}}\cdot \left(\hash(x)^{\tpk^{(i_\crec)}}\right)^{\Lambda_{i_\crec}(i^\ast)}\\
\end{align*}
The last term is simulated using $\cB$'s own oracle access. 
\newcommand{\A}{\cA}
\begin{proof}
Let $\cA$ be a PPT attacker against the pseudorandomness property of $\tprf$, having advantage $\epsilon$.

$\A$ first chooses $\crec-1$ indices $\bfK=\{i_1,\ldots,i_{\crec-1}\}$ where each index is a subset of $\{1,\ldots,\csize\}$. $\cA$ receives the shares of the keys $\tpk^{(i_1)}=f(i_1),\ldots,\tpk^{(i_{\crec-1}) }=f(i_{\crec-1})$ (for unknown polynomial $f$ of degree $\crec$ such that $f(0)=\tpk\cdot \Delta$ with $\Delta:=\Delta$. Further, $\A$ has access to $\oEval(i,x)$ receiving $\tppeval(\tpk_i,x)$ ins response. Additionally, $\A$ expects to have oracle access to the random oracle $\hash$.

Using this attacker $\A$, we now define a PPT attacker $\cB$ which will break the pseudorandomness property of Construction~\ref{cons:khprf}. Note that $\cB$ is given access to the oracle that either outputs the real evaluation of the PRF on key $\tpk^\ast$ or a random value. Additionally, $\cB$ expects to have oracle access to the random oracle $\hash$.
 
\begin{itemize}
\item \textbf{Setup:} $\cB$ does the following during Setup. 
\begin{itemize}
    \item Receive set $S=\{i_1,\ldots,i_{\crec-1}\}$ from $\A$. 
    \item Next $\cB$ generates the key shares and public key as follows: 
    \begin{itemize}
        \item Sample $\tpk^{(i_1)},\ldots,\tpk^{(i_{\crec-1})}\in \bbZ$.
        \item $\cB$ picks an index $i_\crec$ at random and implicitly sets the PRF key chosen by its challenger as $\tpk^{(i_t)}$.
        \item Immediately, given the $\crec$ indices, one can construct the secret sharing polynomial $f\in\bbZ[X]$ as described earlier, but instead recreating the polynomial $f'(X)$ using the coefficients $\Lambda_{i_j}(X)$ for $j=1,\ldots,\crec$ with $\tpk^{(i_t)}$ being unknown to $\cB$ and using its challenger to simulate a response.
        \item $\cB$ gives $\tpk^{(i_1)},\ldots,\tpk^{(i_{\crec-1})}$ to $\cA$.
    \end{itemize}
    \end{itemize}
\item \textbf{Queries to H:} 
$\cB$ merely responds to
all queries from $\cA$ to $\hash$ by using its oracle access to $\hash$.
\item \textbf{Queries to Partial Evaluation:} $\cB$ receives as query input, some choice of key index specified by $i^\ast$ and input $x_j$ for $i=1,\ldots,Q$.  In response $\cB$ does the following:
\begin{itemize}
    \item Forward $x_j$ to its challenger. In response it implicitly receives $\tppeval(\tpk^{(i_\crec)},x_j)$, but off by a factor of $\Delta$ in the exponent. Call this $h_{j,\crec}$.
    \item Compute: $h_{j,i^\ast}=\hash(x_j)^{\sum_{i=1}^{\crec-1}\Lambda_{i_j}(i^\ast)\cdot \tpk^{(i_j)}}\cdot (h_{j,\crec})^{\Lambda_{i_\crec}(i^\ast)}$ where $\cB$ uses its own access to hash oracle to get $\hash(x_j)$.
    \item It returns $h_{j,i^\ast}$ to $\cA$. 
\end{itemize}
\item \textbf{Challenge Query:} On receiving the challenge input $x^\ast$, $\cB$ does the
following:
\begin{itemize}
    \item Ensure that it is a valid input, i.e., there is no partial evaluation queries on $x^\ast$ at any unknown index point. 
    \item If not, $\cB$ forwards to its challenger $x^\ast$. In response it implicitly receives $\tppeval(\tpk^{(i_\crec)},x^\ast)$, but off by a factor of $\Delta$ in the exponent. Call this $h^\ast$. 
    \item It also uses its oracle access to $\hash$ to receive
    $h=\hash(x^\ast)$.
    \item It finally computes $y=\hash(x_j)^{\Delta^2\cdot \sum_{i=1}^{\crec-1}\Lambda_{i_j}(0)\cdot \tpk^{(i_j)}}\cdot (h^\ast)^{\Delta^2\cdot \Lambda_{i_\crec}(0)}$ and outputs
    $y$ to $\cA$
\end{itemize}
\item \textbf{Finish:} It forwards $\cA$'s guess as its own guess. 
\end{itemize}
\paragraph{Analysis of the Reduction.} Note that for the case when $b=0$, $\cA$ expects to receive $\hash(x^\ast)^{\Delta^3\cdot \tpk}$ where $\tpk$ is defined at the point 0. So, we get:
\begin{align*}
    \hash(x^\ast)^{\Delta^3\cdot \tpk^{(0)}}&=\hash(x^\ast)^{\Delta^2\cdot f'(0)}=\hash(x)^{\Delta^2\sum_{j=1}^{\crec} \Lambda_{i_j}(0)\cdot \tpk^{(i_j)}}\\
    &=\hash(x)^{\Delta^2\sum_{i=1}^{\crec-1}\Lambda_{i_j}(0)\cdot \tpk^{(i_j)}}\cdot \left(\hash(x)^{\tpk^{(i_\crec)}}\right)^{\Delta^2\Lambda_{i_\crec}(0)}\\
\end{align*}
This shows that the returned value $y$ is consistent when $b=0$. Meanwhile, when $b=1$, $h^\ast$ is a random element in the group and then $y$ is a truly random value which means that $\cB$ has produced a valid random output for $\cA$. Similarly, when $b=0$, every response to partial evaluation is also done consistently by correctness of the underlying secret sharing scheme. Meanwhile, when $b=1$, we can rely on the statistical privacy preserving guarantee of the underlying secret sharing scheme to argue that the difference that the adversary can notice is statistically negligible. This concludes the proof where $\cB$ can only succeed with advantage $\epsilon$. 
\end{proof}
\else 
The proof of security and correctness can be found in Section~\ref{sub:tprf}.
\fi

%% file: fl.tex
\subsection{Construction of \CAPS\ without Leakage Simulation in $\cl$ Framework}
\label{sub:async}

\begin{construction}
    \label{cons:caps-cl}
   We present the construction of \CAPS\ without Leakage Simulation, in the $\cl$ Framework in Figure~\ref{fig:async}.
\end{construction}
\renewcommand{\fbox}{\fcolorbox{red}{white}}
\newcommand{\fbbox}{\fcolorbox{blue}{white}}
\ifdefined\IsSP
\begin{figure}[!tb]
\else
\begin{figure}[!tb]
\fi
\centering
	\resizebox{\columnwidth}{!}{\begin{protocolbox}{$\caps'$}
		\scalebox{0.9}{
			\begin{minipage}[t]{\linewidth}
			    \algoHead{System Parameters Generation}
			    \begin{algorithmic}
                       \State Run $\pp\getsr\cl.\Gen(1^\kappa,1^\cthr,1^\crec,1^\csize)$
			    \end{algorithmic}
               
        \algoHead{Server Initiating Iteration $\lab$}
			    \begin{algorithmic}
                       \State Select $n$ clients $\cC_\lab$. 
                       \State Select $\csize=\floor{\log_2 n}$-sized committee
                       \State Broadcast to $n$ clients the committee.
			    \end{algorithmic}
			     \algoHead{Data Encryption Phase by Client $i$ in iteration $\lab$}
			    \begin{algorithmic}
                    \State \textbf{Input:} $\mathbf{\psain_{i,\lab}}=(\psain_{i,\lab}^{(1)},\ldots,\psain_{i,\lab}^{(L)})$
                     \State {Sample $\tpk_{i,1},\ldots,\tpk_{i,L}\getsr\tprf.\cK$}
                    \For{ ~$in=1,\ldots,L$}
                    \State $h_{i,\lab}^{(in)}=\tprf.\tpeval(\tpk_{i,in},\lab)$
                    \State Compute $\psact_{i,\lab}^{(in)}=f^{\psain_{i,\lab}^{(in)}}\cdot h_{i,\lab}^{(in)}$
                    \State Compute $\set{\tpk_{i,in}^{(j)}}_{j\in[\csize]}\getsr\tprf.\tpshare(\tpk_{i,in},\cthr,\crec,\csize)$
                    \State For $j\in[\csize]$, set $\set{\psaaux_{i,\lab}^{(j,in)}}=\tprf.\tppeval(\tpk_{i,in}^{(j)},\lab$
                    \EndFor
                    \State {Send $\psact_{i,\lab}^{(1)},\ldots,\psact_{i,\lab}^{(L)}$ to the Server}
                    
                    \State {Send $\psaaux_{i,\lab}^{(j,1)},\ldots,\psaaux_{i,\lab}^{(j,L)}$ to committee member $j~\forall j\in[\csize]$, via Server appropriately encrypted}
			    \end{algorithmic}
        \algoHead{Server Forwards in iteration $\lab$}
			    \begin{algorithmic}
                    \State Let $\cC$ be the clients who sent the prescribed messages. \State Send encrypted $\psaaux_{i,\lab}^{(j)}$ to committee member $j~\forall j\in[\csize],i\in\cC$
			    \end{algorithmic}
        \algoHead{Data Combination Phase by Committee Member $j$ in iteration $\lab$}
			    \begin{algorithmic}
                    \State {\textbf{Input:} $\set{\psaaux_{i,\lab}^{(j,1)},\ldots,\psaaux_{i,\lab}^{(j,L)}}_{i\in\cC}$}
                    \For{$in=1,\ldots,L$}
                    \State {$(\AUX^{(j,in)})\gets\otimes_{i\in\cC}({\psaaux_{i,\lab}^{(j,in)}})$}
                    \EndFor
                    \State {Send $\AUX^{(j,1)},\ldots,\AUX^{(j,L)}$ to Server}
			    \end{algorithmic}
                     \algoHead{Data Aggregation Phase by Server in iteration $\lab$}
                     
			    \begin{algorithmic}
                     \State {\textbf{Input:}  $\set{\set{\AUX_\lab^{(j,in)}}_{j\in\cS},\set{\psact_{i,\lab}^{(in)}}_{i\in\cC}}_{in\in[L]}$, $|\cS|\geq\cthr$.}
                    \For{$in=1,\ldots,L$}
                    \State {Run $\AUX_\lab^{(in)}\gets \tprf.\tpcombine(\set{\AUX_\lab^{(j,in)}}_{j\in\cS})$}
                     \State {Run $X_\lab^{(in)}\gets\cl\Solve(\pp,(\AUX_\lab^{(in)})^{-1}\cdot \prod_{i\in\cC}{\psact_{i,\lab}^{(in)}})$}
                    \EndFor
                    \State  \textbf{return} $\set{X_\lab^{(in)}}_{in\in[L]}$
			    \end{algorithmic}
		\end{minipage}
		}
  \end{protocolbox}}
	\caption{In this figure, we present the modified steps for $\caps'$. For simplicity, we present only the semi-honest construction. For malicious security, we augment similarly with a second mask. }
	\label{fig:async}
 \label{fig:three}
\ifdefined\IsSP
\end{figure}
\else
\end{figure}
\fi

%% file: lat.tex
\section{Constructions of Leakage Resilient Key-Homomorphic Pseudorandom Functions}
\label{sec:khprf-cons}
\label{sec:prf-cons}
\subsection{Construction from \LWR\ Assumption}
\label{sub:cons-prf-lwr}
We also have the construction from Boneh~\etal\cite{C:BLMR13} of an \emph{almost}
Key Homomorphic PRF from LWR in the Random Oracle model which was later formally proved
secure by Ernst and Koch~\cite{PETS:ErnKoc21} with $\gamma=1$. 

\begin{construction}[Key Homomorphic PRF from \LWR]
    \label{cons:khprf-lwr}
    Let $\hash:\cX\to\bbZ_q^\lambda$. Then, define the efficiently computable function $F_\LWR:\bbZ_q^\lambda\times\cX\to\bbZ_p$ as $\floor{\angle{\hash(x),\bfk}}_\Prime$. $F_\LWR$ is an almost key homomorphic PRF with $\gamma=1$. 
\end{construction}
\begin{construction}[Length Extended Key-Homomorphic from \LWR]
\label{cons:prf-lwr}
    Let $F_\lwr$ be the function as defined in Construction~\ref{cons:khprf-lwr}. Then, consider $F_\LWR^L:=(F_\LWR(\tpk,(x,1)),\ldots F_\LWR(\tpk,(x,L)))$. 
\end{construction}
It is easy to see that Construction~\ref{cons:prf-lwr} is a secure pseudorandom function. An adversary that can break the security of Construction~\ref{cons:prf-lwr} can be used to break the security of Construction~\ref{cons:khprf-lwr}. 

\begin{theorem}[Leakage Resilience of Construction~\ref{cons:prf-lwr}]
\label{thm:lr-prf-lwr}
Let $\prf$ be the PRF defined in Construction~\ref{cons:prf-lwr}. Recall that $\prf.\cK=\bbZ_q^\lambda$. Then, it is leakage resilient in the following sense:
\begin{gather*}
  \set{\prf(\tpk,x),(\tpk+r)\bmod q:\tpk,r\getsr\cK} \approx_c \set{Y,(\tpk+r)\bmod q:Y\getsr\cY,\tpk,r\getsr\cK}  
\end{gather*}
\end{theorem}
\begin{proof}
The proof proceeds through a sequence of hybrids. 
\begin{gamedescription}[name=Hybrid,nr=-1]
    \describegame The left distribution is provided to the adversary. In other words, the adversary gets: 
    \[
    \set{\prf(\tpk,x),(\tpk+r)\bmod q:\tpk,r\getsr\cK}
    \]
    \describegame In this hybrid, we replace $(\tpk+r)\bmod q$ with a uniformly random value $\tpk'\getsr\cK$.
    \[
    \set{\prf(\tpk,x),{\color{blue}\tpk'}:\tpk,\tpk'\getsr\cK}
    \]
    Note that $(\tpk+r)\bmod q$ and $\tpk'$ are identically distributed. Therefore, the $\Hybrid_0,\Hybrid_1$ are identically distributed. 
    \describegame In this hybrid, we will replace the PRF computation with a random value from the range. 
    \[
    \set{{\color{blue}Y},\tpk':Y\getsr\cY,\tpk'\getsr\cK}
    \]
    Under the security of the PRF, we get that $\Hybrid_1,\Hybrid_2$ are computationally indistinguishable. 
    \describegame We replace $\tpk'$ with $(\tpk+r)\bmod q$. 
    \[
    \set{Y,{\color{blue}(\tpk+r)\bmod q}:Y\getsr\cY,\tpk,r\getsr\cK}
    \]
    As argued before $\Hybrid_2,\Hybrid_3$ are identically distributed. 
\end{gamedescription}
Note that $\Hybrid_3$ is the right distribution from the theorem statement. This completes the proof. 
\end{proof}
\subsection{Learning with Errors Assumption}
\label{sub:cons-prf-lwe}
\begin{construction}[Key Homomorphic PRF from \LWE]
    \label{cons:khprf-lwe}
    Let $\hash_\bfA:\cX\to\bbZ_q^{L\times \lambda}$. Then, define the efficiently computable function $F_\LWE:(\bbZ_q^\lambda\times\chi^L)\times\cX\to\bbZ_q^L$ as $F_\LWE((\bfk,\bfe),x):=\hash_\bfA(x) \bfk+\bfe$. $F_\LWE$ is an almost key homomorphic PRF. 
\end{construction}
\begin{remark}
    We note that it has been shown that one can also sample $\bfx$ from the same distribution as $\bfe$. This is known as the short-secret LWE assumption and was employed in the encryption scheme of Lyubashevsky~\etal~\cite{EC:LyuPeiReg10}. An immediate consequence of this assumption is that one can set $\bfA$'s dimension $m$ to be much smaller than what is required for the LWE assumption. 
\end{remark}
Similar to the~\LWR\ construction, we need to show that the PRF based on \LWE\ assumption is also leakage resilient. Unfortunately, a similar proof technique does not work because the construction also suffers from leakage on the error vector which are usually Gaussian secrets. Instead, we rely on the Hint-LWE Assumption (Definition~\ref{def:hint-lwe}, as before. 


%% file: lattice.tex
\section{Construction from \LWR}
\label{sec:lattice}
\label{sec:lwr-cons}
\subsection{Distributed PRF from \LWR}
\label{sub:tprf-lwr}
Let us revisit Construction~\ref{cons:khprf-lwr}. First, observe that the key space is from $\bbZ_q^\rho$ which implies that the order of $\cK$ is known. Further, the computation occurs over a group whose structure and order is known. This is a departure from the construction based on the \HSM\ assumption. Consequently, by assuming that both $\Prime$ and $q$ are primes, one can avoid integer secret sharing but instead rely on traditional Shamir's Secret Sharing over a field, which we defined earlier (see Construction~\ref{cons:ssf}). 


We saw earlier that Construction~\ref{cons:khprf-lwr} was only almost key homomorphic, i.e.: 
\[
F(\bfk_1+\bfk_2,x)=F(\bfk_1,x)+F(\bfk_2,x)+e
\]
where $e\in\set{0,1}$. It also follows that:
\[
T\cdot F(\bfk_1,x)=F(T\cdot \bfk_0,x)-e_T
\]
where $e_T\in{0,\ldots,T}$. This becomes a cause for concern as, in the threshold construction using the Shamir Secret Sharing over the field as shown in Construction~\ref{cons:ssf}, one often recombines by multiplying with a Lagrange coefficient $\lambda_{i_j}$. Unfortunately, multiplying the result by $\lambda_{i_j}$ implies that the error term $e_{\lambda_{i_j}}\in\set{0,\ldots,i_j}$. The requirement is that this error term should not become ``too large''. However, interpreting Lagrange coefficients as elements in $\bbZ_\Prime$ results in the error term failing to be low-norm leading to error propagation. To mitigate this, we use techniques quite similar to Construction~\ref{cons:ss} by essentially clearing the denominator by multiplying with $\offset:=\csize!$. This is a technique made popular by the work of Shoup~\cite{EC:Shoup00a} and later used in several other works including in the context of lattice-based cryptography by Agrawal~\etal\cite{PKC:ABVVW12} and later to construct a distributed key homomorphic PRF from any almost key homomorphic PRF by Boneh~\etal\cite{C:BLMR13}. \footnote{However, their generic construction is incorrect, owing to issues in their security proof which is not entirely sketched out. We fix the issues in our construction.} Then, the combine algorithm will simply multiply all partial evaluations with $\offset$ as well. 
\begin{construction}[Distributed Almost Key Homomorphic PRF from \LWR]
    \label{cons:tprf-lwr}
         A $(\cthr,\csize)$-Distributed PRF is a tuple of PPT algorithms $\tprf:=(\tpgen,\tpshare,\allowbreak\tpeval,\allowbreak\tppeval,\allowbreak\tpcombine)$ with the algorithms as defined in Figure~\ref{fig:tprf-lwr}.
\end{construction}
\paragraph{Issues with the Construction from Boneh~\etal~\cite[\S 7.1.1]{C:BLMR13}.} As remarked earlier, their generic construction suffers from issues stemming from their security reduction. Specifically, their security reduction proceeds similarly to the proof of Theorem~\ref{thm:tprf} and requires $\cB$ to answer honest evaluation queries for key indices for which it does not know the actual key share. Their explanation suggests that we again use the ``clearing out the denominator'' trick by multiplying with $\offset$. However, the issue is that the resulting response will be of the form $\offset\cdot F(k_{i^\ast},x)$ for $i^\ast$, unknown to $\cB$. Consequently, one has to change the partial evaluation response to also include this offset to ensure the correctness of reduction. This would imply that the $\tpcombine$ algorithm will multiply with $\offset$ again, which would thus result in the actual $\tpeval$ algorithm having an offset of $\offset^2$. Furthermore, the partial evaluation algorithm should also have to round down to the elements in $[0,u-1]$ for the same reason that the $\tpcombine$ algorithm required this fix.  
\begin{figure}[!tb]
    \renewcommand{\algorithmiccomment}[1]{{\vspace*{-0.3em}\color{gray}// #1}}

	\begin{protocolbox}{Distributed PRF from Learning with Rounding Assumption}
		\scalebox{0.8}{
			\begin{minipage}[t]{0.55\linewidth}
			
			    \algoHead{$\tpgen(1^\kappa,1^\cthr,1^\csize)$}
			    \begin{algorithmic}
                \State Parse $\kappa=(\rho)$
                \State Run $\pp_\lwr=(\rho,q,\Prime)\getsr\lwr\Gen(1^\rho)$
                \State Set $\ssize:=|q|$
                \State Set $u$ such that $\floor{\Prime/u}>(\offset+1)\cthr\offset$
                \State Set $v$ such that $\floor{u/v}>\offset\cthr$
                \State Sample $\hash:\cX\to\bbZ_q^{\rho}$
                \State \Return $\pp_{\prf}:=(\pp_\lwr,\pp_{\sss},\hash,u)$
			    \end{algorithmic}
       
			    \algoHead{$\tpshare(\bfk\in\bbZ_q^{\rho},\cthr,\csize)$}
			    \begin{algorithmic}
                        \For{$i=1,\ldots,\rho$}
                            \State $\tpk_i^{(1)},\ldots,\tpk_i^{(\csize)}\getsr\sss.\term{SecretShare}(\tpk_i,\cthr,\csize)$
                        \EndFor
                        \State \Return $\set{\bfk^{(j)}=(\tpk_1^{(j)},\ldots,\tpk_\rho^{(j)}}_{j\in[\csize]}$
			    \end{algorithmic}

		\end{minipage}
		}
		\hfill\scalebox{.8}{
			\begin{minipage}[t]{0.55\linewidth}	
   
       			\algoHead{$\tpeval(\bfk,x)$}
			    \begin{algorithmic}
                \State Compute $Y=\Floor{\offset\Floor{\offset\floor{\angle{\hash(x),\bfk}}_\Prime}_u}_v$
                \State \Return $Y$
			    \end{algorithmic}
			 \algoHead{$\tppeval(\bfk^{(i)},x)$}
			    \begin{algorithmic}
                \State Compute $y_i=\Floor{\offset\Floor{\angle{\hash(x),\bfk^{(i)}}}_\Prime}_u$
                \State \Return $y_i$
			    \end{algorithmic}
			    \algoHead{$\tpcombine(\set{y_i}_{i\in \cS})$}
			    \begin{algorithmic}
                    \State Run $\lambda_{i\in\cS}=\sss.\term{CoefF}(\cS)$
			        \State  Compute $Y'=\floor{\sum_{i\in\cS}{\offset\lambda_i \cdot y_i}}_v$
                    \State \Return $Y'$
			    \end{algorithmic}
			\end{minipage}
		}
	\end{protocolbox}
	\vspace*{-0.7em}
	\caption{Construction of Distributed PRF based on the Secret Sharing scheme of Construction~\ref{cons:ssf} where $\sss=(\term{SecretShare},\term{Coeff})$ with $\pp_\sss$ denoting the public parameters of the secret sharing scheme.}
	\label{fig:tprf-lwr}
\end{figure}
\paragraph{Correctness.} 
\ifdefined\IsSP
{\footnotesize\allowdisplaybreaks
\begin{align*}
    \floor{\offset\floor{\angle{\hash(x),\bfk)}}_\Prime}_u&=\floor{\offset \floor{\sum_{z=1}^{\rho} H(x,z)\cdot \tpk_z}_\Prime}_u\\
\end{align*}}%
{\footnotesize\allowdisplaybreaks
\begin{align*}
    &=\floor{\offset \floor{\sum_{z=1}^{\rho} H(x,z) \sum_{i_j\in\set{i_1,\ldots,i_\cthr}} \lambda_{i_j} s_z^{(i_j)}}_\Prime}_u\\
    &=\floor{\offset\floor{\sum_{i_j\in\set{i_1,\ldots,i_\cthr}} \sum_{z=1}^{\rho} H(x,z)\cdot \lambda_{i_j}\cdot s_z^{(i_j)}}_\Prime}_u\\
    &=\floor{\offset\floor{\sum_{i_j\in\set{i_1,\ldots,i_\cthr}}  \angle{\hash(x),\lambda_{i_j}\bfk^{(i_j)}}}_\Prime}_u\\
    &=\floor{\offset\left(e_\cthr+\sum_{i_j\in\set{i_1,\ldots,i_\cthr}}\floor{\lambda_{i_j}{\angle{\hash(x),\bfk^{(i_j)}}}}_\Prime\right)}_u\\
    &=\floor{\offset\left(e_\cthr+\sum_{i_j\in\set{i_1,\ldots,i_\cthr}}e_{\lambda_{i_j}}+\lambda_{i_j}\floor{{\angle{\hash(x),\bfk^{(i_j)}}}}_\Prime\right)}_u\\
    &=\floor{\offset\left(e_\cthr+\sum_{i_j\in\set{i_1,\ldots,i_\cthr}}e_{\lambda_{i_j}}\right)+\sum_{i_j\in\set{i_1,\ldots,i_\cthr}}\offset~ \lambda_{i_j} \floor{\angle{\hash(x),\bfk^{(i_j)}}}_\Prime}_u\\
    &=\floor{\sum_{i_j\in\set{i_1,\ldots,i_\cthr}}\offset\cdot \lambda_{i_j} \floor{\angle{\hash(x),\bfk^{(i_j)}}}_\Prime}_{u}
\end{align*}}%
\else
{\footnotesize\allowdisplaybreaks
\begin{align*}
    \floor{\offset\floor{\angle{\hash(x),\bfk)}}_\Prime}_u&=\floor{\offset \floor{\sum_{z=1}^{\rho} H(x,z)\cdot \tpk_z}_\Prime}_u\\
    &=\floor{\offset \floor{\sum_{z=1}^{\rho} H(x,z) \sum_{i_j\in\set{i_1,\ldots,i_\cthr}} \lambda_{i_j} s_z^{(i_j)}}_\Prime}_u\\
    &=\floor{\offset\floor{\sum_{i_j\in\set{i_1,\ldots,i_\cthr}} \sum_{z=1}^{\rho} H(x,z)\cdot \lambda_{i_j}\cdot s_z^{(i_j)}}_\Prime}_u\\
    &=\floor{\offset\floor{\sum_{i_j\in\set{i_1,\ldots,i_\cthr}}  \angle{\hash(x),\lambda_{i_j}\bfk^{(i_j)}}}_\Prime}_u\\
    &=\floor{\offset\left(e_\cthr+\sum_{i_j\in\set{i_1,\ldots,i_\cthr}}\floor{\lambda_{i_j}{\angle{\hash(x),\bfk^{(i_j)}}}}_\Prime\right)}_u\\
    &=\floor{\offset\left(e_\cthr+\sum_{i_j\in\set{i_1,\ldots,i_\cthr}}e_{\lambda_{i_j}}+\lambda_{i_j}\floor{{\angle{\hash(x),\bfk^{(i_j)}}}}_\Prime\right)}_u\\
    &=\floor{\offset\left(e_\cthr+\sum_{i_j\in\set{i_1,\ldots,i_\cthr}}e_{\lambda_{i_j}}\right)+\sum_{i_j\in\set{i_1,\ldots,i_\cthr}}\offset\cdot \lambda_{i_j} \floor{\angle{\hash(x),\bfk^{(i_j)}}}_\Prime}_u\\
    &=\floor{\sum_{i_j\in\set{i_1,\ldots,i_\cthr}}\offset\cdot \lambda_{i_j} \floor{\angle{\hash(x),\bfk^{(i_j)}}}_\Prime}_{u}
\end{align*}}%
\fi
The last step follows provided the error term is small. Recall that $e_\cthr\in\set{0,\ldots,\cthr}$ and $e_{\lambda_{i_j}}\in\set{0,\ldots,\lambda_{i_j}}$. Now observe that we multiply with $\offset$ and $\lambda_{i_j}$ has a maximum value $\offset$. Therefore, $\offset \cdot e_{\lambda{i_j}}<\offset^2$. Therefore, the size of the error term is $\leq \cthr\cdot\offset+\cthr\cdot \offset^2$. Therefore, provided $u$ is chosen such that $\floor{p/u}>(\offset+1)\cdot \cthr\cdot \offset$, then the last step is correct. Now, we have:
\begin{align*}
    \floor{\offset\floor{\angle{\hash(x),\bfk)}}_\Prime}_u = \floor{\sum_{i_j\in\set{i_1,\ldots,i_\cthr}}\offset\cdot \lambda_{i_j} \floor{\angle{\hash(x),\bfk^{(i_j)}}}_\Prime}_{u}
\end{align*}

\ifdefined\IsSP
\noindent Or, computing \(\floor{\offset \floor{\offset\floor{\angle{\hash(x),\bfk)}}_\Prime}_u}_v\) we get:

{\small
\begin{align*}
    &=\floor{\offset\floor{\sum_{i_j\in\set{i_1,\ldots,i_\cthr}}\offset\cdot \lambda_{i_j} \floor{\angle{\hash(x),\bfk^{(i_j)}}}_\Prime}_{u}}_v\\
    &=\floor{\offset\left(e_\cthr+\sum_{i_j\in\set{i_1,\ldots,i_\cthr}} \lambda_{i_j} \floor{\offset\cdot \floor{\angle{\hash(x),\bfk^{(i_j)}}}_\Prime}_{u}\right)}_v\\
    &=\floor{\offset\left(e_\cthr+\sum_{i_j\in\set{i_1,\ldots,i_\cthr}} \lambda_{i_j} \tppeval(\bfk^{i_j},x)\right)}_v\\
    &=\floor{\sum_{i_j\in\set{i_1,\ldots,i_\cthr}}\lambda_{i_j}\cdot\offset\cdot \tppeval(\bfk^{(i_j)},x)}_v\\
    &=\tpcombine(\set{\tppeval(\bfk^{(i_j)},x)_{i_j\in\set{i_1,\ldots,i_\cthr}}})
\end{align*}
}
\else
{\small \begin{align*}
    \floor{\offset \floor{\offset\floor{\angle{\hash(x),\bfk)}}_\Prime}_u}_v&=\floor{\offset\floor{\sum_{i_j\in\set{i_1,\ldots,i_\cthr}}\offset\cdot \lambda_{i_j} \floor{\angle{\hash(x),\bfk^{(i_j)}}}_\Prime}_{u}}_v\\
    &=\floor{\offset\left(e_\cthr+\sum_{i_j\in\set{i_1,\ldots,i_\cthr}} \lambda_{i_j} \floor{\offset\cdot \floor{\angle{\hash(x),\bfk^{(i_j)}}}_\Prime}_{u}\right)}_v\\
    &=\floor{\offset\left(e_\cthr+\sum_{i_j\in\set{i_1,\ldots,i_\cthr}} \lambda_{i_j} \tppeval(\bfk^{i_j},x)\right)}_v\\
    &=\floor{\sum_{i_j\in\set{i_1,\ldots,i_\cthr}}\lambda_{i_j}\cdot\offset\cdot \tppeval(\bfk^{(i_j)},x)}_v\\
    &=\tpcombine(\set{\tppeval(\bfk^{(i_j)},x)_{i_j\in\set{i_1,\ldots,i_\cthr}}})
\end{align*}}%

\fi
provided $\floor{u/v}>\cthr\offset$. 
\paragraph{Pseudorandomness.} The proof of pseudorandomness follows the outline of the proof of Theorem~\ref{thm:tprf} but with some important differences. First, we do not rely on integer secret sharing but rather plain secret sharing over the field. Therefore, the Lagrange coefficients correspond to $\lambda_{i_j}$. Or more formally, to respond to a partial evaluation query at point $x_j$ with target key index $i^\ast$, the adversary $\cB$ does the following:
\begin{itemize}
    \item Use its oracle to get partial evaluation on $x_j$ at $i_\cthr$, which we call as $h_{j,\cthr}$. 
    \item Then, use Lagrange coefficients but with suitably multiplying with $\offset$ to compute the correct distribution by rounding down to $u$. The choice of $u$ guarantees that the response is correct. 
\end{itemize}
For challenge query, it simply does two rounding down, first to $u$ and then to $v$. 

\paragraph{Verification of Almost Key Homomorphism.} Let $\bfk_1,\bfk_2$ be two keys that are shared. Now, let the key shares received by some party $i_j$ be $\bfk_1^{(i_j)}$ and $\bfk_2^{(i_j)}$. Then,
\ifdefined\IsSP
{\small\begin{align*}
\tppeval(\bfk_1^{(i_j)},x)+\tppeval(\bfk_2^{(i_j)},x)
\end{align*}
\begin{align*}  
&=\floor{\offset\floor{\angle{H(x),\bfk_1^{(i_j)}}}_\Prime}_u+\floor{\offset\floor{\angle{H(x),\bfk_2^{(i_j)}}}_\Prime}_u\\
&=\floor{\offset\floor{\angle{H(x),\bfk_1^{(i_j)}}}_\Prime+\offset\floor{\angle{H(x),\bfk_2^{(i_j)}}}_\Prime}_u - e_1\\
&=\floor{\offset\floor{\angle{H(x),\bfk_1^{(i_j)}}+\angle{H(x),\bfk_2^{(i_j)}}}_\Prime}_u -2e_1\\
&=\tppeval(\bfk_1^{(i_j)}+\bfk_2^{(i_j)},x)-2e_1\\
\end{align*}}%
\else
{\footnotesize\begin{align*}
\tppeval(\bfk_1^{(i_j)},x)+\tppeval(\bfk_2^{(i_j)},x)&=\floor{\offset\floor{\angle{H(x),\bfk_1^{(i_j)}}}_\Prime}_u+\floor{\offset\floor{\angle{H(x),\bfk_2^{(i_j)}}}_\Prime}_u\\
\end{align*}

\begin{align*}  
&=\floor{\offset\floor{\angle{H(x),\bfk_1^{(i_j)}}}_\Prime+\offset\floor{\angle{H(x),\bfk_2^{(i_j)}}}_\Prime}_u - e_1\\
&=\floor{\offset\floor{\angle{H(x),\bfk_1^{(i_j)}}+\angle{H(x),\bfk_2^{(i_j)}}}_\Prime}_u -2e_1\\
&=\tppeval(\bfk_1^{(i_j)}+\bfk_2^{(i_j)},x)-2e_1\\
\end{align*}}%
\fi
It follows that for $n$ such keys: 
{\begin{align*}
\sum_{i=1}^{n}\tppeval(\bfk_i^{(i_j)},x) &=\tppeval(\sum_{i=1}^{n} \bfk_i^{(i_j)},x)-n\cdot e_1\\
\end{align*}}%
This shows that the $\tppeval$ is almost key-homomorphic. Consequently, one can verify that the whole $\tpeval$ procedure is almost key homomorphic for the appropriate error function. To do this, recall that from the correctness of our algorithm: 
{\small\[
\tpshare(\bfk,\csize,\cthr)=\set{\bfk^{(i)}}_{i\in[\csize]},\tpcombine(\set{\tpeval(\bfk^{(i_j)},x)}_{i_\in\set{i_1,\ldots,i_t}}=\tpeval(\bfk,x)
\]
}%
In other words, for $\bfk_1,\bfk_2\in\cK$, $\tpshare(\bfk_1,\csize,\cthr)=\set{\bfk_1^{(i)}}_{i\in[\csize]}$ and $\tpshare(\bfk_2,\csize,\cthr)=\set{\bfk_2^{(i)}}_{i\in[\csize]}$, we will have for $i\in[\csize]$, $\tpeval(\bfk_1^{(i)},x),\tpeval(\bfk_2^{(i)},x)=\tpeval(\bfk_1^{(i)}+\bfk_2^{(i)},x)-2\cdot e_1$. 

\subsection{$\caps'$ Construction based on \LWR\ Assumption}
\label{sub:def-caps-lwr}
We build $\caps$ based on the \LWR\ Assumption, building it based on the Key Homomorphic, Distributed PRF as presented in
Construction~\ref{cons:khprf-lwr}. However, our construction is largely different from the template followed to build $\caps$
from the \HSM\ assumption. 
This is primarily because of the growth in error when combining partial evaluations. Specifically,
will get that $\tppeval(\sum_{i=1}^{n}\bfk_i^{(j)},x)=\sum_{i=1}^{n}{\tppeval(\bfk_i^{(j},x)}+e$ where $e\in\set{0,\ldots,n-1}$ where $n$ is the number of clients participating for that label. This would require us to round down to a new value $u'$ such that $\floor{u/u'}>n-1$. Therefore, while we still employ the underlying functions of the distributed, key-homomorphic PRF based on LWR, we have to open up the generic reduction. For simplicity, we detail the construction for $L=1$. Then, these are the differences:
\begin{itemize}
    \item As done fore Construction~\ref{cons:caps-lwr}, the input is encoded as $x_i\cdot n+1$. 
    \item Specifically, the client's share to the committee will only be the first level of the evaluation, i.e., rounded down to $\Prime$.
    \item Then, committee member $j$ will then add the shares up, multiply with the offset, and then round down to $u$. We will show that provided $\floor{\Prime/u}>\offset\cdot n$, this is consistent with $\tppeval(\sum_{i=1}^{n}\bfk_i^{(j)},x)$.
    \item The decoding algorithm is also similar to the one from Construction~\ref{cons:caps-lwr}. 
\end{itemize}
Recall that $\tprf$ correctness requires that $\floor{\Prime/u}>\cthr\cdot \offset+\cthr\cdot \offset^2$, and so one just needs $\floor{\Prime/u}>\max(\cthr\cdot \offset+\cthr\cdot \offset^2,n\cdot\offset)$. Then, one can rely on the correctness of $\tprf$ as shown below to argue that when the server runs $\tprf.\tpcombine$, the output is $\tpeval(\sum_{i=1}^{n}\bfk_i,x)$. 

\paragraph{Correctness.} We showed how the output of the server's invocation of $\tprf.\tpcombine$ is $\tpeval(\sum_{i=1}^n \bfk_i,x)$. 
Now, let us look at the remaining steps:
{\footnotesize
\begin{align*}
    X_\lab&=\sum_{i=1}^n \psact_{i,\lab}-\AUX_\lab\\
    &=\sum_{i=1}^n (\psain_{i,\lab}* n +1) + \tpeval(\bfk_i,\lab)-\tpeval(\sum_{i=1}^n \bfk_i,\lab)\bmod v\\
    &=n \cdot \sum_{i=1}^n \psain_{i,\lab} + n +\sum_{i=1}^n \tpeval(\bfk_i,\lab)-\tpeval(\sum_{i=1}^n \bfk_i,\lab)\bmod v\\
    &=n \cdot \sum_{i=1}^n \psain_{i,\lab} + n +\tpeval(\sum_{i=1}^n \bfk_i,\lab)-\tpeval(\sum_{i=1}^n \bfk_i,\lab)-e_{n-1}\bmod v\\
    &=n \cdot \sum_{i=1}^n \psain_{i,\lab} + n - e_{n-1} \bmod v\\
    &=n \cdot \sum_{i=1}^n \psain_{i,\lab} + n - e_{n-1}
\end{align*}
}
For the last step to hold, we need that
\[
0\leq n \cdot \sum_{i=1}^n \psain_{i,\lab} + n - e_{n-1} < v
\]
$e_{n-1}$ the small value is 0 and the largest value is $n-1$ which requires that $\sum_{i=1}^n \psain_{i,\lab}<(v-n)/n$. Note that we already require $\floor{p/u}>n\offset,\floor{u/v}>\cthr\offset \implies \floor{p/v}>n\cthr\offset^2$. In other words, $\sum \psain_i<\frac{p}{n^2\cthr\offset^2}$. Finally, $X_\lab'=n\cdot \sum_{i=1}^{n}\psain_{i,\lab}+n$ and that completes the remaining steps.

%% file: deferred-proofs.tex
\section{Deferred Proofs}
\label{sec:def-proof}
\label{sub:def}
\ifdefined\IsSub{}
\lrprglwr*
\begin{proof}
The proof proceeds through a sequence of hybrids. 
\begin{gamedescription}[name=Hybrid,nr=-1]
    \describegame The left distribution is provided to the adversary. In other words, the adversary gets: 
    \[
   \set{\prg_\LWR(\seed)\bmod p, \seed+\seed'\bmod q:\seed,\seed'\getsr\bbZ_q^{n_1}}
    \]
    \describegame In this hybrid, we replace $\seed+\seed'\bmod q$ with a uniformly random value $\seed''\getsr \bbZ_q^{n_1}$.
    \[
    \set{\prg_\LWR(\seed)\bmod p,{\color{blue}\seed''}:\seed,\seed''\getsr\bbZ_q^{n_1}}
    \]
    Note that $(\seed+\seed')\bmod q$ and $\seed''$ are identically distributed.     Let us assume that there exists a leakage function oracle $L$ that can be queried with an input $\seed$, for which it either outputs $\seed+\seed'\bmod q$ for a randomly sampled $\seed'\getsr\bbZ_q^{n_1}$ or outputs $\bfs'\getsr\bbZ_q^{n_1}$. If one could distinguish between hybrids $\Hybrid_0,\Hybrid_1$, then one could distinguish between the outputs of the leakage oracle, but the outputs are identically distributed. Therefore, the $\Hybrid_0,\Hybrid_1$ are identically distributed. Therefore, the $\Hybrid_0,\Hybrid_1$ are identically distributed. 
    \describegame In this hybrid, we will replace the PRG computation with a random value from the range. 
    \[
    \set{{\color{blue}\bfy},{\seed''}:\bfy\getsr\bbZ_p^L,\seed''\getsr\bbZ_q^{n_1}}
    \]
    Under the security of the PRG, we get that $\Hybrid_1,\Hybrid_2$ are computationally indistinguishable. 
    \describegame We replace $\seed''$ with $(\seed+\seed')\bmod q$. 
    \[
    \set{\bfy,{\color{blue}(\seed+\seed')\bmod q}:\bfy\getsr\bbZ_p^L,\seed,\seed''\getsr\bbZ_q^{n_1}}
    \]
    As argued before $\Hybrid_2,\Hybrid_3$ are identically distributed. 
\end{gamedescription}
Note that $\Hybrid_3$ is the right distribution from the theorem statement. This completes the proof. 
\end{proof}
\fi
\subsection{Simulation-Based Proofs of Security}
\ifdefined\IsPRF{}
\thmopalwr*
\begin{proof}
We will prove the theorem statement by defining a simulator $\tSim$, through a sequence of hybrids such that the view of the adversary $\cA$ between any two subsequent hybrids are computationally indistinguishable. Let $H=[n]\setminus \kcorr$, which are the set of honest clients. Further, let $\cC=[n]\setminus D$ where $D$ is the set of dropout clients. 

It is important to note that the server is semi-honest. Therefore, it is expected to compute the set intersection of online clients $\cC$, as expected. In other words, all committee members (and specifically the honest committee members) receive the same $\cC$. This is an important contrast from active adversaries as a corrupt and active server could deviate from expected behavior and send different $\cC^{(j)}$, for different committee members. This could help it glean some information about the honest clients. 

We now sketch the proof below:
\begin{gamedescription}[name=Hybrid,nr=-1]
    \describegame This is the real execution of the protocol where the adversary $\cA$ interacts with the honest parties. 
    \describegame In this hybrid, we will rely on the security of the secret sharing scheme to do two things:
    \begin{itemize}
        \item On the one-hand, all corrupt committee members receive a random share from the honest client's seed. Note that there can be only a maximum of $\cthr$ corrupt committee members. By appropriately choosing $\csize$, conditioned on $\eta$, we can guarantee that this holds with overwhelming probability. Then, for an honest client $i$, these are the shares denoted by $\set{\seed_i^{(j)}}_{j\in[\csize]\cap \kcorr}$ and are generated randomly. 
        \item On the other hand, all the honest committee members receive a valid share of the honest client's seeds. However, each honest client $i$ need to generate this from a polynomial $p(X)$ that satisfies $p(0)=\seed_i$, while also ensuring $p(j)=\seed_{i}^{(j)}$ for $j\in[\csize]\cap \kcorr$. Note that this is a polynomial time operation and is similar to the way packed secret sharing is done where multiple secrets are embedded at distinct points of the polynomial. See Construction~\ref{cons:pssf} for how to build such a polynomial.
    \end{itemize}
    It is clear that by relying on the privacy of the secret sharing scheme, $\Hybrid_0,\Hybrid_1$ are indistinguishable from the adversary. Specifically, we guarantee that under the IND-CPA security of the public key encryption scheme, the adversary only receives an insufficient number of shares, thereby ensuring the privacy of the secret. 
    \describegame In this hybrid, we change the definition of the last honest party's ciphertext. WLOG, let $n$ be the last honest party in $\cC$. Then, we will set $\bpsact_n:=\prg.\prge(\seed_\lab)+\bpsain_\tau-\sum_{i\in\cC\cap H} \bpsact_i$. Here, $\bpsain_\tau$ is the sum of the honest clients inputs. We are still in the hybrid where $\tSim$ knows all the inputs. 

    It is clear that $\Hybrid_1,\Hybrid_2$ are identically distributed, by the almost seed-homomorphism property of $\prg$, provided $\tSim$ chooses the inputs for the honest parties such that they sum up to the value in $\bpsain_\tau$. 
    \describegame Again, without loss of generality, let client 1 be the first honest client in $\cC\cap H$. We will modify the way $\bpsact_{1,\lab}$ is generated. We will set it as $\bpsact_{1,\lab}:=\bpsain_{1,\lab}+u$ where $u\getsr\prg.\cY$. 

    $\Hybrid_2,\Hybrid_3$ are indistinguishable, provided Theorem~\ref{thm:lr-prg-lwr} holds. In the reduction, we will implicitly set $\seed_1+\seed_n$ to be the leakage obtained from the Theorem~\ref{thm:lr-prg-lwr}'s challenger. In this hybrid, $\tSim$ still continues to know all the inputs. If it was a real PRG output, then we can simulate $\Hybrid_2$, while simulating $\Hybrid_3$ in the random case. 
    \describegame In this hybrid, we will replace $\bpsact_{1,\tau}:=u'$ for $u'\getsr\prg.\cY$. It is clear that $\Hybrid_3,\Hybrid_4$ are identically distributed. 
\end{gamedescription}
At this point, observe that we have successfully replace the first honest client's ciphertext, with a uniformly random value that is independent of its input. $\tSim$ will continue to do this modification for every non-dropout honest client $i\in\cC\cap H$. This leaves the clients with all-but-the-last honest clients' ciphertext to be independent of the input, while leaving the last honest client's ciphertext to be only a function of the sum of the inputs, which can be obtained by $\tSim$'s query to the functionality. 
$\tSim$ beings its interaction with the functionality. After all the honest clients have provided inputs to the trusted party $\cT$, in Step~\ref{stepb}, $\tSim$ does not instruct any corrupted client to abort but rather set their inputs to be 0. Then in Step~\ref{stepc}, $\tSim$ does not abort the server. Therefore, in Step~\ref{stepd}, $\tSim$ will learn the sum of the honest parties inputs. Denote it as $\bpsain_\lab$, which is also the sum of the inputs of the honest, surviving clients. With this information, $\tSim$ uses the last hybrid to interact with the adversary $\cA$, who's expecting the real world interaction. This will enable $\tSim$ to run $\cA$ internally. This is crucial to ensure that $\tSim$ can get the output of $\cA$, in the real world, which might depend on its view (including the output) of the server. This view will, in turn, depend on the the honest clients' inputs. Since $\tSim$ sets the honest inputs, in this internal execution, to match the sum of inputs in the real world, we can guarantee that the output of $\cA$ in the internal simulation is indistinguishable from $\cA$'s interaction in the real world by the aforementioned hybrid arguments. 

\paragraph{Proof of Security against Active Server.}
Our constructions so far have relied on providing security against
a semi-honest server. Note that, as shown in the proof of security for Theorems~\ref{thm:opa-lwr}, we can use the functionality query to obtain the sum of all the honest non-drop out clients, as before. 
     
In the semi-honest setting, it is easy to see that the set $\cC$, with respect to which aggregation is performed, includes \emph{all} the honest, non-dropout clients' inputs. Therefore, querying the functionality, $\tSim$ does indeed get the sum of all the honest clients' inputs that are also included in the summation in the real world. This is imperative to ensure that $\tSim$, when internally invoking $\cA$, can get the output of $\cA$ which should be indistinguishable from $\cA$'s output in the real world. Specifically, this output of $\cA$ (in either the internal invocation or the actual execution) will depend on the view which consists of the output of the server. Therefore, if the output of the server in the real world does not include any of these honest clients' inputs, then the output produced by the internal invocation of $\cA$ can be different from that in the real world. 

Let us look at the case when the server is corrupted. Such a server can mount an attack whereby the real-world execution of the protocol may exclude inputs of some of those honest parties but actually included in the output of the ideal functionality. The proof of malicious security is tricky in this setting. Specifically, a malicious server can drop clients after seeing the honest input. This is an issue in the simulation as the simulator has to generate the masked inputs for the honest clients without knowing which of them would be dropped later. 

Prior works, beginning with that of Bonawitz~\etal~\cite{bonawitz2017practical} have relied on using signatures to ensure that a malicious server does not compromise the privacy of an honest user. Fortunately, for $\caps$, we can rely on the one-shot nature of communication flow to secure messages and avoid using signatures.  

As before, let $\kcorr$ denote the corrupted clients.
Then, $\khon_{\term{Cli}}:=[n]\setminus \kcorr$ is the set of honest clients, $\khon_\term{Com}:=[\csize]\setminus\kcorr$ is the set of honest
committee members. Let $\kcorr_{\term{Cli}}:=[n]\cap \kcorr$ denote the corrupted clients and $\kcorr_{\term{Com}}:=[\csize]\cap \kcorr$ denotes the corrupted
committee members. 

Note that we do not rely on signatures. To achieve a protocol with signatures, there needs to be an additional round of communication between the committee members and the server. First, the server forwards the message to the committee members. Then, the committee members responds with their set $\cC^{(j)}$, which is also duly signed. Then, the server performs the intersection and contacts the committee member with this intersection along with signatures. A committee member then only aggregates if there are $(\csize+\cthr)/2$ valid signatures. 

Our focus is to ensure that the committee members only speaks once. In other words, our construction currently has the server identify $\cC^{(j)}$, for each $j\in[\csize]$, based on the information it has received from the client. Then, the server forwards the message to the committee member along with its computed intersection. This setting allows the server to selectively forward shares to committee member and also choose different sets for different committee members. We will show that if $\crec>(\csize+\cthr)/2$ where $\crec$ is the reconstruction threshold in the committee and $\cthr$ is the corruption threshold, then the server doing so will receive meaningless information. Formally, we will show that  there does not exist two sets of users $\cC\neq \cC'$ such that the server can reconstruct the shares over these two sets. 

Observe that the server controls $\cthr$ committee members. We require each honest committee member to participate once, per iteration. This is easily enforced as the share from the honest client encrypts, along with the share for the honest committee member, also the identity of the honest client and the iteration count. Therefore, a server cannot replay the same share, in another iteration. With this guarantee, a malicious server, in order to reconstruct the shares of two distinct sets $\cC,\cC'$, will require the cooperation of at least $\crec-\cthr$ honest users, while there are $\crec-\cthr$ honest users present. We will therefore need $2(\crec-\cthr)>\csize-\cthr$. Or, $\crec>(\csize+\cthr)/2$. This ensures that the server can only effectively reconstruct with respect to a unique set $\cC$ and $\khon^\ast$ is the set of honest users in this set. 
Note that the above inequality holds for $\crec=2*\csize/3, \cthr<\csize/3$. Indeed, prior works such as Bonawitz~\etal~\cite{bonawitz2017practical} and most recently LERNA~\cite{AC:LLPT23} also tolerated only upto a $\csize/3$ corruption threshold. 

While we have shown that there is a unique set $\khon^\ast$ of honest users, $\khon^\ast$ is only revealed after all the honest clients have sent their inputs. Therefore, the simulator, during its internal execution of $\cA$, needs to be able to generate the masked inputs for the honest users and it only knows the sum of \emph{all} the honest clients that have not dropped out. This set may be distinct from $\khon^\ast$. Therefore, we need a way for the simulator to generate masked inputs, independent of the sum of the inputs, and then ensure that the correct sum is computed during reconstruction.

The simulator does the following:
\begin{itemize}
    \item For every honest client that hasn't dropped out, i.e., for all $i\in \khon_{\term{Cli}}$, the simulator does the following:
    \begin{itemize}
        \item Samples $\hseed_{i,\lab}\getsr\{0,1\}^{\log q}$
        \item Samples $\seed_{i,\lab}\getsr\prg.\cK$
        \item It computes $\bmask_{i,\lab}':=\hash(\hseed_{i,\lab}$
        \item Computes $\bmask_{i,\lab}=\prg.\prge(\seed_{i,\lab})$
        \item Sets $\bpsact_{i,\lab}:=\bmask_{i,\lab}+\bmask_{i,\lab}'$
        \item Like shown in proof of Theorem~\ref{thm:opa-lwr}, the shares of $\seed_{i,\lab}^{(j)},\allowbreak\hseed_{i,\lab}^{(j)}$ for corrupt committee members $j\in \kcorr_{\term{Com}}$ are chosen at random. Meanwhile, the shares for the honest committee members are to be sampled in the second phase, with a specific purpose. However, the server still expects an encryption of shares from honest client to honest committee members. Therefore, it simply encrypts some random shares for the honest committee members too and sends it to the server.
        \item It sends to $\cA$, $\bpsact_{i,\lab}$ and $\seed_{i,\lab}^{(j)}$ and $\hseed_{i,\lab}^{(j)}$ for $j\in[\csize]$, which is encrypted appropriately. 
    \end{itemize}
    \item This concludes the client phase of the operation. Then, comes the interaction with the committee. Note that the simulator is also required to simulate the honest committee member $j$. 
    \item The simulator, which has received $\cC^{(j)}$ for each honest committee member $j$ does the check to make sure that there exists at least $\crec-\cthr$ such committee members with the same $\cC^{(j)}$. We will call this client set as $\cC$, while calling the set of these committee members to be $C_{\term{good}}$. 
    Meanwhile, it records those committee members with a different $\cC{(j)}$. We will call this as some set $C_{\term{bad}}$. 
    Looking ahead, for those honest committee members in $C_{\term{bad}}$, the shares of the honest clients that are to be added up is going to be random values. Note that $|C_\term{bad}|\leq \csize-\crec$. 
    \item The simulator now operates in two phases for honest committee member $j\in\khon_{\term{Com}}$. First is the share generation phase for honest clients $i$. It does the following:
    \begin{itemize}
        \item If $j\in C_{\term{bad}}$, then for honest client $i\in \cC^{(j)}\cap \khon_{\term{Cli}}$, set $\seed_{i,\lab}^{(j)},\hseed_{i,\lab}^{(j)}$ to be random values. 
        \item Now, the simulator computes the shares for all honest clients $i$ to $j\in C_{\term{good}}$. These are valid shares of $\seed_{i,\lab},\hseed_{i,\lab}$ subject to the constraint that random values were fixed for those $j\in C_{\term{bad}}$ where $i\in\cC^{(j)}$, and for those $j\in \kcorr_{\term{Com}}$. 
        \item The honest committee member $j$ receives from $\cA$, $\seed_{i,\lab}^{(j)}$ and $\hseed_{i,\lab}^{(j)}$ for $i\in\cC^{(j)}$. Note that the maximum number of prefixed values is $\csize-\crec+\cthr$, and by our constraint $\crec>\csize-\crec+\cthr$ which guarantees that these prefixed values cannot uniquely determine a polynomial of degree $\crec$. 
    \end{itemize}
    \item The second phase, is the combination phase. It responds, as expected, subject to the set $\cC^{(j)}$ that it receives. 
    \item $\tSim$ now queries the functionality. First, it provides $[n]\setminus\cC$ as the set of dropped out clients. Then, it sends for those corrupted clients $\kcorr\cap \cC$, input as 0 to the functionality. In response, it gets
    $\sum_{\cC\cap H} \bpsain_{i,\lab}$. Call this $\bpsain_{H}$.
    \item $\tSim$ now picks $i^\ast\in \khon^\ast$. It programs the random oracle by setting $\hash(\hseed_{i^\ast,\lab})=\bpsain_{H}-\bmask_{i^\ast,\lab}'$.
    \item $\tSim$ continues to respond, on behalf of the honest committee member, as expected. 
    \item Finally, $\cA$ (which controls the server) will make queries to random oracle and it answers as expected. $\tSim$ outputs whatever $\cA$ outputs at the end. 
\end{itemize}
We will now need to show that the above simulation is indistinguishable from the real world execution that $\cA$ expects when it is internally run. The hybrids proceeds as follows:
\begin{gamedescription}[name=Hybrid,nr=-1]
    \describegame This is the real world execution. 
    \describegame In this execution, we replace the shares sent by the honest client $i$ to honest committee member $j$, which are encrypted under $\upkepk_j$ with a random value. Under the semantic security of this encryption scheme, we can guarantee that this is indistinguishable from the previous hybrid. Meanwhile, these honest committee members (which the simulator controls) will receive the shares directly from the simulator. The view of $\cA$, in this hybrid, is indistinguishable from the real world execution, under the semantic security of the encryption scheme. 
    \describegame We will rely on the security of the secret sharing scheme to sample the shares for the honest clients, similar to $\Hybrid_1$ of semi-honest security. For those $j\in\kcorr_{\term{Com}}$, the shares are randomly chosen. 
    Furthermore, for those $j\in C_{\term{bad}}$ also the shares are randomly chosen. Finally, for those $j\in C_{\term{good}}$ it gets a valid share subject to those previously chosen random values. This is similar to $\Hybrid_1$ in the proof of semi-honest security. 
    \describegame In this hybrid, for all those honest clients $i$ that are not in $\cC$, we will set $\bpsact_{i,\lab}=\bmask_{i,\lab}+\bmask_{i,\lab}'$, effectively setting the input to be 0. Observe that the view of $\cA$ remains unchanged as these honest clients inputs were never incorporated in the final sum anyway. Furthermore, if any of these $i\in\cC^{(j)}$ for $j\in C_{\term{bad}}$, the shares from these honest clients $i$ to these $j$ are completely random \emph{and} independent of $\bmask_{i,\lab}$ and $\bmask_{i,\lab}'$. 
    \describegame In this hybrid, we pick an honest surviving client $i^\ast\in\khon^\ast$. It sets the inputs for all $i\neq i^\ast \in \khon^\ast$ to be 0. Then sets $\bpsain_{i^\ast,\lab}$ to be the sum of all the $i\in\khon^\ast$. Call this sum as $\bpsain_{H}$. Observe that the values are still correlated and pseudorandom. 
    \describegame In this hybrid, we will program $\hash(\hseed_{i^\ast,\lab})=\bpsain_H-\bmask_{i^\ast,\lab}'$, while setting $\bpsact_{i^\ast,\lab}=\bmask_{i^\ast,lab}+\bmask_{i^\ast,\lab}'$. Note that because $\hseed_{i^\ast,\lab}$ is chosen uniformly at random from $q$ values where $q$ is a large prime. The probability of collision is negligible. There is only negligible difference in the view of $\cA$.
    \describegame In this hybrid, we will set $\bpsact_{i,\lab}$ for $i\neq i^\ast$ to be some random term in the ciphertext space. Then, we will set $\bpsact_{i^\ast,\lab}=\hash(\hseed_{i^\ast,\lab})-\sum_{i\neq i^\ast}\bpsact_{i,\lab}$.
    
    Note that under the leakage resilience property of the seed-homomorphic PRG, we can conclude that the two hybrids are computationally indistinguishable. 
    \describegame In this hybrid, we will replace $\bpsact_{i,\lab}=\bmask_{i,\lab}+\bmask_{i,\lab}'$ for $i\neq i^\ast$. 
\end{gamedescription}
Observe that this last hybrid is exactly what the simulator produces. This concludes the proof. 
\end{proof}
\fi

\ifdefined\IsSub{}
\begin{restatable}{theorem}{thmopalwe}
    \label{thm:opa-lwe}
     Let $\delta,\eta$ (resp. $\delta_C,\eta_C$) be the dropout and corruption fraction among the clients chosen for summation (resp. clients in committee). Let $\kappa$ be the security parameter. Let $N$ be the total universe of clients and $n$ be the number of clients chosen for summation in each iteration while $\csize$ be the number of committee clients chosen to help in each iteration. Let $L$ be the length of the vector. 

    Let $\prg=(\prg.\tpgen,\prg.\prge)$ be the leakage-resilient, seed-homomorphic PRG defined in Construction~\ref{cons:shprg-lwe} and $\ss=(\ss.\tpshare,\allowbreak\ss.\recons)$ be the $(\cthr,\crec,\csize)$-secret sharing scheme such that $\crec>(\csize+\cthr)/2)$ defined in Construction~\ref{cons:pssf}. Further, assuming a PKI (or authenticated channels) where each client knows a public key $\pkepk_j$ for a committee member $j$, associated with an IND-CPA secure public key encryption scheme $\pke$. Then, if $\delta_C+\eta_C<1/3$,
    $\caps_\LWR$ securely realizes the functionality $F_{D,\delta}^\lab(X)$  (defined in Equation~\ref{eq:functionality}) with server malicious security 
    with abort where $X=\set{\bpsain_{i,\lab}}_{i\in[n]-\setminus \kcorr}$ and $\kcorr\subset [N]$ and $|\kcorr|\cap [n]\leq \eta n$, under the Hint-\LWE\ assumption.
\end{restatable}
\begin{proof}
    The proof proceeds similar to that of Theorem~\ref{thm:opa-lwr}, through a sequence of hybrids. However, there are a few differences.
    Construction~\ref{cons:shprg-lwe} has the error vector $\bfe\getsr\chi$. However, we will replace $\bfe=\bfe'+\bff'$ where $\bfe',\bff'\getsr\chi'$, the distribution present in Hint-LWE Assumption (see Definition~\ref{def:hint-lwe}). The hybrid descriptions are similar, so we only specify the differences: 
    
    \begin{itemize}
        \item In $\Hybrid_2$ we will set: 
         $$\bpsact_{n,\lab}=\bfA\cdot \seed_\lab-\sum_{i\in(\cC\cap H)\setminus \set{n}} \bpsact_i +\bfe_\lab+\Delta \bpsain_{\lab}$$
         \item We will argue that $\Hybrid_2,\Hybrid_3$ are indistinguishable under Hint-LWE Assumption. We will sketch the reduction now. 
         \begin{itemize}
             \item  Recall that, from the Hint-LWE Challenge, we get $(\bfA,\bfu^\ast,\bfs^\ast:=\bfs+\bfr,\bfe^\ast:=\bfe'+\bff')$. 
             \item As done for the $\LWR$ construction, we will set the $\bfs_1+\bfs_n=\bfs^\ast$, the leakage on key.
             \item  For generating $\bpsact_{1,\lab}$ we will use $\bfu^\ast$, while also sampling a separate $\bff_1\getsr\chi'$. This gives us: $\bpsact_{1,\lab}=\bfu^\ast+\bff_{n,\lab}'+\Delta\cdot \bpsain_{1,\lab}$
             \item We will set $\bpsact_{n,\lab}:=\bfA\cdot \seed_\lab-\sum_{i\in(\cC \cap H)\setminus \set{n}} \bpsact_i+\bfe^\ast+\sum_{i\in(\cC \cap H)\setminus \set{1}} (\bfe_{i,\lab}'+\bff_{i,\lab}')+\Delta\cdot \bpsain_{\lab}$
             \item When $\bfu^\ast$ is the real sample, then $\bpsact_{1,\lab}$ satisfies $\Hybrid_2$'s definition. Meanwhile, the $\bpsact_{n,\lab}$ is also correctly simulated. Similarly, the case when it a random sample. 
         \end{itemize}
    \end{itemize}
    The proof of security against malicious server also follows that of the previous theorem. 
    \end{proof}
\fi

%% file: deferred-prfs.tex
\ifdefined\IsSub{}
\else 
\subsection{Proof of Theorem~\ref{thm:khprf}}
\label{sub:thm-1}
\khprf*

\begin{proof}
We denote the challenger by $\cB$. Let $S_j$ be the event that the adversary wins in $\Hybrid_j$ for each $j\in\{0,\ldots,2\}$. Let $q_e$ (resp. $q_h$) denote the number of evaluation queries (resp. hash oracle queries) that the adversary makes. We use an analysis similar to the technique by Coron~\cite{C:Coron00}.
\begin{gamedescription}[name=Hybrid, nr=-1]
\describegame Corresponds to the security game as defined for security of PRF. It follows that the advantage of the adversary is $$\Adv_0=2\cdot |\Pr[S_0]-1/2]=\term{Adv}_\cA^{\term{PRF}}$$

\describegame This game is identical to $\Hybrid_1$ with the following difference. The challenger tosses biased coin $\delta_t$  for each random oracle query $H(t)$. The biasing of the coin is as follows: takes a value 1 with probability $\frac{1}{q_{e}+1}$ and 0 with probability $\frac{q_{e}}{q_{e}+1}$. Then, one can consider the following event $E$: that the adversary makes a query to the random oracle with $x_i$ as an input where $x_i$ was one of the evaluation inputs and for this choice we have that $\delta_t$ was flipped to 0. 

If $E$ happens, the challenger halts and declares failure. Then, we have that:
\[
\Pr[\neg E]=\left(\frac{q_{e}}{q_{e}+1}\right)^{q_{e}}\geq \frac{1}{e(q_{e}+1)}
\]
where $e$ is the Napier's constant. 
Finally, we get that:
\[
\Pr[S_1]=\Pr[S_0]\cdot \Pr[\neg E] \geq \frac{\Pr[S_0]}{e(q_{e}+1)}
\]
\describegame This game is similar to $\Hybrid_1$ with the following difference: we modify the random oracle outputs. 
\begin{itemize}
    \item If $\delta_t=0$, the challenger samples $w_t\getsr\cD_H$ and sets $H(t)=h^{w_t}$
    \item If $\delta_t=1$, the challenger samples $w_t\getsr\cD_H,u_t\getsr\allowbreak\bbZ/\modulus\bbZ$ and sets $H(t)=h^{w_t}\cdot f^{u_t}$
\end{itemize}
Note that, under the HSM assumption, an adversary cannot distinguish between the two hybrids. Therefore, we get:
\[
|\Pr[S_2]-\Pr[S_1]|\leq \epsilon_{\HSM}
\]
where $\epsilon_{\HSM}$ is the advantage that an adversary has in the $\HSM$ game. 
Note that $\Hybrid_2$ corresponds to the case where the outputs are all random elements in $\bbG$. Therefore, the inputs are sufficiently masked and leak no information about the key. Therefore, $\Pr[S_2]=0$
Then,
\[
\Adv_\cA^{\term{PRF}}\leq (e\cdot (q_e+1)\cdot \epsilon_{\HSM}
\]
\end{gamedescription}
\end{proof}
\begin{remark}
    Note that the above scheme is simply an adaptation of the famous DDH-based construction of a key-homomorphic PRF that was shown to be secure by Naor~\etal\cite{EC:NaoPinRei99}. It is easy to verify that our construction is also key homomorphic as $H(x)^{(k_1+k_2)}=H(x)^{k_1}\cdot H(x)^{k_2}$. 
\end{remark}
\fi
\ifdefined\IsSub{}
\else \subsection{Distributed Pseudorandom Function}
\label{sub:tprf}
\dprf*
\paragraph{Correctness.} For a polynomial $f\in\bbZ[X]$, every $f(i)$ leaks information about the secret $\secret\bmod i$ leading to a choice of polynomial $f$ such that $f(0)=\Delta\cdot \secret$. For our use case, the secret is the PRF key $\tpk$. Let us consider a set $\cS=\set{i_1,\ldots,i_\cthr}$ of indices and corresponding evaluations of the polynomial $f$ at $i_1,\ldots,i_\cthr$ giving us key shares: $\tpk^{(i_1)},\ldots,\tpk^{(i_\cthr)}$. To begin with, one can compute the Lagrange coefficients corresponding to the set $\cS$ as: $\forall~i\in\cS,\lambda_i(X):=\prod_{j\in\cS\setminus\set{i}} \frac{x_j-X}{x_j-x_i}$. This implies that the resulting polynomial is $f(X):=\sum_{j=1}^{\cthr} \lambda_{i_j}(X)\cdot \tpk^{(i_j)}$. 

However, $\lambda_i(X)$ requires one to perform a division $x_j-x_i$ which is undefined as $\hash$ hashes to $\bbG$ whose order is unknown. To avoid this issue, a standard technique is to instead compute coefficient $\Lambda_i(X):=\Delta\cdot \lambda_i(x)$. Thereby, the resulting polynomial that is reconstructed if $f'(X)=\Delta\cdot f(X)=\sum_{j=1}^{\cthr} \Lambda_{i_j}(X)\cdot \tpk^{(i_j)}$. Consequently, 

\begin{align*}
    \hash(x)^{\Delta^3\cdot \tpk}=\hash(x)^{\Delta\cdot f'(0)}&=\hash(x)^{\Delta\cdot \sum_{j=1}^{\cthr} \Lambda_{i_j}(0)\cdot \tpk^{(i_j)}}\\&=\prod_{j=1}^t \left(\hash(x)^{\Delta\cdot \tpk^{{(i_j)}}}\right)^{\Lambda_{i_j}(0)}\\
     &=\prod_{j=1}^t \left(\tppeval(\tpk^{(i_j)},x)\right)^{\Lambda_{i_j}(0)}
\end{align*}
Thus, our protocol is correct. 

\paragraph{Pseudrandomness.} Next, we consider the pseudorandomness property of our construction. 
 \thmdprf*
Boneh~\etal\cite{C:BLMR13} showed that from any Key Homomorphic PRF (which Construction~\ref{cons:khprf}), one can build a Distributed PRF. The proof of the following theorem follows the template of this scheme with certain important adaptations as our secret sharing scheme is over integers. The proof technique is to show that if there exists an adversary $\cA$ that can break the $\tprf$ security, one can then use it to build an adversary $\cB$ to break the pseudorandomness of our original PRF, as defined in Construction~\ref{cons:khprf}. The idea behind the proof is for $\cB$, upon receiving choice of $\cthr-1$ corruptions as indices $i_1,\ldots,i_{\cthr-1}$, to then choose a random index $i_\cthr$ and implicitly set $\tpk^{(i_\cthr)}$ to be the PRF key chosen by its challenger. Therefore, the $\cB$ now has knowledge of $\cthr$ indices, with which it can sample the Lagrange coefficients as before: 
\begin{code}
    \cFor $j=1,\ldots,\cthr$ \cDo\\
        \> $\Lambda_{i_j}(X):=\prod_{\zeta\in\set{1,\ldots,\cthr}\setminus{j}}^{\cthr} \frac{i_\zeta-X}{i_\zeta-i_j}\cdot(\Delta)$
    \end{code}
Now, $\cB$ with knowledge of the keys for indices $i_1,\ldots,i_{\cthr-1}$ along with access to Oracle needs to simulate valid responses to $\tppeval$ queries for an unknown index. Call this index $i^\ast$. Then, we have:
\begin{align*}
    \tppeval(\tpk^{(i^\ast)},x):
    &=\hash(x)^{\Delta\cdot \tpk^{(i^\ast)}}=\hash(x)^{\sum_{j=1}^{\cthr} \Lambda_{i_j}(i^\ast)\cdot \tpk^{(i_j)}}\\
    &=\hash(x)^{\sum_{i=1}^{\cthr-1}\Lambda_{i_j}(i^\ast)\cdot \tpk^{(i_j)}}\cdot \left(\hash(x)^{\tpk^{(i_\cthr)}}\right)^{\Lambda_{i_\cthr}(i^\ast)}\\
\end{align*}
The last term is simulated using $\cB$'s own oracle access. 
\newcommand{\A}{\cA}
\begin{proof}
Let $\cA$ be a PPT attacker against the pseudorandomness property of $\tprf$, having advantage $\epsilon$.

$\A$ first chooses $\cthr-1$ indices $\bfK=\{i_1,\ldots,i_{\cthr-1}\}$ where each index is a subset of $\{1,\ldots,\csize\}$. $\cA$ receives the shares of the keys $\tpk^{(i_1)}=f(i_1),\ldots,\tpk^{(i_{\cthr-1}) }=f(i_{\cthr-1})$ (for unknown polynomial $f$ of degree $\cthr$ such that $f(0)=\tpk\cdot \Delta$ with $\Delta:=\Delta$. Further, $\A$ has access to $\oEval(i,x)$ receiving $\tppeval(\tpk_i,x)$ ins response. Additionally, $\A$ expects to have oracle access to the random oracle $\hash$.

Using this attacker $\A$, we now define a PPT attacker $\cB$ which will break the pseudorandomness property of Construction~\ref{cons:khprf}. Note that $\cB$ is given access to the oracle that either outputs the real evaluation of the PRF on key $\tpk^\ast$ or a random value. Additionally, $\cB$ expects to have oracle access to the random oracle $\hash$.
 
\begin{itemize}
\item \textbf{Setup:} $\cB$ does the following during Setup. 
\begin{itemize}
    \item Receive set $S=\{i_1,\ldots,i_{\cthr-1}\}$ from $\A$. 
    \item Next $\cB$ generates the key shares and public key as follows: 
    \begin{itemize}
        \item Sample $\tpk^{(i_1)},\ldots,\tpk^{(i_{\cthr-1})}\in \bbZ$.
        \item $\cB$ picks an index $i_\cthr$ at random and implicitly sets the PRF key chosen by its challenger as $\tpk^{(i_t)}$.
        \item Immediately, given the $\cthr$ indices, one can construct the secret sharing polynomial $f\in\bbZ[X]$ as described earlier, but instead recreating the polynomial $f'(X)$ using the coefficients $\Lambda_{i_j}(X)$ for $j=1,\ldots,\cthr$ with $\tpk^{(i_t)}$ being unknown to $\cB$ and using its challenger to simulate a response.
        \item $\cB$ gives $\tpk^{(i_1)},\ldots,\tpk^{(i_{\cthr-1})}$ to $\cA$.
    \end{itemize}
    \end{itemize}
\item \textbf{Queries to H:} 
$\cB$ merely responds to
all queries from $\cA$ to $\hash$ by using its oracle access to $\hash$.
\item \textbf{Queries to Partial Evaluation:} $\cB$ receives as query input, some choice of key index specified by $i^\ast$ and input $x_j$ for $i=1,\ldots,Q$.  In response $\cB$ does the following:
\begin{itemize}
    \item Forward $x_j$ to its challenger. In response it implicitly receives $\tppeval(\tpk^{(i_\cthr)},x_j)$, but off by a factor of $\Delta$ in the exponent. Call this $h_{j,\cthr}$.
    \item Compute: $h_{j,i^\ast}=\hash(x_j)^{\sum_{i=1}^{\cthr-1}\Lambda_{i_j}(i^\ast)\cdot \tpk^{(i_j)}}\cdot (h_{j,\cthr})^{\Lambda_{i_\cthr}(i^\ast)}$ where $\cB$ uses its own access to hash oracle to get $\hash(x_j)$.
    \item It returns $h_{j,i^\ast}$ to $\cA$. 
\end{itemize}
\item \textbf{Challenge Query:} On receiving the challenge input $x^\ast$, $\cB$ does the
following:
\begin{itemize}
    \item Ensure that it is a valid input, i.e., there is no partial evaluation queries on $x^\ast$ at any unknown index point. 
    \item If not, $\cB$ forwards to its challenger $x^\ast$. In response it implicitly receives $\tppeval(\tpk^{(i_\cthr)},x^\ast)$, but off by a factor of $\Delta$ in the exponent. Call this $h^\ast$. 
    \item It also uses its oracle access to $\hash$ to receive
    $h=\hash(x^\ast)$.
    \item It finally computes $y=\hash(x_j)^{\Delta^2\cdot \sum_{i=1}^{\cthr-1}\Lambda_{i_j}(0)\cdot \tpk^{(i_j)}}\cdot (h^\ast)^{\Delta^2\cdot \Lambda_{i_\cthr}(0)}$ and outputs
    $y$ to $\cA$
\end{itemize}
\item \textbf{Finish:} It forwards $\cA$'s guess as its own guess. 
\end{itemize}
\paragraph{Analysis of the Reduction.} Note that for the case when $b=0$, $\cA$ expects to receive $\hash(x^\ast)^{\Delta^3\cdot \tpk}$ where $\tpk$ is defined at the point 0. So, we get:
\begin{align*}
    \hash(x^\ast)^{\Delta^3\cdot \tpk^{(0)}}&=\hash(x^\ast)^{\Delta^2\cdot f'(0)}=\hash(x)^{\Delta^2\sum_{j=1}^{\cthr} \Lambda_{i_j}(0)\cdot \tpk^{(i_j)}}\\
    &=\hash(x)^{\Delta^2\sum_{i=1}^{\cthr-1}\Lambda_{i_j}(0)\cdot \tpk^{(i_j)}}\cdot \left(\hash(x)^{\tpk^{(i_\cthr)}}\right)^{\Delta^2\Lambda_{i_\cthr}(0)}\\
\end{align*}
This shows that the returned value $y$ is consistent when $b=0$. Meanwhile, when $b=1$, $h^\ast$ is a random element in the group and then $y$ is a truly random value which means that $\cB$ has produced a valid random output for $\cA$. Similarly, when $b=0$, every response to partial evaluation is also done consistently by correctness of the underlying secret sharing scheme. Meanwhile, when $b=1$, we can rely on the statistical privacy preserving guarantee of the underlying secret sharing scheme to argue that the difference that the adversary can notice is statistically negligible. This concludes the proof where $\cB$ can only succeed with advantage $\epsilon$. 
\end{proof}
\fi

%% file: byzantine.tex
\section{Heterogeneity and Poisoning Attacks}
\label{sec:byzantine}
Secure Aggregation was a useful tool to realize FedAvg~\cite{pmlr-v54-mcmahan17a}. However, research~\cite{fedopt} has shown that FedAvg does not yield good accuracy or model convergence when confronted with non-i.i.d client dataset distribution. Prior works such as DReS-FL~\cite{dresfl} often require expensive cryptographic techniques to be resilient to heterogeneous datasets. For example, DReS-FL requires that each client secretly share its datasets with other clients before having each client train. In this section, we show how to combine secure aggregation with FedOpt algorithm~\cite{fedopt}, extending existing secure aggregation techniques to privacy-preserving federated learning that can handle heterogeneity in the dataset. 

\subsection{FedOpt} 
\label{sub:fedopt} 
FedOpt~\cite{fedopt} is a family of algorithms that abstracts (and generalizes FedAvg). It allows for a choice of optimizer other than Stochastic Gradient Descent (SGD) on the client side and a more resilient update rule on the server side. Indeed, the work of Reddi~\etal~\cite{fedopt}, also presents instantiations of various server-side update rules. See Algorithm~\ref{fig:caps-opt} for pseudocode where the text is in black. Here $T$ is the number of iterations, and $x_t$ is the global model at $t$. At the start of every iteration, each client $i$ sets its model $x_{i,0}^t$ to be $x_t$. Meanwhile, $K$ is the number of local iterations the client performs, with $k$ being the iterating variable. $\term{ClientOpt}$ is the algorithm employed by the client based on its local learning rate $\eta_i$ to update the model $x_{i,k}^t$ to $x_{i,k+1}^t$ (\refmarker{8}).  $\Delta_i^t$ is the update between the global model at iteration $t$ ad the local model at the end of $K$ local iterations, at iteration $t$. The former is denoted by $x_t$ while the latter is $x_{i,K}^t$ (\refmarker{9}). Finally, \refmarker{13} shows the server side optimization $\term{ServerOpt}$ to update the current global model $x_t$ based on the computed aggregate of clients $\Delta_t$ along with global learning rate $\eta$. 

\paragraph{FedOpt and $\caps$.}  Observe that the input to $\term{ServerOpt}$ is independent of the individual client updates and instead only takes as parameter $\Delta_t$ (the \emph{average} of client updates $\Delta_i^t$), the current model $x_t$, learning parameter $\eta$, and iteration count $t$ (\refmarker{11}). Therefore, with Secure Aggregation, the server can compute $\Delta_t$, while preserving the honest client's updates $\Delta_i^t$, and later rely on a suitable $\term{ServerOpt}$ that is more resilient to heterogeneity. We concretely formalize the pseudocode in Algorithm~\ref{fig:caps-opt}, with the additional steps needed marked in blue. 
\begin{algorithm}[!tb]
\caption{FedOpt and $\caps$}
\label{fig:caps-opt}
\begin{algorithmic}[1]
 \linemarker{1}\State \textbf{Input:} $x_0$, \texttt{ClientOpt}, \texttt{ServerOpt}
 \linemarker{2}\For{$t = 0$ \textbf{to} $T - 1$}
 \linemarker{3}   \State Sample a subset $S$ of clients and subset $C$ of committee clients 
    \Comment{{We can use the approach of Flamingo to generate $C$, independent of server.}}
  \linemarker{4}   \State $x_{i,0}^t = x_t$
  \linemarker{5}   \For{each client $i \in S$ \textbf{in parallel}}
     \linemarker{6}    \For{$k = 0$ \textbf{to} $K - 1$}
     \linemarker{7}        \State Compute an unbiased estimate $g_{i,k}^t$ of $\nabla F_i(x_{i,k}^t)$
      \linemarker{8}       \State $x_{i,k+1}^t = \texttt{ClientOpt}(x_{i,k}^t, g_{i,k}^t, \eta_i, t)$
        \EndFor
      \linemarker{9}   \State $\Delta_i^t = x_{i,K}^t - x_t$
       \linemarker{10}  \State \textcolor{blue}{Use $\caps_\text{LWR}$ to send masked $\Delta_i^t$ to server and route the encrypted seed share to $C$.}
    \EndFor
    \linemarker{11} \State \textcolor{blue}{Reconstruct the seed sum and aggregate the masked updates.}
    \linemarker{12} \State \textcolor{blue}{Recover $\Delta_t = \frac{1}{|S|} \sum_{i \in S} \Delta_i^t$}
     \linemarker{13}\State $x_{t+1} = \texttt{ServerOpt}(x_t, -\Delta_t, \eta, t)$
\EndFor
\end{algorithmic}
\end{algorithm}

\subsection{Byzantine-Robust Stochastic Aggregation (bRSA)}
\label{sub:brsa}
bRSA \cite{brsa} is a class of stochastic sub-gradient methods for distributed learning resilient to Byzantine workers (i.e., clients sending arbitrary inputs). It mitigates the effects of incorrect messages due to poisoning behaviors, communication failures, or uneven data distribution by incorporating a regularization term in the objective function.
At each iteration $t$, clients compute parameter updates based on local data, prior local models, and global parameters.




At each iteration $k$, client $i$ computes parameter updates based on local data ($\xi_i^k$), prior local models ($x_i^k$), and global parameters ($w^k$). The client and server updates are:

{\small \[
\text{Client: } x_i^{k+1} = x_i^k - \eta^k \left( \nabla F(x_i^k, \xi_i^k) + \lambda\text{sign}(x_i^k - w^k) \right) 
\]
\[
\text{Server: } w^{k+1} = w^k - \eta^k \left( \nabla f_0(w^k) + \lambda \sum_{i\in\lbrack n\rbrack} \text{sign}(w^k - x_i^k) \right)
\]}
where $\eta$ is the learning rate, $\xi$ is a local dataset sample, $F(\cdot, \cdot)$ is the loss function, $f_{\ell_2}(\cdot)$ is the robust regularization term, $\lambda$ weights the robustness term, $\text{sign}$ is element-wise, and $\lbrack n\rbrack$ is the client set.

 However, unlike FedOpt, this approach aggregates not the model gradient but rather a function of the current worldwide model and the gradient update sent by the client. Specifically, the server computes $\text{sign}(w^k,x_i^k)$, which can be viewed as a function of how far away the client's gradient ($x_i^k$) is from the current global model ($w^k$). These are then added before proceeding with additional server-side optimization.



\paragraph{bRSA and $\caps$.} As pointed out by Franzese~\etal\cite{franzese2023robust}, the only information needed by the server to aggregate is $\text{sign}(\Delta_i^t-x_t)$. The clients, rather than providing $\Delta_i^t$, computes locally $\text{sign}(\Delta_i^t-x_t)$ and provides this as an input to the server. Note that this is simply a vector with elements in $\set{-1,1}$. Let this be vector $\bfu_i\in \set{-1,1}^{L}$ where $L$ is the size of the model. 

We can optimize further by sending a binary vector instead (say $\bfv_i$) with the property that $\bfv_i[j]=0$ iff $\bfu_i[j]=0$ for $j=1,\ldots,L$. Then, while the server requires $\sum \bfu_i$, this is equivalent to $2\cdot \sum_{i=1}^{n} \bfv_i-n$. Sending $\bfv_i$ lends itself to efficient zero-knowledge proof to show that a masked input is a binary vector. 


%% file: app-malicious.tex
\begin{figure*}[!tb]
\centering
\begin{protocolbox}{Malicious Security with Abort}
\begin{minipage}{\textwidth}
\begin{description}
    \item[Setup:] Let $\cH:\set{0,1}^\ast\to \bbF^{\csize-d-2}$ where $d=\crec-1$ be a hash function modeled as a random oracle. Let $\cH':\set{0,1}^\ast\to \bbF$ be the hash function used to generate the challenge. Let $G$ be a group generated by $g$ where the Discrete Logarithm and DDH problems are hard, and $G$ is of prime order $q$, the same as the order of the field used for Shamir Secret Sharing.
    
    \item[Client $i$:] Performs the following steps:
    \begin{description}
        \item[1.] Commit to $\seed_{i,\lab}^{(0)}=\seed_{i,\lab},\seed_{i,\lab}^{(1)},\ldots,\seed_{i,\lab}^{(\csize)}$ as $C_i^{(j)}:=g^{\seed_{i,\lab}^{(j)}}$.
        \item[2.] Generate the coefficients of the polynomial of degree $\csize-d-2$ using the Fiat-Shamir transform: $m_0,\ldots,m_{\csize-d-2}\gets \cH(C_i^{(0)},\ldots,C_i^{(\csize)})$.
        \item[3.] Compute $v_0,\ldots,v_{\csize}$ as $v_i:=\prod_{j\in\set{0,\ldots,\csize}\setminus i} (i-j)^{-1}$.
        \item[4.] Compute $\bfw:=(v_0\cdot m^\ast(0),\ldots,v_\csize\cdot m^\ast(\csize))$.
        \item[5.] Generate $\bft:=(t_0,\ldots,t_{\csize})\getsr\bbF$.
        \item[6.] Commit to $t_0,\ldots,t_\csize$ as $C_t^{(j)}:=g^{t_j}$.
        \item[7.] Compute $r:=\langle \bft,\bfw\rangle$.
        \item[8.] Compute $c:=\cH'(C_i^{(0)},\ldots,C_i^{(\csize)},C_t^{(0)},\ldots,C_t^{(\csize)},\bfw,r)$.
        \item[9.] Compute $z_0,\ldots,z_\csize$ where $z_i:=t_i+c \cdot \seed_{i,\lab}^{(i)}$.
        \item[10.] Set $\pi_i:=\left(\set{C_i^{(j)}},r,\bfz=(z_0,\ldots,z_\csize),c\right)$.
    \end{description}

    \item[Server:] Upon receiving $\pi_i$ from client $i$, performs the following:
    \begin{description}
        \item[1.] Parse $\pi_i:=(\set{C_i^{(j)}},r,\bfz=(z_0,\ldots,z_\csize),c)$.
        \item[2.] Compute $\bfw$ (as done by the client) and check if $\langle \bfw,\bfz\rangle = r$.
        \item[3.] For each $j=0,\ldots,\csize$, compute $C_t^{(j)}=g^{z_j}\cdot \left(C_i^{(j)}\right)^{-c}$.
        \item[4.] Compute $c'=\cH'(C_i^{(0)},\ldots,C_i^{(\csize)},C_t^{(0)},\ldots,C_t^{(\csize)},\bfw,r)$.
        \item[5.] Accept input from client $i$ if $c == c'$, else client $i$ is dropped.
        \item[6.] Send $C_i^{(j)}$ and the encrypted shares for committee member $j$ to committee member $j$.
    \end{description}

    \item[Committee Member $j$:] Upon receiving data from the server:
    \begin{description}
        \item[1.] Decrypt and recover the share $\seed_{i,\lab}^{(j)}$.
        \item[2.] Verify that the recovered share matches the commitment forwarded by the server.
        \item[3.] If verification fails, complain to the server.
    \end{description}
    \item[Server:] 
    \begin{description}
        \item[1.] If any complaint is received, protocol is aborted.
        \item[2.] Server checks if $\prod_{i\in\cC} C_i^{(0)}\eqq g^{\sum_{i\in\cC} \seed_{i,\lab}}$ where $\sum_{i\in\cC} \seed_{i,\lab}$ is obtained by reconstruction from the shares. 
    \end{description}
\end{description}
\end{minipage}
\end{protocolbox}
\caption{Malicious Security with Abort}
\label{fig:mal-clients}
\end{figure*}

%% file: stronger.tex
\ifdefined\IsSP 
\section{Extensions of \CAPS}
\label{sec:extensions}
\subsection{Security Against Committee Members}
\label{sec:stronger}
\else \section{Stronger Security Definition}
\label{sec:stronger}
\fi
Hitherto, we have only considered the indistinguishability of information from the perspective of the server. However, one can consider the requirement to hold for even corrupt committee members. Specifically, the client's input remains hidden if their entire committee collude (or at least $\cthr$ of them). It is easy to observe that Figure~\ref{fig:async} does not satisfy the stronger security definition. If we had a single committee member, the auxiliary information (available to the committee member) masks the input and, therefore, can be unmasked. 
%
Thus, to accommodate security against collusion of all committee members, we modify $\caps$ syntax and construction to include a key from the server to keep client privacy. Informally, we do the following:
\begin{itemize}
    \item First, the server or the aggregator also has a secret key, denoted by $\tpk_0$. 
    \item Second, for each label $\lab$, the server first publishes a ``public key'', as a function of the following algorithm $\psaaux_{0,\lab}\gets\psag(\tpk_0,\lab)$
    \item Third, the encrypt procedure takes into account this auxiliary information, i.e., 
    
    $\left(\psact_{i,\lab},\allowbreak \set{\psaaux_{i,\lab}^{(j)}}_{j\in[\csize]}\allowbreak \right)\allowbreak \getsr\allowbreak \psae(\pp,\psask_i,\psain_{i},\psaaux_{0,\lab},\cthr,\csize,\lab)$. In other words, the server publishes the public key, and then the client can encrypt it to a label. 
    \item Fourth, the decrypt procedure takes $\tpk_0$ as input too. 
\end{itemize} 
\ifdefined\IsSP 
\else The committee indistinguishability game proceeds in phases.
\begin{itemize}
    \item Setup Phase: The challenger begins by running the setup algorithm to generate the system parameters. The adversary is then provided with the system parameters $\pp$ and is asked to output an adversarial choice of $n$, which is the number of users that will be registered. In response, the challenger runs the $\psakg$ algorithm $n+1$ times, each for the $n$ users and once for the server's secret key. This phase ends with the adversary being provided with the server's secret key denoted by $\psask_0$. 
    \item Learning Phase: The adversary issues queries to the various oracles defined by $\oCorr,\oEnc$ to learn any information it could. $\oCorr$ proceeds where the adversary can corrupt any user and receive its key. These corruptions are tracked. Meanwhile, $\oEnc$ allows the adversary to issue any arbitrary encryption queries on behalf of any of the users, with the restriction that it can only do so once per user per label. In response, it receives both the ciphertext encrypting the input and \emph{all} the auxiliary information. 
    
    This phase ends with the adversary committing to a target label $\tau$. 
    \item Challenge Phase: In this phase, the challenger begins by identifying eligible users $\cU$ who are honest, which is defined by $[n]\setminus \kcorr$. Without loss of generality, we assume there have been no queries to $\oEnc$ with $\tau$ as the label. Should there be such queries, those users $i$ such that $(i,\tau,\cdot)\in\bfE$ are also removed from the set $\cU$, and these inputs are later used to compute the challenge. Upon receiving $\cU$, the adversary commits to two sets: $\cH\subseteq \cU$ is the set of honest users that the adversary is targeting, and $\cS$ that is the set of committee members for whom the adversary receives $\set{\psaaux_{i,\tau}^{(j)}}_{i\in\cH,j\in\cS}$ provided $|\cS|\leq \cthr-1$. Further, the adversary also provides inputs two choices of inputs for user in $\cH$ denoted by $\set{\psain_{i,0},\psain_{i,1}}_{i\in \cH}$ and inputs $\set{\psain_{i}}_{i\in [n]\setminus \cH}$ for the remaining users.
    \item Finally, the adversary is provided with individual encryptions and auxiliary information for all committee members.  
    \item Guessing Phase: The adversary outputs a guess $b'$ and wins if $b'=b$, provided trivial attacks do not happen.
\end{itemize}
\begin{definition}[C-IND-CPA Security]
    We say that a $(\cthr,\csize,\modulus)$ \CAPS~Scheme $\allowbreak\caps$ with label space $\cL$ is Server-Indistinguishable under Chosen Plaintext Attack (S-IND-CPA) if for any PPT adversary $\cA$, there exists a negligible function $\negl$ such that: 
    \begin{gather*}
	\Pr\left[
	\begin{array}{c|c}
	& \pp\getsr\pres(1^\kappa);b\getsr\{0,1\}\\
	 & (\state,n)\getsr\cA(\pp),\set{\psask_i\getsr\psakg()}_{i\in[n]\cup \{0\}}\\
	b=b'&(\state,\tau)\getsr\cA^{\oCorr,\oEnc}(\state,\psask_0), \cU:=[n]\setminus \kcorr\\ 
    & (\cH,\cS,\set{\psain_{i,0},\psain_{i,1}}_{i\in\cH},\set{\psain_{i}}_{i\in [n]\setminus \cH})\getsr\cA(\state,\cU)\\
    & \set{\prect_{i,\tau},\set{\psaaux_{i,\tau}^{(j)}}_{j\in[\csize]}\getsr\psae(\pp,\presk_i,\psain_{i,b})}_{i\in\cH}\\ 
    & \set{\prect_{i,\tau},\set{\psaaux_{i,\tau}^{(j)}}_{j\in[\csize]}\getsr\psae(\pp,\presk_i,\psain_{i})}_{i\in[n]\setminus\cH}\\ 
	& b'\getsr\cA(\state,\set{\psact_{i,\tau},\psaaux_{i,\tau}^{(j)}}_{i\in[n],j\in[\csize]})\\
	\end{array}
	\right] \leq \frac{1}{2}+\negl(\kappa)
    \end{gather*}
\end{definition}
\subsection{Updated Committee Indistinguishable Construction }
\label{sub:stronger}
These are the changes to $\caps$ construction based on the \HSM\ assumption to adapt it to the stronger security definition:

\begin{itemize}
    \item  \algoHead{$\psag(\tpk_0,\lab)$}
			    \begin{algorithmic}
                \State Compute $\psapk_{0,\lab}\gets\tprf.\tpeval(\tpk_0,\lab)$
                \State \Return $\psapk_{0,\lab}$
			    \end{algorithmic}
    \item Modify the encryption procedure as follows: 
    \item[]\algoHead{$\psae(\pp,\psask_i,\psain_{i},\psapk_{0,\lab},\cthr,\csize,\lab)$}
			    \begin{algorithmic}
                    \State Parse $\psask_i=\tpk_i$
                    \State Compute $h_{i,\lab}=\tprf.\tpeval(\tpk_i,\lab)$
                    \State Compute $\psact_{i,\lab}=f^{\psain_{i}}\cdot \psapk_{0,\lab}^{\psask_i}$
                    \State Compute $(\tpk_i^{(j)})_{j\in[\csize]}\getsr\tprf.\tpshare(\tpk_i,\cthr,\csize)$
                    \For{$j=1,\ldots,\csize$}
                    \State $\psaaux_{i,\lab}^{(j)}=\tprf.\tpeval(\tpk_i^{(j)},\lab)$
                    \EndFor
			        \State \Return $\psact_{i,\lab},\set{\psaaux_{i,\lab}^{(j)}}_{j\in[\csize]}$
			    \end{algorithmic}
    In other words, $\psaaux_{i,\lab}^{(j)}$ can be viewed as the $j$-th partial evaluation of the key $(\tpk_i)$ where $\tpk_i$ was key shared using Secret Sharing scheme, while $\psact_{i,\lab}$ was masked by the $\tprf$ evaluation on the key $\psask_i\cdot \psask_0$
\end{itemize}
Now, the committee member multiplies all the auxiliary information. As a result, $\AUX_\lab^{(j)}$ is simply a partial evaluation of the following key share $\sum_{i=1}^{n} \tpk_i^{(j)})$. Therefore, the server computes $\tprf.\tpeval(\cdot (\sum_{i=1}^{n} \tpk_i)$. Now, let us look at the decryption procedure:

\algoHead{$\psaa(\AUX_\lab,\tpk_0\set{\prect_{i,\lab}}_{i\in\cC})$}
			    \begin{algorithmic}
                    \State Compute $M=\AUX_\lab^{-\psask_0}\cdot (\prod_{i\in\cC} \psact_{i,\lab})$
                    \State Compute $X_\lab\gets\cl\Solve(\pp,M)$
                    \State \Return $X_\lab\bmod \modulus$
			    \end{algorithmic}


The security of this construction follows from the intuition that the adversary gets all of the auxiliary information, from which it can only construct a Diffie-Hellman key on the fly, from which it cannot compute any masking information to unmask the inputs. 

\fi